\documentclass[english]{article}
\usepackage[T1]{fontenc}
\usepackage[latin9]{inputenc}
\usepackage{verbatim}
\usepackage{amsthm}
\usepackage{amsmath}
\usepackage{amssymb}
\usepackage{esint}
\usepackage{varwidth}
\usepackage{url}
\usepackage{constants}
\usepackage{graphicx}
\usepackage{subfig}
\usepackage{float}
\usepackage{fancyhdr}
\usepackage{eufrak}
\usepackage{algorithm}
\usepackage{adjustbox}
\usepackage[dvipsnames]{xcolor}

\makeatletter
\theoremstyle{plain}
\newtheorem{thm}{\protect\theoremname}
  \theoremstyle{plain}
  \newtheorem{lem}[thm]{\protect\lemmaname}
  \theoremstyle{plain}
  
 \theoremstyle{plain}
  \newtheorem{rem}{Remark}
 \theoremstyle{plain}
  \newtheorem{prop}[thm]{\protect\propositionname}

\usepackage{hyperref}
\def \SU {\mathcal{S}}
\def \A {\mathcal{A}}
\def \B {\mathcal{B}}
\def \alg {\operatorname{SNIPE}}
\def \Q {\mathcal{Q}}
\def \E {\mathbb{E}}
\def \h {\widehat}
\def \U {\mathcal{U}}
\def \grouse {\operatorname{GROUSE}}

\def \R {\mathbb{R}}
\def \len {l} 
\def \Int {I} 

\def \GR {\mathbb{G}}

\def \h {\widehat}
\def \wt{\widetilde}

\def \Q {\mathcal{Q}}
\def \A {\mathcal{A}}
\def \B {\mathcal{B}}
\def \reject {\mathcal{A}}

\def \I {\kappa} 

\hypersetup{
   colorlinks=true,%
   citecolor=black,%
   filecolor=black,%
   linkcolor=black,%
   urlcolor=blue
}

\def \l{\left}
\def \r{\right}
\def \fail{\phi}

\def \ev{\mathfrak{E}}
\def \o{\h{\SU}_{k-1}}
\def \n{\h{\SU}_{k}}

\def \Span{\operatorname{span}}

\makeatother

\usepackage{babel}
  \providecommand{\corollaryname}{Corollary}
  \providecommand{\lemmaname}{Lemma}
\providecommand{\theoremname}{Theorem}
\providecommand{\propositionname}{Proposition}

\providecommand{\keywords}[1]{\textbf{\textit{Keywords---}} #1}

\begin{document}

\title{Streaming Principal Component Analysis From Incomplete Data
}
\author{Armin Eftekhari, Gregory Ongie, Laura Balzano, and Michael B. Wakin\footnote{AE is with the Alan Turing Institute in London. GO and LB are with the Department of Electrical Engineering and Computer Science at the University of Michigan, Ann Arbor. MBW is with the Department of  Electrical Engineering at the Colorado School of Mines. (E-mails: aeftekhari@turing.ac.uk; gongie@umich.edu; girasole@umich.edu; mwakin@mines.edu)} }

\maketitle
\begin{abstract}

Linear subspace models are pervasive in computational sciences  and particularly used for large datasets which are often incomplete due  to  privacy issues or sampling constraints. Therefore, a critical problem is developing an efficient algorithm  for detecting low-dimensional linear structure from incomplete data efficiently, in terms of both computational complexity and storage. 

In this paper we propose a streaming subspace estimation algorithm called Subspace Navigation via Interpolation from Partial Entries ($\alg$) that efficiently processes blocks of incomplete data to estimate the underlying subspace model. In every iteration, $\alg$ finds the subspace that best fits the new data block but remains close to the previous estimate. We show that $\alg$ is a streaming solver for the underlying nonconvex matrix completion problem, that it converges globally {to a stationary point of this program} regardless of  initialization, and that the convergence is locally linear with high probability. We also find that  $\alg$ shows state-of-the-art performance in our numerical simulations.


\end{abstract}


\section{Introduction}
\label{sec:problem statement}

Linear models are the backbone of computational science, and {Principal Component Analysis} (PCA) in particular is an indispensable tool for detecting linear structure in collected data \cite{van2012subspace,ardekani1999activation,krim1996two,tong1998multichannel}. 
Principal components of a dataset are used, for example, to perform linear dimensionality reduction, which is in turn at the heart of classification, regression and other learning tasks that often suffer from the ``curse of dimensionality'', where having a small number of training samples in relation to the number of features typically leads to overfitting \cite{hastie2013elements}.

In this work, we are particularly interested in applying PCA to data that is presented sequentially to a user, with limited processing time available for each item. Moreover, due to hardware limitations, we assume the user can only store small amounts of data. Finally, we also consider the possibility that the incoming data is incomplete, {either due to physical sampling constraints, or deliberately subsampled to facilitate faster processing times or to address privacy concerns.}

As one example, consider monitoring network traffic over time, where acquiring complete network measurements at fine time-scales is impractical and subsampling is necessary \cite{lakhina2004diagnosing,gershenfeld2010intelligent}. 
As another example, suppose we have a network of cheap, battery-powered sensors that must relay summary statistics of their measurements, say  their principal components, to a central node on a daily basis. Each sensor cannot store or process all its daily measurements locally, nor does it have the power to relay all the raw data to the central node. Moreover, many measurements are not reliable and can be treated as missing. 
It is in this {and similar contexts} that we hope to develop a \emph{streaming} algorithm for PCA from incomplete data.

{More formally, we consider the following problem:} Let $\SU$ be an $r$-dimensional subspace with orthonormal basis $S\in\mathbb{R}^{n\times r}$.
For an  integer $T$, let the coefficient vectors $\{q_t\}_{t=1}^T\subset \mathbb{R}^r$ be independent copies of a random vector $q\in\mathbb{R}^r$ with bounded expectation, namely $\E \|q\|_2 <\infty$. Consider the sequence of vectors $\{Sq_t\}_{t=1}^T\subset \SU$ and set $s_t := Sq_t$ for short.  At each time $t\in[1:T]:=\{1,2,\cdots,T\}$,  we observe each entry of  $s_t $ independently with a probability of $p$ and  collect the observed entries in $y_t\in\mathbb{R}^n$. Formally, we let $\omega_t\subseteq [1:n]$ be the random index set over which $s_t$ is observed and write this measurement process as $y_t=P_{\omega_t}\cdot s_t$, where  $P_{\omega_t} \in\mathbb{R}^{n\times n}$ is the projection onto the coordinate set $\omega_t$, namely it equals one on its diagonal entries corresponding to the index set $\omega_t$ and is zero elsewhere. 
\begin{center}
\begin{figure}[H]
\begin{center}
\subfloat[\label{fig:variable prob}]{\protect\includegraphics[width=0.3\textwidth]{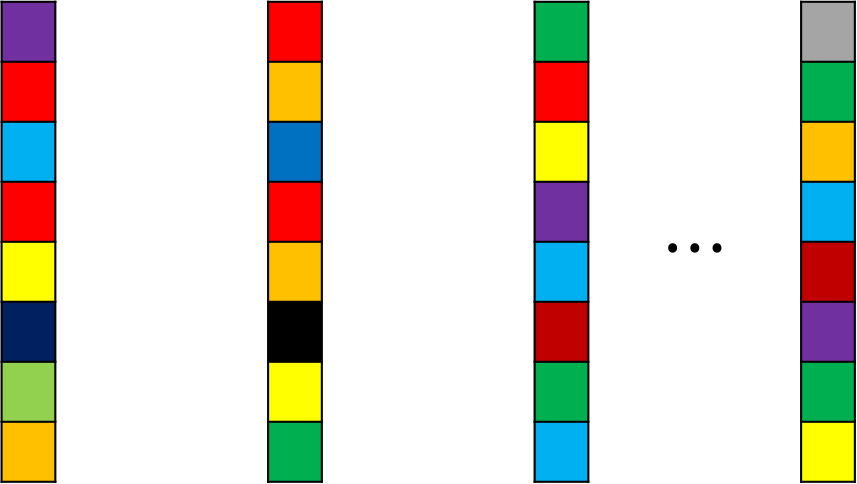}
}

\subfloat[\label{fig:variable rank}]{\protect\includegraphics[width=0.3\textwidth]{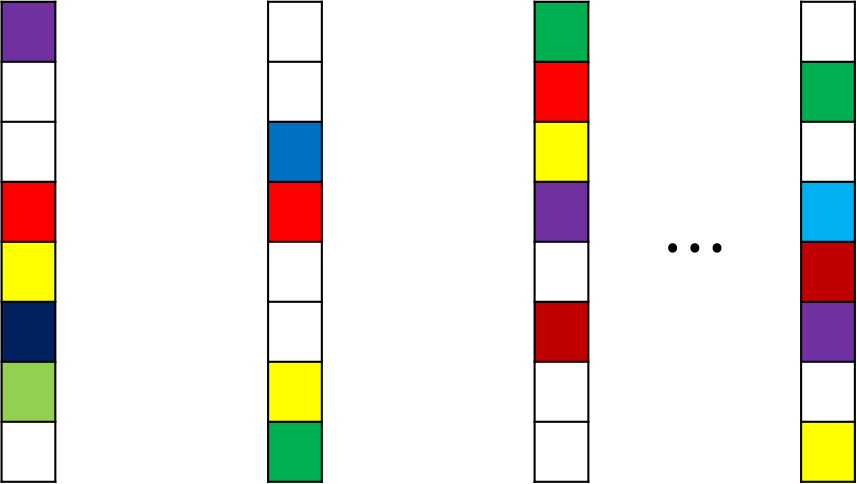}
}
\caption{\label{fig:streamVis} The sequence of generic vectors $\{s_t\}_{t=1}^T$ drawn from an unknown $r$-dimensional subspace $\SU$ in panel~(a) is only partially observed on random index sets $\{\omega_t\}_{t=1}^T$.  That is, we only have access to incomplete data vectors $\{y_t\}_{t=1}^T$ in panel (b), where the white entries are missing. Our objective is to estimate the  subspace $\SU$ from the incomplete data vectors, when limited storage and processing resources are available. See Section \ref{sec:problem statement} for the detailed setup.  }
\end{center}
\end{figure}
\end{center}

Our objective in this paper is to design a streaming algorithm to identify the subspace $\SU$ from the incomplete data $\{y_t\}_{t=1}^T$ supported on the index sets $\{\omega_t\}_{t=1}^T$.  Put differently, our objective is to design a streaming algorithm to compute  leading $r$ principal components of the full (but hidden) data matrix $[s_1 \, s_2 \cdots s_T]$ from the incomplete observations $[y_1 \, y_2 \, \cdots y_T]$, see Figure \ref{fig:streamVis}. By the  Eckart-Young-Mirsky Theorem, this task is equivalent to computing  leading $r$ left singular vectors of the full data matrix from its partial observations \cite{eckart,mirsky}. 

Assuming that $r=\mbox{dim}(\SU)$ is known \emph{a priori} (or estimated from data by other means), we present the $\alg$ algorithm for this task in Section \ref{sec:proposed alg}. $\alg$ is designed based on the \emph{principle of   least-change} and, in every iteration,  finds the subspace that best fits the new data block but remains close to the previous estimate.  $\alg$ requires $O(n)$ bits of memory and performs $O(n)$ flops in every iteration, which is optimal in its dependence on the ambient dimension $n$.  As discussed in Section \ref{sec:interp}, $\alg$ has a natural interpretation as a streaming algorithm for low-rank matrix completion \cite{davenport2016overview}. 

Section \ref{sec:theory} discusses the global and local convergence of $\alg$. 
In particular, the local convergence rate is linear near the true subspace, namely the estimation error of $\alg$ reduces by a factor of $1-cp$ in every iteration, for a certain factor $c$ and with high probability. This local convergence guarantee for $\alg$ is a key technical contribution of this paper which is absent in its close competitors, see Section \ref{sec:related work}. 


Even though we limit ourselves to the ``noise-free'' case of $y_t = P_{\omega_t} s_t$ in this paper, $\alg$ can also be applied (after minimal changes)  when $y_t = P_{\omega_t}(s_t+n_t)$, where we might think of $n_t\in\R^{n}$ as measurement noise  from a signal processing viewpoint. Alternatively from a statistical viewpoint, $n_t$ represents the tail of the covariance matrix of the generative model, from which $\{s_t\}_t$ are drawn. Moreover, $\alg$ can be easily adapted to the \emph{dynamic} case where the underlying subspace $\SU=\SU(t)$ changes over time. We  leave the convergence analysis of $\alg$ under a noisy time-variant model to a future work. Similarly, entries of incoming vectors are observed uniformly at random in our theoretical analysis but $\alg$ also applies to any incomplete data.

A review of prior art is presented in Section~\ref{sec:related work}, and the performance of $\alg$  and rival algorithms are examined numerically in Section \ref{sec:simulations}, where we find that $\alg$  shows the state-of-the-art performance. Technical proofs appear in Section \ref{sec:analysis} and in the appendices, with Appendix \ref{sec:Toolbox} (Toolbox) collecting some of the frequently-used mathematical tools. Finally, Appendix \ref{sec:alt init} offers an alternative initialization for $\alg$.

\section{SNIPE}
\label{sec:proposed alg}

In this section, we present Subspace Navigation  via Interpolation from Partial Entries (SNIPE), a streaming algorithm for subspace identification from incomplete data, received sequentially.

Let us first introduce some additional notation. 
Recall that we denote the incoming sequence of incomplete vectors by $\{y_t\}_{t=1}^T\subset\mathbb{R}^n$, which are supported on index sets $\{\omega_t\}_{t=1}^T\subseteq [1:n]$.
For a block size $b\ge r$, we concatenate every $b$ consecutive vectors into a data block, thereby partitioning the incoming data into $K=T/b$ non-overlapping blocks $\{Y_k\}_{k=1}^K$, where $Y_k\in\mathbb{R}^{n\times b}$ for every $k$.  We assume for convenience that $K$ is an integer. We also often take $b=O(r)$ to maximize the efficiency of SNIPE, as discussed below.

 At a high level, $\alg$ processes the first incomplete block  $Y_1$ to produce an estimate $\h{\SU}_1$ of the true subspace $\SU$. This estimate is then iteratively updated  after receiving each of the new incomplete blocks $\{Y_k\}_{k= 2}^K$,
 thereby producing a sequence of estimates $\{\h{\SU}_k\}_{k= 2}^K$ , see Figure \ref{fig:Stream2}. Every $\h{\SU}_k$ is an $r$-dimensional subspace of $\mathbb{R}^n$ with orthonormal basis $\h{S}_k\in\mathbb{R}^{n\times r}$; the particular choice of orthonormal basis is inconsequential throughout the paper. 
 
\begin{center}
\begin{figure}[H]
\begin{center}
\includegraphics[width=0.3\textwidth]{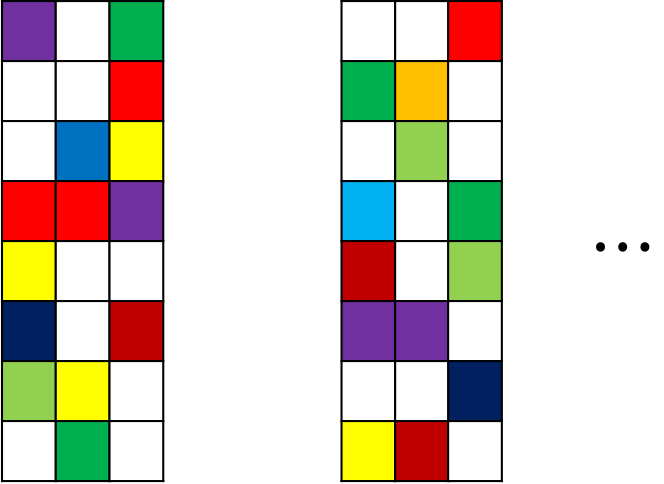}
\caption{\label{fig:Stream2} 
$\alg$ concatenates every $b$ incoming vectors into a block and iteratively updates its estimate of the true subspace $\SU$ after receiving each new block. That is, $\alg$ updates its estimate of $\SU$, from $\h{\SU}_{k-1}$ to $\h{\SU}_{k}$, after receiving the incomplete data block $Y_k\in\R^{n\times b}$, see Section \ref{sec:proposed alg} for the details.    }
\end{center}
\end{figure}
\end{center}

More concretely, $\alg$ sets $\h{\SU}_1$ to be the span of leading $r$ left singular vectors of $Y_1$, namely the left singular vectors corresponding to largest $r$ singular values of $Y_1$, with ties broken arbitrarily. Then, at iteration $k\in[2:K]$ and given the previous estimate $\h{\SU}_{k-1}=\mbox{span}( \h{S}_{k-1})$,   $\alg$ processes the columns of the $k$th incomplete block $Y_k$ one by one and forms the matrix
\begin{equation}
R_k =
\l[
\begin{array}{ccc}
\cdots
&
y_t + P_{\omega_t^C} \h{S}_{k-1}\l( \h{S}_{k-1}^* P_{\omega_t} \h{S}_{k-1} + \lambda I_r  \r)^{\dagger} \h{S}_{k-1}^*  y_t
&
\cdots
\end{array}
\r]
\in\mathbb{R}^{n\times b},
\qquad  t\in [(k-1)b+1:  kb],
\label{eq:Rk}
\end{equation}
where $\dagger$ denotes the pseudo-inverse and $\lambda\ge 0$ is a parameter. Above, $P_{\omega_t^C}=I_n-P_{\omega_t}\in\mathbb{R}^{n\times n}$ projects a vector onto the complement of the index set $\omega_t$. {The motivation for the particular choice of $R_k$ above will become clear in Section \ref{sec:interp}.}
$\alg$ then updates its  estimate by setting $\h{\SU}_k$ to be the span of leading $r$ left singular vectors of $R_k$. Algorithm \ref{alg:Alg} summarizes these steps. Note that Algorithm \ref{alg:Alg} rejects ill-conditioned updates in Step 3 for the convenience of analysis and that similar reject options have precedence in the literature \cite{balzano2015local}. We however found implementing this reject option to be unnecessary in numerical simulations.


\begin{rem} \label{rem:complexity} \emph{\textbf{[Computational complexity of SNIPE]}} \emph{We measure the algorithmic complexity of $\alg$ by calculating the average number of floating-point operations (flops) performed on  an incoming vector. } 
\emph{Every iteration of $\alg$ involves  finding leading $r$ left singular vectors of  an $n\times b$ matrix. Assuming that $b=O(r)$, this could be done with $O(nr^2)$ flops.
At the $k$th iteration  with $k\ge 2$, $\alg$ also requires finding the pseudo-inverse of $P_{\omega_j}\h{S}_{k-1}\in\mathbb{R}^{n\times r}$ for each incoming vector which costs $O(nr^2)$ flops. Therefore the overall computational complexity of $\alg$ is $O(nr^2)$ flops per  vector. As further discussed in Section \ref{sec:related work}, this matches the complexity of other algorithms for streaming PCA even though here the received data is highly incomplete.}\hfill\qedsymbol
\end{rem}

\begin{rem}\label{rem:storage}\emph{\textbf{[Storage requirements of SNIPE]}}
\emph{We measure the storage required by $\alg$ by calculating the number of memory elements stored by $\alg$ at any given instant.  
At the $k$th iteration, $\alg$ must store the current estimate $\h{S}_{k-1}\in\mathbb{R}^{n\times r}$ (if available) and the new incomplete block $Y_k\in\mathbb{R}^{n\times b}$. Assuming that $b=O(r)$, this translates into $O(nr)+O(pnr)=O(nr)$ memory elements. $\alg$ therefore requires $O(nr)$ bits of storage, which is optimal up to a constant factor. }\hfill\qedsymbol
\end{rem}

\begin{center}
\begin{algorithm}
\caption{$\alg$ for streaming  PCA from incomplete data \label{alg:Alg}}
\vspace{5pt}

\textbf{Input:}
\begin{itemize}

\item Dimension $r$,
\item Received data $\{y_t\}_{t=1}^T\subset \mathbb{R}^n$ supported on index sets $\{\omega_t\}_{t=1}^T\subseteq [1:n]$, presented sequentially in $K=T/b$ blocks of size $b \ge r$,
\item Tuning parameter $\lambda\ge 0$, 
\item {Reject thresholds $\sigma_{\min},\tau\ge 0$.}
\end{itemize}
\vspace{5pt}

\textbf{Output: } 
\begin{itemize}
\item $r$-dimensional subspace $\widehat{\SU}_K$.
\end{itemize}
\vspace{5pt}

\textbf{Body:} 
\begin{itemize}

\item Form $Y_1\in\mathbb{R}^{n\times b}$ by concatenating the first $b$  received vectors $\{y_t\}_{t=1}^{b}$. Let $\h{\SU}_1$, with orthonormal basis $\h{S}_1\in\mathbb{R}^{n\times r} $,  be the span of  leading $r$ left singular vectors of $Y_1$, namely those corresponding to $r$ largest singular values. Ties are broken arbitrarily.
\item For $k \in [2:K]$, repeat:
\begin{enumerate}
\item Set $R_{k}\leftarrow \{\}$.
\item For $t\in [(k-1)b+1: kb]$, repeat
\begin{itemize}
\item  Set $$R_{k}\leftarrow \l[
\begin{array}{cc}
R_{k}
&
y_t + P_{\omega_t^C} \h{S}_{k-1}\l( \h{S}_{k-1}^* P_{\omega_t} \h{S}_{k-1} + \lambda I_r  \r)^{\dagger} \h{S}_{k-1}^* y_t 
\end{array}
\r],
$$
where  $P_{\omega_t} \in\mathbb{R}^{n\times n}$ equals one on its  diagonal entries corresponding to the index set $\omega_t$, and is zero elsewhere. Likewise, $P_{\omega_t^C}$ projects a vector onto the complement of the index set $\omega_t$.
\end{itemize}
\item If $\sigma_r(R_k)< \sigma_{\min}$ or $\sigma_{r}(R_k) \le (1+\tau) \cdot \sigma_{r+1}(R_{k})$, then set $\h{\SU}_k\leftarrow\h{\SU}_{k-1}$.  Otherwise, let $\h{\SU}_{k}$, with orthonormal basis $\h{S}_k\in\mathbb{R}^{n\times r} $, be the span of  leading $r$ left singular vectors of $R_k$. Ties are broken arbitrarily. Here, $\sigma_i(R_k)$ is the $i$th largest singular value of $R_k$. 

\end{enumerate}
\item Return $\h{\SU}_K$.
\end{itemize}
\end{algorithm}
\end{center}
\section{Interpretation of SNIPE}
\label{sec:interp}

$\alg$ has a natural interpretation as a streaming  algorithm for low-rank matrix completion, which we now discuss. First let us enrich our notation. Recall the incomplete data blocks $\{Y_k\}_{k=1}^K \subset \R^{n\times b}$ and let the random index set $\Omega_k\subseteq [1:n]\times [1:b]$ be the support of $Y_k$ for every $k$. We write that $Y_k=P_{\Omega_k}(S_k)$   for every $k$, where the complete (but hidden) data block $S_k\in\R^{n\times b}$ is formed by concatenating $\{s_t\}_{t=(k-1)b+1}^{kb}$. Here,  $P_{\Omega_k}(S_k)$  retains only the entries of $S_k$ on the index set $\Omega_k$, setting the rest to zero. By design, $s_t = S q_t$ for every $t$ and we may therefore write that $S_k=S\cdot Q_k$  
for the coefficient matrix $Q_k\in\R^{r\times b}$ formed by concatenating $\{q_t\}_{t=(k-1)b+1}^{kb}$. 
%
To summarize, $\{S_k\}_{k=1}^K$   is formed by partitioning $\{s_t\}_{t=1}^T$   into $K$ blocks. Likewise,  $\{Q_k,Y_k,\Omega_k\}_{k=1}^K$ are formed by partitioning  $\{q_t,y_t,\omega_t\}_{t=1}^T$, respectively.

 With this introduction, let us form $Y\in\R^{n\times T}$ by concatenating the incomplete blocks $\{Y_k\}_{k=1}^K \subset \R^{n\times b}$, 
supported on  the index sets $\Omega\subseteq [1:n]\times [1:T]$. 
To find the  true subspace $\SU$, one might consider solving
\begin{align}
\begin{cases}
\underset{X,\U}{\min}\,\,\,
 \l\|P_{\U^\perp}X  \r\|_F^2+ \lambda \|P_{\Omega^C}(X) \|_F^2 \\
P_{\Omega}\l( X\r) = Y,
\end{cases}
\label{eq:non cvx}
\end{align}
where the minimization is over a matrix $X\in\mathbb{R}^{n\times T}$ and $r$-dimensional subspace $\U\subset\mathbb{R}^n$.
Above, $P_{\U^\perp}\in\mathbb{R}^{n\times n}$ is the orthogonal projection onto the orthogonal complement of subspace $\U$ and $P_\Omega(X)$ retains only the entries of $X$ on the index set $\Omega$, setting the rest to zero.  Note that Program \eqref{eq:non cvx} encourages its solution(s) to be low-rank while matching the observations $Y$ on the index set $\Omega$.  The term $\lambda\| P_{\Omega^C}(X) \|_F^2$  for $\lambda\ge 0$  is the \emph{Tikhonov regularizer} that controls the energy of solution(s) on the complement of index set $\Omega$.

With complete data, namely when $\Omega=[1:n]\times [1:T]$, Program \eqref{eq:non cvx} reduces to PCA, as it returns $X=Y$ and searches for an $r$-dimensional subspace that captures most of the energy of $Y$. That is, Program~\eqref{eq:non cvx} reduces to $\min_{\U} \|P_{\U^\perp}Y\|_F^2$ when $\Omega = [1:n]\times [1:T]$, solution of which is the span of leading $r$  left singular vectors of $Y$ in light of  the   Eckart-Young-Mirsky Theorem \cite{eckart,mirsky}. In this sense then, Program \eqref{eq:non cvx} performs PCA from incomplete data. Note crucially that Program \eqref{eq:non cvx} is a nonconvex problem because the Grassmannian $\mathbb{G}(n,r)$, the set of all $r$-dimensional subspaces in $\mathbb{R}^n$, is a nonconvex set.\footnote{The Grassmannian can be embedded in  $\mathbb{R}^{n\times n}$ via the map that takes $\U\in\mathbb{G}(n,r)$ to the corresponding orthogonal projection $P_{\U}\in\R^{n\times n}$. The resulting submanifold of $\mathbb{R}^{n\times n}$ is a nonconvex set.} However, given a fixed subspace $\U\in\GR(n,r)$, Program \eqref{eq:non cvx} reduces to the simple least-squares program   
\begin{equation}
\begin{cases}
\underset{X}{\min}\,\,\,
 \l\|P_{\U^\perp}X  \r\|_F^2 + \lambda\| P_{\Omega^C}(X) \|_F^2 \\
P_{\Omega}\l( X\r) = Y,
\end{cases}
\label{eq:fixed U}
\end{equation}
where the minimization is over $X\in\R^{n\times T}$. If in addition $\lambda$ is positive, then Program \eqref{eq:fixed U} is strongly convex and has a unique minimizer. 
Given a  fixed feasible $X\in\R^{n\times T}$, Program \eqref{eq:non cvx} has the same minimizers as 
\begin{equation}
\underset{\U \in \GR(n,r)}{\min} \|P_{\U^\perp} X\|_F^2.
\label{eq:fixed X}
\end{equation}
That is, for a fixed feasible $X$, Program \eqref{eq:non cvx} simply performs PCA on $X$. 
We might also view Program \eqref{eq:non cvx} from a matrix completion perspective. More specifically, let 
\begin{equation}
\rho_r^2(X) = \sum_{i\ge r+1} \sigma_i^2(X),
\end{equation}
be the \emph{residual} of $X$, namely the energy of its trailing singular values $\sigma_{r+1}(X)\ge \sigma_{r+2}(X)\ge \cdots $. Like the popular nuclear norm $\|X\|_* = \sum_{i\ge 1}\sigma_i(X)$ in \cite{davenport2016overview}, the residual $\rho_r(X)$ gauges the rank of $X$. In particular, $\rho_r(X)=0$ if and only if {$\mbox{rank}(X) \leq r$}. Unlike the nuclear norm, however, the residual is still a nonconvex function of $X$. 
We now rewrite Program \eqref{eq:non cvx} as 
\begin{align}
\begin{cases}
\underset{X,\U}{\min}\,\,\,
 \l\|P_{\U^\perp}X  \r\|_F^2 + \lambda\| P_{\Omega^C}(X) \|_F^2 \\
P_{\Omega}\l( X\r) = Y
\end{cases}
& = 
\begin{cases}
\underset{X}{\min}\,\,\, \underset{\U\in \GR(n,r)}{\min}\,\,\,
 \l\|P_{\U^\perp}X  \r\|_F^2 + \lambda\| P_{\Omega^C}(X) \|_F^2\\
P_{\Omega}\l( X\r) = Y
\end{cases} \nonumber\\
& = 
\begin{cases}
\underset{X}{\min}\,\,\, \rho_r^2(X) + \lambda\| P_{\Omega^C}(X) \|_F^2\\
P_{\Omega}\l( X\r) = Y.
\end{cases}
\label{eq:residual}
\end{align}
That is, if we ignore the regularization term $\lambda \|P_{\Omega^C}(X)\|_F^2$, Program \eqref{eq:non cvx} searches for a matrix with the least residual, as a proxy for least rank, that matches the observations $Y$.  In this sense then, Program~\eqref{eq:non cvx} is a ``relaxation'' of the low-rank matrix completion problem. Several other formulations for the matrix completion problem are reviewed in \cite{davenport2016overview,eftekhari2018weighted,eftekhari2016mc}. We can also rewrite Program \eqref{eq:non cvx} in terms of its data blocks by considering the equivalent program
\begin{align}
\begin{cases}
\min\,\,\, \sum_{k=1}^K  \|P_{\U}^\perp X_k  \|_F^2 + \lambda\| P_{\Omega_k^C}(X_k) \|_F^2\\
P_{\Omega_k}\l(X_k\r) = Y_k \qquad  k\in[1:K],
\end{cases} 
\label{eq:non cvx parallel pre}
\end{align}
where the minimization is over  matrices $\{X_k\}_{k=1}^K \subset \R^{n\times b}$ and subspace $\U \in \mathbb{G}(n,r)$. Let us additionally introduce a number of auxiliary variables into Program \eqref{eq:non cvx parallel pre} by considering the equivalent program 
\begin{align}
\begin{cases}
\min\,\,\, \sum_{k=1}^K  \|P_{\U_k}^\perp X_k  \|_F^2+ \lambda\| P_{\Omega_k^C}(X_k) \|_F^2\\
P_{\Omega_k}\l(X_k\r) = Y_k \qquad  k\in[1:K]\\
\U_1 = \U_2 = \cdots = \U_K,
\end{cases}
\label{eq:non cvx parallel}
\end{align}
where the minimization is over matrices $\{X_k\}_{k=1}^K \subset \R^{n\times b}$ and subspaces $\{\U_k\}_{k=1}^K \subset \mathbb{G}(n,r)$. Indeed, Programs (\ref{eq:non cvx},\ref{eq:non cvx parallel pre},\ref{eq:non cvx parallel}) are all equivalent and all nonconvex. Now consider the following approximate solver for Program \eqref{eq:non cvx parallel} that alternatively solves for matrices and subspaces:
\begin{itemize}
\item Setting $X_1=Y_1$ in Program \eqref{eq:non cvx parallel},   we minimize $ \|P_{\U_1}^\perp Y_1 \|_F^2
$ over $\U_1\in\mathbb{G}(n,r)$ and, by the Eckart-Young-Mirsky Theorem, find a minimizer to be the span of  leading $r$ left singular vectors of $Y_1$, which coincides with $\h{\SU}_1$ in $\alg$.
\item For $k\in[2:K]$, repeat:
\begin{itemize}
\item Setting $\U_k=\h{\SU}_{k-1}$ in Program \eqref{eq:non cvx parallel}, we  solve
\begin{align}
\begin{cases}
\min_{X_k} \,\,\, \|P_{\h{\SU}_{k-1}^\perp}X_k \|_F^2+\lambda\| P_{\Omega_k^C}(X_k) \|_F^2 \\
P_{\Omega_k}\l(X_k\r) = Y_k,
\end{cases}
\label{eq:Rk interp}
\end{align}
over matrix $X_k\in\mathbb{R}^{n\times b}$. We verify  in Appendix \ref{sec:interp verifying} that the minimizer of Program \eqref{eq:Rk interp} coincides with $R_k$  in $\alg$, see \eqref{eq:Rk}. 
\item If $\sigma_r(R_k)<\sigma_{\min}$ or $\sigma_r(R_k)\le (1+\tau) \sigma_{r+1}(R_k)$, then no update is made, namely we set $\h{\SU}_k = \h{\SU}_{k-1}$. Otherwise, setting $X_k=R_k$ in Program \eqref{eq:non cvx parallel},   we solve
$ \min \|P_{\U_k}^\perp R_k \|_F^2
$ over $\U_k\in\mathbb{G}(n,r)$ to find  $\h{\SU}_k$. That is, by the Eckart-Young-Mirsky Theorem again, $\h{\SU}_k$ is the span of leading $r$ left singular vectors of $R_k$. The output of this step matches $\h{\SU}_k$ produced in $\alg$.
\end{itemize}
\end{itemize}
To summarize, following the above procedure   produces $\{R_k\}_{k=1}^K$ and $\{\h{\SU}_k\}_{k=1}^K$  in $\alg$. 
In other words, we might think of $\alg$ as an approximate solver for Program \eqref{eq:non cvx}, namely $\alg$ is a streaming algorithm for low-rank matrix completion.  In fact, the output of $\alg$ always converges to a stationary point of Program \eqref{eq:non cvx} in the sense described in Section \ref{sec:theory}.

Another insight about the choice of $R_k$ in \eqref{eq:Rk} is as follows. Let us set $\lambda=0$ for simplicity. At the beginning of the $k$th iteration of $\alg$ with $k\ge 2$, the available estimate of the true subspace is $\h{\SU}_{k-1}$ with orthonormal basis $\h{S}_{k-1}$. Given a new incomplete vector $y\in\mathbb{R}^n$, supported on the index set $\omega\subseteq [1:n]$, $z= \h{S}_{k-1} (P_{\omega} \h{S}_{k-1})^\dagger y$ best approximates $y$ in $\h{\SU}_{k-1}$ in  $\ell_2$ sense.
In order to agree with the measurements, we minimally adjust this to $y+P_{\omega^C}z$, where $P_{\omega^C}$ projects onto the complement of index set $\omega$. This indeed matches the expression for the columns of $R_k$ in $\alg$. We note that this type of least-change  strategy has been successfully used in the development of quasi-Newton methods for optimization \cite[Chapter 6]{nocedal2006numerical}.

\section{Performance of SNIPE}
\label{sec:theory}

To measure the performance of $\alg${---whose output is a {subspace}---}we naturally use principal angles as an error metric. More specifically, recall that  $\SU$ and $\h{\SU}_K$  denote the true subspace and the output of $\alg$, respectively. Then the $i$th largest singular value of $P_{\SU^{\perp}} P_{\h{\SU}_{K}}$ equals $\sin(\theta_i(\SU,\h{\SU}_K))$, where 
$$
\theta_1(\SU,\h{\SU}_K)\ge \theta_2(\SU,\h{\SU}_K) \ge \cdots \ge \theta_r(\SU,\h{\SU}_K)
$$
denote the principal angles between the two $r$-dimensional subspaces $\SU,\h{\SU}_K$ \cite{golub2013matrix}.
The estimation error of $\alg$ is then
\begin{equation}
d_{\mathbb{G}}(\SU,\h{\SU}_K ) := \sqrt{
\frac{1}{r}
\sum_{i=1}^r \sin^2\l( \theta_i (\SU,\h{\SU}_K )
\r)}
=
\frac{\| P_{\SU^\perp}P_{\h{\SU}_K} \|_F}{\sqrt{r}},
\label{eq:err metric}
\end{equation}
which also naturally induces a metric topology on the Grassmannian $\GR(n,r)$.\footnote{Another possible  error metric is simply the largest principal angle $\theta_1(\SU,\h{\SU}_K)$. The two metrics are very closely related: $\theta_1(\SU,\h{\SU}_K)/\sqrt{r} \le d_{\GR}(\SU,\h{\SU}_K) \le \theta_1(\SU,\h{\SU}_K)$. However,  we find that $\theta_1(\SU,\h{\SU}_K)$ is not amenable for analysis of our problem, as opposed to $d_{\GR}(\SU,\h{\SU}_K)$.  } Note also that we will always reserve calligraphic letters for subspaces and capital letters for their orthonormal bases, for example subspace $\SU$ and its orthonormal basis $S$.

Our first result loosely speaking states that {a subsequence of} $\alg$ converges to a stationary point of  the nonconvex Program \eqref{eq:non cvx}  as $T$ goes to infinity, see Section \ref{sec:cvg to crit pnt} for the proof.  
\begin{thm}\label{prop:cvg to critic point} \textbf{\emph{[Global convergence]}} 
Consider an $r$-dimensional subspace $\SU$ with orthonormal basis $S\in\mathbb{R}^{n\times r}$.  For an  integer $T$, let the coefficient vectors $\{q_t\}_{t=1}^T\subset \mathbb{R}^r$  be independent copies of a random vector $q\in\mathbb{R}^r$ with bounded expectation, namely $\mathbb{E}\|q\|_2<\infty$.
For every $t\in[1:T]$,  we observe each coordinate of  $s_t = S q_t \in \SU$ independently with a probability of $p$ and collect the observations in $y_t\in\mathbb{R}^n$, supported on a random index set $\omega_t\subseteq [1:n]$. Fix positive $\lambda$, block size $b\ge r$, positive reject thresholds $\sigma_{\min},\tau$, and consider the  output sequence of $\alg$ in Algorithm \ref{alg:Alg}, namely $\{(R_k,\h{\SU}_k)\}_k$.  Also by partitioning $\{q_t,s_t,y_t,\omega_t\}_t$, form the coefficient blocks $\{Q_k\}_k\subset \R^{r\times b}$, data blocks $\{S_k\}_k\subset \R^{n\times b}$, and incomplete data blocks $\{Y_k\}_k \subset \R^{n\times b}$ supported on index sets $\{\Omega_k\}_k \subseteq [1:n]\times [1:b]$, as described in Section \ref{sec:interp}.

For every integer $l$, there exists an integer $k_l$, for which the following {asymptotic statement  is almost surely true as $T\rightarrow\infty$.} 
Consider  the restriction of Program (\ref{eq:non cvx}) to iteration $k_l$, namely 
\begin{equation}
\label{eq:restricted}
\begin{cases}
\min_{X,\U}\,\,\,  \| P_{\U^\perp} X\|_F^2 + \lambda \| P_{\Omega_{k_l}^C}(X) \|_F^2\\
P_{\Omega_{k_l}}(X) = Y_{k_l},
\end{cases}
\end{equation}
where the minimization is over matrix $X\in\mathbb{R}^{n\times b}$ and $r$-dimensional subspace $\U$. 
 Then there exists $R\in\R^{n\times b}$ and $r$-dimensional subspace $\h{\SU}$ such that
\begin{itemize}
\item $\h{S}$ is the span of leading $r$ left singular vectors of $R$,
\item   $(R,\h{\SU})$ is a stationary pair of Program (\ref{eq:restricted}), namely it satisfies the first-order optimality conditions of Program (\ref{eq:restricted}), 
\item  $\lim_{l\rightarrow\infty} \| R_{k_l} - R \|_F = 0,$
\item $\lim_{l\rightarrow\infty} d_{\GR} (\h{\SU}_{k_l} ,\SU) = 0.$
\end{itemize} 
\end{thm} 
\begin{rem}
\label{rem:disc of thm 1}
\emph{\textbf{[Discussion of Theorem \ref{prop:cvg to critic point}]} Theorem \ref{prop:cvg to critic point}  roughly speaking states that there is a subsequence of $\alg$ that converges to a stationary point of Program \eqref{eq:non cvx}, which was the program designed  in Section \ref{sec:interp} for  PCA from incomplete data or, from a different perspective, for low-rank matrix completion. 
Theorem \ref{prop:cvg to critic point} is however silent about the nature of this stationary point, whether it is a local or global minimizer/maximizer, or a saddle point. To some extent, this question is addressed below in Proposition \ref{prop:if lim exists}. }

\emph{More generally, we have been able to show that this stationary point is in fact rank-$r$. When the block size of $\alg$ is sufficiently large, namely when $b=\Omega(n)$, we can further establish that the  limit point of $\alg$ is  indeed a global minimizer of Program \eqref{eq:non cvx} and moreover $\alg$ recovers the true subspace, namely $\lim_{l \rightarrow\infty} d(\h{\SU}_{k_l},\SU)=0$, with high probability and under certain standard conditions on the  \emph{coherence} of the true subspace $\SU$ and sampling probability $p$. We have not included these results here because $\alg$ is  intended as a streaming algorithm and we are therefore more interested in the setting where  $b=O(r)$, see Remarks \ref{rem:complexity}~and~\ref{rem:storage} about the implementation of $\alg$. It is not currently clear to us when $\alg$ converges in general but, as suggested by Proposition \ref{prop:if lim exists} below, if $\alg$ converges, it does indeed converge to the true subspace $\SU$.   }
\hfill\qedsymbol
\end{rem}
\begin{rem}
\emph{\textbf{[Technical point about Theorem \ref{prop:cvg to critic point}]} 
Note that Theorem \ref{prop:cvg to critic point} is proved for positive (but possibly arbitrarily small)  $\sigma_{\min},\tau$. In particular, an update $(R_k,\h{\SU}_k)$ is rejected if 
\begin{equation}
\frac{\sigma_r(R_k)}{\sigma_{r+1}(R_{k+1})} \le 1+\tau, 
\end{equation}
for an (otherwise arbitrary) positive $\tau$ and whenever the ratio is well-defined. Here, $\sigma_i(R_k)$ is  the $i$th largest singular value of $R_k$. This is merely a technical nuance to avoid the output subspace $\h{\SU}_k$ from oscillating in the limit. Likewise, Theorem \ref{prop:cvg to critic point} does not address the case $\lambda=0$, even though $\lambda$ can be made arbitrarily small in Theorem \ref{prop:cvg to critic point}, see Program \eqref{eq:restricted}. This is again for technical convenience and in fact the numerical simulations in Section \ref{sec:simulations} are all conducted with $\lambda=0$. } \hfill\qedsymbol
\end{rem}

Our second result establishes that, if $\alg$ converges, then it converges to the true subspace $\SU$, see Section \ref{sec:proof of if converges} for the proof. 

\begin{prop}
\label{prop:if lim exists}
\emph{\textbf{[Convergence]}} Consider the setup in the first paragraph of Theorem \ref{prop:cvg to critic point}. Suppose that $r$ independent copies of random coefficient vector $q\in\mathbb{R}^r$ almost surely form a basis for $\mathbb{R}^r$.\footnote{For example, this requirement is met if entries of $q$ are independent Gaussian random variables with zero-mean and unit variance.} Suppose also that the output of $\alg$ converges to an $r$-dimensional subspace $\h{\SU}$, namely 
\begin{equation}
\lim_{k\rightarrow\infty} d_{\GR}(\h{\SU}_k,\h{\SU}) = 0. 
\label{eq:limExists}
\end{equation}
Then almost surely it must hold that $\h{\SU} = \SU$. 
\end{prop}

\begin{rem}
\emph{\textbf{[Discussion of Proposition \ref{prop:if lim exists}]}
Proposition \ref{prop:if lim exists} does not specify the conditions under which $\alg$ converges. Indeed, if the sampling probability $p$ is too small, namely if very few of the entries of incoming vectors are observed, then $\alg$ might  not converge at all as the numerical evidence suggests, see also Remark \ref{rem:disc of thm 1}. However, if $\alg$ converges, then it converges to the true subspace $\SU$. The local rate of convergence is specified below. 
 } \hfill\qedsymbol
\end{rem}

The concept of \emph{coherence} is  critical in  specifying the local convergence rate, since we consider entrywise subsampling. The coherence of an $r$-dimensional subspace $\SU$ with orthonormal basis $S\in\mathbb{R}^{n\times r}$ is defined as
\begin{equation}
\eta\l( \SU\r) := 
\frac{n}{r} \max \l\|S[i,:] \r\|_2^2,
\label{eq:def of coh}
\end{equation}
where $S[i,:]$ is the $i$th row of $S$. It is easy to verify that $\eta(\SU)$ is independent of the choice of orthonormal basis $S$ and that 
\begin{equation}
1 \le \eta(\SU)\le \frac{n}{r}. \label{eq:coh is bounded}
\end{equation} 
It is also common to say that $\SU$ is \emph{coherent} (\emph{incoherent}) when $\eta(\SU)$ is large (small). Loosely speaking, when $\SU$ is coherent, its orthonormal basis $S$ is ``spiky.'' An example is when $\SU$ is the span of a column-subset of the identity matrix. In contrast, when $\SU$ is incoherent, entries of $S$ tend to be ``diffuse.'' Not surprisingly, identifying a coherent subspace from subsampled data  may require many more samples  \cite{balzano2015local,mitliagkas2014streaming,chen2015incoherence}.

We will  also use $\lesssim$ and $\gtrsim$ below to suppress (most of) the  universal constants. 
 Moreover, throughout $C$ represents a universal constant, the value of which is subject to change in every appearance. 

Our next results specify the local convergence rate of $\alg$. Indeed, 
the convergence speed near the true subspace $\SU$ is linear as detailed in Theorems \ref{thm:local cvg expectation} and \ref{thm:main result}, and proved in Section \ref{sec:Refinement}. 
In particular, Theorem \ref{thm:local cvg expectation} states that, when sufficiently small, the expected estimation error of $\alg$ reduces by a factor of $1-p/32$ in every iteration. 
\begin{thm}\textbf{\emph{[Locally linear convergence of $\alg$ in expectation]}} \label{thm:local cvg expectation}
Consider the setup in the first paragraph of Theorem~\ref{prop:cvg to critic point}. Fix a positive tuning parameter $\alpha$, iteration $k\in [2:K]$, and let $\ev_{k-1}$ be the event where
 \begin{equation}
\frac{1}{\sqrt{nb}}\gtrsim  p\gtrsim  \log n \log^2 \l( p\sqrt{r} d(\SU,\h{\SU}_{k-1})  \r)    \frac{\eta(\SU)r}{n},
\label{eq:sampling p exp}
\end{equation}
\begin{equation}
d_{\GR}(\SU,\h{\SU}_{k-1})  
\log\l(\frac{16}{ p\sqrt{r}  d_{\GR}(\SU,\h{\SU}_{k-1}) } \r) 
\lesssim  
\frac{p^{\frac{3}{2}}}{    \sqrt{\log n}},
\label{eq:activate exp thm}
\end{equation} 
\begin{equation}
\|Q_k\| \le \frac{\sigma_{\min}}{\sqrt{1-p/4}},
\label{eq:well-cnd exp}
\end{equation}
where $\sigma_{\min}$ is the reject threshold in $\alg$. 
Let also $\reject_k$ be the event where the $k$th iteration of $\alg$ is not rejected (see Step 3 of Algorithm \ref{alg:Alg}) and let $1_{\reject_k}$ be the indicator  for this event, taking one if the event happens and zero otherwise. 
Then it holds that 
\begin{equation}
\E\l[ 1_{\reject_k} \cdot  d_{\GR}(\SU,\h{\SU}_{k}) \, |\, \h{\SU}_{k-1},\ev_{k-1} \r] \le  \l( 1- \frac{p}{32}\r) d_{\GR}(\SU,\h{\SU}_{k-1}).
\label{eq:exp cvg rate thm}
\end{equation}
\end{thm}

\begin{rem}\label{rem:discussion of local cvg exp}\emph{\textbf{[Discussion of Theorem \ref{thm:local cvg expectation}]} When the sampling probability $p$ is large enough and $\alg$ is near the true subspace $\SU$, Theorem \ref{thm:local cvg expectation}  states that the expected estimation error of $\alg$ reduces by a factor of $1-p/32$, if the iterate of $\alg$ is not rejected. Note that \raisebox{.5pt}{\textcircled{\raisebox{-.9pt} {1}}} The lower bound on the sampling probability $p$ in \eqref{eq:sampling p exp} matches the one in the low-rank matrix completion literature up  to a logarithmic factor \cite{davenport2016overview}. Indeed,   $\alg$ can be interpreted as a streaming matrix completion algorithm as discussed in Section~\ref{sec:interp}. 
 The upper bound on $p$ in \eqref{eq:sampling p exp} is merely for technical convenience and a tidier presentation in the most interesting regime for $p$. Indeed, since we often take $b=O(r)\ll n$, one might loosely read \eqref{eq:sampling p exp} as 
\begin{equation}
\frac{1}{\sqrt{nr}} \gtrsim p \gtrsim \frac{\eta(\SU)r}{n},
\label{eq:loosely reads}
\end{equation}
in which the upper bound  hardly poses a restriction even for moderately large data dimension $n$, as it forces $r =O(n^{\frac{1}{3}})$. \raisebox{.5pt}{\textcircled{\raisebox{-.9pt} {2}}} Ignoring the logarithmic factors for simplicity, we may read  \eqref{eq:activate exp thm} as $d_{\GR}(\SU,\h{\SU}_{k-1}) \lesssim p^{3/2}$, which ``activates'' \eqref{eq:exp cvg rate thm}. In other words, the basin of attraction of the true subspace $\SU$ as a (possibly local) minimizer of the (nonconvex) Program (\ref{eq:non cvx}) has a {radius of} $O(p^{3/2})$.  
\raisebox{.5pt}
{\textcircled{\raisebox{-.9pt} {3}}}
The indicator $1_{\reject_k}$ in \eqref{eq:exp cvg rate thm} removes the rejected iterates and similar conditions implicitly exist in the analysis of other streaming PCA algorithms~\cite{balzano2015local}. 
} 
\hfill\qedsymbol
\end{rem}

Note that Theorem \ref{thm:local cvg expectation} \emph{cannot} tell us what the local convergence rate of $\alg$ is, even in expectation. Indeed, the expected reduction in the estimation error of $\alg$, specified in \eqref{eq:exp cvg rate thm}, is not enough to activate \eqref{eq:activate exp thm} for  the next iteration (namely, with $k$ instead of $k-1$). That is, we cannot apply Theorem \ref{thm:local cvg expectation} iteratively and find the expected convergence rate of $\alg$. A key technical contribution of this paper is specifying the local behaviour of $\alg$ below. 
With high probability, the estimation error does not \emph{increase} by much in every iteration near  the true subspace. However, only in some of these iterations  does the error reduce.
Overall, on a long enough interval, the estimation error of $\alg$ near the  true subspace indeed reduces substantially and with high probability as detailed in Theorem \ref{thm:main result} and proved in Section \ref{sec:Refinement}. Performance guarantees for stochastic algorithms on long intervals is not uncommon, see for example \cite{cartis2017global}.

\begin{thm}\label{thm:main result} \textbf{\emph{[Locally linear convergence of $\alg$]}}
Consider the setup in the first paragraph of Theorem \ref{prop:cvg to critic point}. Suppose that the output $\h{\SU}_{K_0}$ of $\alg$ at iteration $K_0\in [2:K]$ satisfies 
\begin{equation}
d_{\GR}(\SU,\h{\SU}_{K_0})\lesssim  
 \frac{e^{-\frac{Cp^3 n b}{\widetilde{\eta}}} p^{\frac{7}{2}} nb}{  \wt{\eta}  \log b  \log\l( C\wt{\eta}n \r) \log^2(K-K_0)   },
 \label{eq:close enough thm}
\end{equation}
and that 
\begin{align}
K-K_0 \gtrsim \frac{ \wt{\eta} \log b \log(K-K_0)}{p^2 nb }.
\label{eq:large enough interval}
\end{align}
Then it holds that 
\begin{align}
\prod_{k=K_0+1}^{K} 1_{\reject_k} \cdot 
 d(\SU,\h{\SU}_K) 
 \lesssim    \l( 1- \frac{Cp^3 nb}{\wt{\eta}  \log b\log(K-K_0)}   \r)^{K-K_0}  d(\SU,\h{\SU}_{K_0}) , 
 \label{eq:cvg rate thm whp}
 \end{align}
except with a probability of at most
\begin{equation}
b^{-C\log(K-K_0)} + \sum_{k=K_0+1}^K \Pr\l[\|Q_k\|>\l( 1+ \frac{Cp^3nb}{\wt{\eta}\log b}\r)\sigma_{\min} \r]
\label{eq:fail pr main thm}
\end{equation}
and provided that 
\begin{equation}
 \frac{1}{\sqrt{nb}} \gtrsim p \gtrsim \log^2 b \log n \frac{\eta(\SU) r }{n}.
 \label{eq:sampling probability whp}  
\end{equation}
Above, $\sigma_{\min}$ is the reject threshold of $\alg$ and 
\begin{equation}
\wt{\eta} = \max_{k\in [K_0:K]} \wt{\eta}_k,
\end{equation}
\begin{equation}
\wt{\eta}_{k} 
:= nb\cdot \frac{\|P_{\h{\SU}_{k-1}^\perp} S_k\|^2_{\infty}}{\|P_{\h{\SU}_{k-1}^\perp}S_k\|_F^2}.
\label{eq:eta tilda}
\end{equation}

\end{thm}

\begin{rem}
\emph{\textbf{[Discussion of Theorem \ref{thm:main result}]} 
Loosely speaking, Theorem \ref{thm:main result} states that with high probability the estimation error of $\alg$ over $O(\wt{\eta}n)$ iterations reduces linearly (i.e., exponentially fast) when $\alg$ is near the true subspace and the sampling probability $p$ is large enough. 
Most of the remarks about Theorem~\ref{thm:local cvg expectation} are also valid here. 
Let us also point out that the dependence on the coefficient matrix $Q_k$ in \eqref{eq:fail pr main thm} is mild but necessary. 
As an example, consider the case where the coefficient vectors $\{q_t\}_t$ are standard random Gaussian vectors so that the coefficient matrices $\{Q_k\}_k$ are standard random Gaussian matrices, namely populated with independent zero-mean Gaussian random variables with unit variance. Then by taking 
$$
\sigma_{\min} = {C\sqrt{b}}/\l(1+ \frac{Cp^3nb}{\wt{\eta}\log b}\r),
$$
we find that 
\begin{equation}
\Pr\l[ \|Q_k\| > \l(1+\frac{Cp^3 nb}{\widetilde{\eta}\log b} \r) \sigma_{\min}\r]  = \Pr\l[\|Q_k\| \gtrsim \sqrt{b} \r] \le e^{-Cb}
\end{equation}
and consequently the failure probability in \eqref{eq:fail pr main thm} becomes $b^{-C\log(K-K_0)}+(K-K_0)e^{-Cb}$, which can be made arbitrarily small by modestly increasing the block size $b$.
For the reader's convenience, Appendix \ref{sec:cnd coh of Gaussian} collects the relevant spectral properties of a standard random Gaussian matrix. 
The dependence on $\|Q_k\|$ in Theorem~\ref{thm:local cvg expectation} is \emph{not} an artifact of our proof techniques. Indeed, when  $\|Q_k\|\gg 1$, it is likely that certain directions in $\SU$ are ``over represented'' which will skew the estimate of $\alg$ in their favor.
}
 \hfill\qedsymbol
\end{rem}

\begin{rem}
\emph{\textbf{[Coherence factor $\wt{\eta}$]} A key quantity in Theorem \ref{thm:main result} is the new ``coherence'' factor $\wt{\eta}$ which is absent in the expected behavior of $\alg$ in Theorem \ref{thm:local cvg expectation}. Somewhat similar to the coherence $\eta(\cdot)$ in \eqref{eq:def of coh}, $\wt{\eta}_k$ measures how ``spiky'' the matrix $ P_{\SU^\perp_{k-1}} S_k\in\R^{n\times b}$ is. In fact, one may easily verify that 
\begin{equation}
\wt{\eta}_k \le \mbox{rank}(P_{\SU^\perp_{k-1}} S_k) \cdot \nu(P_{\SU^\perp_{k-1}} S_k)^2 \cdot \eta(\mbox{span}(P_{\SU^\perp_{k-1}} S_k) ) \cdot \eta(\mbox{span}( S_k^* P_{\SU^\perp_{k-1}})),
\end{equation}
where $\nu(P_{\SU^\perp_{k-1}}S_k)$ is the condition number of $P_{\SU^\perp_{k-1}}S_k$, namely the ratio of largest to smallest nonzero singular values. The number of iterations needed to see a reduction  in estimation error in \eqref{eq:large enough interval} and the convergence rate of $\alg$ in \eqref{eq:cvg rate thm whp} both prefer  $\wt{\eta}$ to be small, namely prefer that $\{P_{\SU_{k-1}^\perp}S_k\}_k$ are all incoherent as measured by $\wt{\eta}_k$.}

\emph{When $\alg$ is close enough to true subspace $\SU$ as required in \eqref{eq:close enough thm}, one would expect that iterates of $\alg$ would be  nearly as coherent as $\SU$ itself in the sense that $\eta(\h{\SU}_k)\approx \eta(\SU)$. This intuition is indeed correct and also utilized in our analysis. However, even when $d_{\GR}(\SU,\h{\SU}_{k-1})$ is small and $\eta(\h{\SU}_{k-1}),\eta(\SU)$ are both small, $\wt{\eta}_{k}$ might be very large, namely $P_{\h{\SU}_{k-1}^\perp}S_k$ might be a  spiky matrix. Indeed, when $b=r$, $(\h{S}_{k-1}^\perp)^*S_k$ is approximately the (horizontal) tangent at $\h{\SU}_{k-1}$ to the geodesic on the Grassmannian $\GR(n,r)$ that connects $\h{\SU}_{k-1}$ to $\SU$. Even though both $\h{\SU}_{k-1}$ and $\SU$ are incoherent subspaces, namely $\eta(\h{\SU}_{k-1}),\eta(\SU)$ are both small, 
the tangent direction connecting the two is not necessarily incoherent. {Despite the dependence of our results on $\wt{\eta}$,} it is entirely possible that $\alg$ with high probability approaches the true subspace $\SU$ from an incoherent tangent direction, of which there are many. Such a result has remained beyond our reach. In fact, similar questions arise in matrix completion.   Iterative Hard Thresholding (IHT) is a powerful algorithm for matrix completion with excellent empirical performance \cite{tanner2013normalized}, the convergence rate of which has remained unknown for many years. With $\{M_k\}_k$ denoting the iterates of IHT, it is not difficult to see that if the differences $\{M_k-M_{k-1}\}_k$ are incoherent matrices (i.e., not spiky), then the linear convergence rate of IHT follows from rather standard arguments. 
\hfill\qedsymbol}
\end{rem}

\section{Related Work}
\label{sec:related work}

In this paper, we presented $\alg$ for streaming PCA from incomplete data and, from a different perspective, $\alg$ might be considered as a streaming matrix completion algorithm, see Section \ref{sec:interp}. In other words, $\alg$ is a ``subspace tracking'' algorithm that identifies the linear structure of data as it arrives. Note also that $t$ in our framework need not correspond to time, see Figure \ref{fig:streamVis}. For example, only a small portion of a large data matrix $Y$ can be stored in the fast access memory of the processing unit, which could instead use $\alg$ to fetch and process the data in small chunks and 
iteratively update the principal components. Moreover, $\alg$ can be easily adapted to the dynamic case where the distribution of data changes over time. In dynamic subspace tracking for example, the (hidden) data vector $s_t$ is drawn from the subspace $\SU(t)\in\GR(n,r)$ that varies with time. Likewise, it is easy to slightly modify  $\alg$ to handle noisy observations or equivalently to the case where $\{s_t\}_t$ are generated from a distribution with possibly full-rank covariance matrix. We leave investigating both of these directions to a future work. 


Among several algorithms that have been proposed for tracking low-dimensional structure in a dataset from partially observed streaming data~\cite{mitliagkas2014streaming,chi2013petrels,mardani2015subspace,xie2013change,eftekhari2016expect}, $\alg$ might be most closely related to $\grouse$ \cite{balzano2010online,balzano2013grouse}. $\grouse$ performs streaming PCA from incomplete data using stochastic gradient projection on the Grassmannian, updating  its estimate of the true subspace with each new incomplete vector. Both $\grouse$ and $\alg$ were designed based on the principle of least-change, discussed in Section~\ref{sec:interp}. In fact, when $\grouse$ is sufficiently close to the true subspace and with a specific choice of its step length, both algorithms have nearly identical updates, see \cite[Equation 1.9]{balzano2015local}. A weaker analogue of Theorem \ref{thm:local cvg expectation} for GROUSE was recently established  in \cite{balzano2015local}. More specifically, \cite{balzano2015local} stipulates that, if the current estimate $\widehat{\SU}_k$ is sufficiently close to the true subspace $\SU$, then $\widehat{\SU}_{k+1}$ will be even closer to $\SU$  in expectation. Such a result however \emph{cannot} tell us what the local convergence rate of  $\alg$ is, even in expectation, see the discussion right before Theorem \ref{thm:main result} above. In this sense, a key technical contribution of our work is establishing the local linear convergence of $\alg$, see Theorem \ref{thm:main result}, which is missing from its close competitor $\grouse$. In fact, the global convergence guarantees listed in Theorem \ref{prop:cvg to critic point} and Proposition \ref{prop:if lim exists} are also unique to $\alg$; such theoretical guarantees are not available for $\grouse$.

It might be interesting to add that our proposed update in $\alg$ was inspired by that of $\grouse$ when we found zero-filled updates were unreliable \cite{eftekhari2016expect}. However, $\grouse$ was derived as a purely streaming algorithm, and it therefore is not designed to leverage common low-rank structure that may be revealed when a block of vectors is processed at once. Therefore, for each block $\alg$ often achieves a more significant reduction in error than is possible with $\grouse$. %

Lastly, both $\alg$ and $\grouse$ have a computational complexity of $O(nr^2)$ flops 
 per incoming vector, see Remark \ref{rem:complexity}. Also, $\alg$ and $\grouse$ both require $O(nr)$ memory elements of storage, see Remark~\ref{rem:storage}.  
 With complete data, namely when no entries are missing, a close relative of both $\alg$ and $\grouse$ are incremental SVD algorithms, a class of  algorithms that efficiently compute the SVD of streaming data  \cite{bunch1978updating,balsubramani2013fast,oja1985stochastic,watanabe1973subspace,balsubramani2013fast,brand2002incremental}.

%

A streaming PCA algorithm might also be interpreted as a stochastic algorithm for PCA \cite{arora2012stochastic}. Stochastic projected gradient ascent in this context is closely related to the classical power method. In particular, the algorithm in \cite{mitliagkas2014streaming} extends the power method to  handle missing data, in part by improving the main result of \cite{lounici2014high}. With high probability, this algorithm converges globally and linearly to the true subspace and, most notably, succeeds for arbitrarily small sampling probability $p$, if the scope of the algorithm $T$ is large enough. Additionally, this algorithm too has a  computational complexity of   $O(nr^2)$ operations per vector and a storage requirement of $O(nr)$ memory elements. In practice, $\alg$ substantially outperforms the power method, as we will see in  Section \ref{sec:simulations}. A disadvantage of the power method is that it updates its estimate of the true subspace with every $O(n)$ incoming vectors; the waiting time might be prohibitively long if $n$ is large. In contrast, $\alg$ frequently updates its estimate with every $b=O(r)$ incoming vectors. As we will see in Section \ref{sec:simulations}, $\alg$ substantially outperforms the power method in practice. Let us add that POPCA \cite{gonen2016subspace} is closely related to the power method, for which the authors provide lower bounds on the achievable sample complexity. {However, POPCA has substantially greater memory demand than $\alg$, since it maintains an estimate of the possibly dense $n\times n$ sample covariance matrix of incoming data.}

 The $\operatorname{PETRELS}$ algorithm \cite{chi2013petrels} operates on one column at a time (rather than blocks) and global convergence for $\operatorname{PETRELS}$, namely convergence to a stationary point  of the underlying nonconvex program, is known. Designed for streaming {matrix completion}, the algorithm in \cite{mardani2015subspace} also operates on one column at a time  and asymptotic onvergence to the true subspace is established, see Propositions 2 and 3 therein. This framework is also extended to tensors. $\operatorname{MOUSSE}$ in  \cite{xie2013change} tracks a union of subspaces  rather than just one;  $\alg $ would function more like an ingredient of this algorithm. Asymptotic consistency of $\operatorname{MOUSSE}$ is also established there.
The theoretical guarantees of $\alg$ are more comprehensive in the sense that we also offer local convergence rate for $\alg$, see Theorems \ref{thm:local cvg expectation} and \ref{thm:main result}. ReProcs, introduced in \cite{lois2015online}, tracks a slowly changing subspace when initialized sufficiently close.

In the next section, we compare the performance of several of these algorithms in practice and find that $\alg$ competes empirically with state-of-the-art algorithms.

Even though we consider uniform random sampling of the entries of incoming vectors, $\alg$ can be applied to any incomplete data. For example, instead of uniform sampling analyzed here, one can perhaps sample the entries of every incoming vector based on their estimated importance. More specifically, in iteration $k$, one might  observe  each entry of the incoming vector with a probability proportional to the \emph{leverage score} of the corresponding row of the current estimate $\h{\SU}_{k-1}$. In batch or offline matrix completion, using the idea of \emph{leveraged sampling} (as opposed to uniform sampling)  alleviates the dependence on the coherence factor $\eta(\SU)$ in  \eqref{eq:sampling p exp} \cite{eftekhari2016mc,chen2015incoherence}. While interesting, we have not pursued this direction in the current work. 


\section{Simulations}
\label{sec:simulations}

This section consists of two parts: first, we empirically study the dependence of $\alg$ on various parameters, and second we compare $\alg$ with existing algorithms for streaming subspace estimation with missing data.
 In all simulations, we consider an $r$-dimensional subspace $\SU\subset\mathbb{R}^n$ and a sequence of generic vectors $\{s_t\}_{t=1}^T\subset \SU$. Each entry of these vectors is observed with probability $p\in(0,1]$ and collected in  vectors $\{y_t\}_{t=1}^T\subset\mathbb{R}^n$. Our objective is to estimate $\SU$ from $\{y_t\}$, as described in Section \ref{sec:problem statement}.

\paragraph{Sampling probability} We first set $n=100$, $r=5$, and let $\SU$ be a generic $r$-dimensional subspace, namely the span of an $n\times r$ standard random Gaussian matrix. For various values of probability $p$, we run $\alg$ with  block size $b=2r=10$ and scope of $T=500r=2500$,  recording the average estimation error $d_{\mathbb{G}}(\SU,\h{\SU}_K)$ over $50$ trials, see \eqref{eq:err metric}. The average error versus probability is plotted in Figure \ref{fig:variable prob}.

\paragraph{Subspace dimension}
With the same setting as the previous paragraph,  we now set $p=3r/n=0.15$ and vary the subspace dimension $r$, block size $b=2r$, and scope $T=500r$. The average error versus subspace dimension is plotted in Figure \ref{fig:variable rank}.

\paragraph{Ambient dimension}
This time, we set $r=5$, $p=3r/n$, $b=2r$, $T=500r$, and vary the ambient dimension $n$. In other words, we vary $n$ while keeping the number of samples per vector fixed at about $ pn = 3r$.
The average error versus ambient dimension is plotted in Figure \ref{fig:variable ambient}. Observe that the performance of $\alg$ steadily degrades as $n$ increases. This is in agreement with Theorem \ref{thm:local cvg expectation} by substituting $p=3r/n$ there, {which states that the error reduces by a factor of $1-Cr/n$, in expcetation.} A similar behavior is observed for our close competitor, namely $\grouse$ \cite{balzano2015local}.

\paragraph{Block size} 
Next we set $n=100$, $r=5$, $p=3r/n$, $T=500r$, and vary the block size $b$.
The average error versus block size in both cases is depicted in Figure \ref{fig:variable block}. From Step 3 of Algorithm \ref{alg:Alg}, a block size of $b\ge r$ is necessary for the success of $\alg$ and qualitatively speaking larger values of $b$ lead to better stability in face of missing data, which might explain the poor performance of $\alg$ for very small values of $b$. However, as $b$ increases, the number of blocks $ K=T/b$ reduces because the scope $T$ is  held fixed. As the estimation error of $\alg$ scales like $(1-cp)^{-K}$ in Theorem \ref{thm:main result} for a certain factor $c$, the performance suffers in Figure~\ref{fig:variable block}. It appears that the choice of $b=2r$ in $\alg$ guarantees the best empirical performance.

\paragraph{Coherence}
Lastly, we set $n=300$, $r=10$, $p=3r/n$, $b=2r$, and $T=500r$. We then test the performance of $\alg$  as the coherence of $\SU$ varies, see \eqref{eq:def of coh}. To that end, let $\SU\subset\mathbb{R}^n$ be a generic subspace with orthonormal basis $S\in\mathbb{R}^{n\times r}$. {In particular $S$ is obtained by orthogonalizing the columns of a standard $n\times n$ random Gaussian matrix.} Then, the average coherence of $\SU$ over $50$ trials was $3.3334\ll n/r$ and  the average estimation error of $\alg$  was $2.795\cdot 10^{-5}$. On the other hand, let $D\in\mathbb{R}^{n\times n}$ be a diagonal matrix with entries $D[i,i]=i^{-1}$ and consider $S'=DS$. Unlike $\SU$, the new subspace $\SU':=\mbox{span}(S')$ is typically coherent since the energy of its orthonormal basis $S'$ is mostly concentrated along its first few rows. This time, the average coherence of $\SU'$ over $50$ trials was $19.1773\approx n/r$ and the average estimation error of $\alg$ was substantially worse at $0.4286$. 

\begin{figure}[ht!]
\begin{center}

\subfloat[\label{fig:variable prob}]{\protect\includegraphics[width=0.49\textwidth]{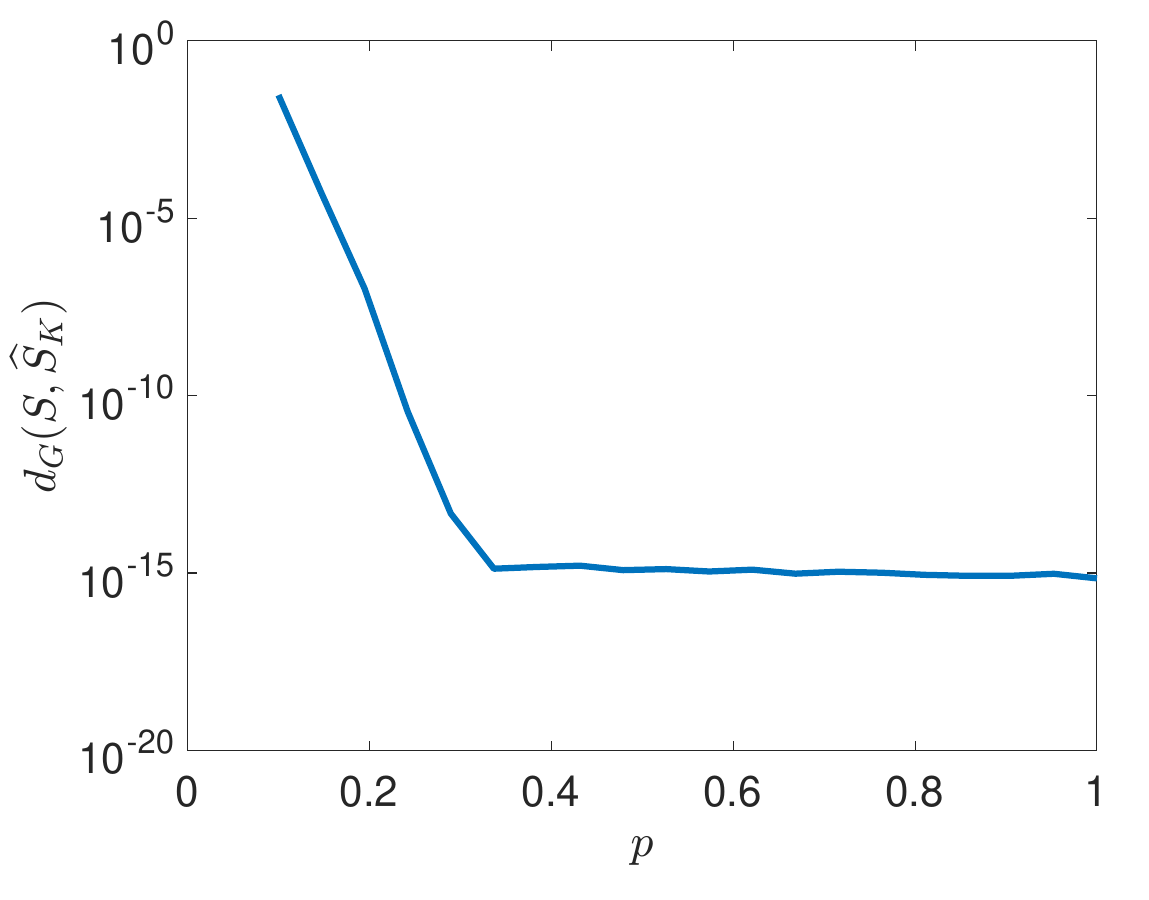}
}
\subfloat[\label{fig:variable rank}]{\protect\includegraphics[width=0.49\textwidth]{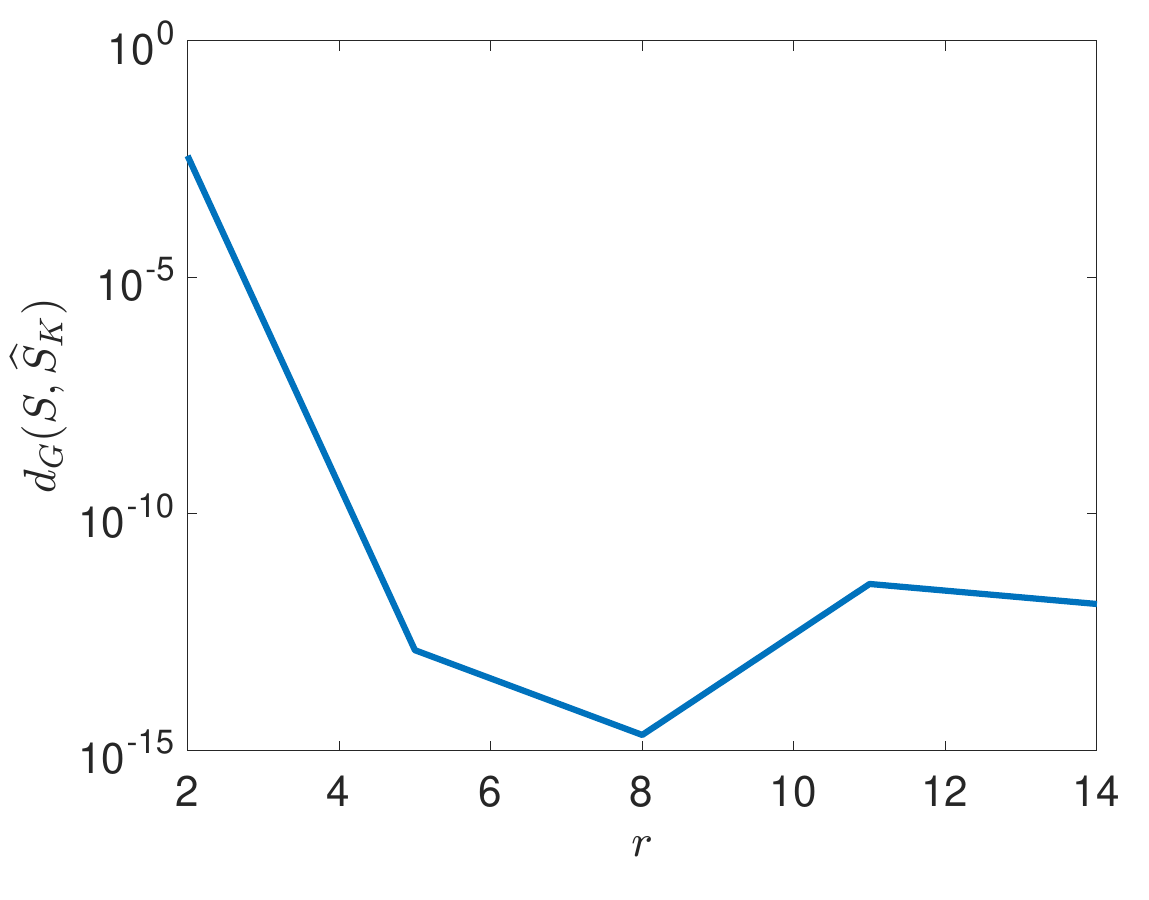}
}

\subfloat[\label{fig:variable ambient}]{\protect\includegraphics[width=0.49\textwidth]{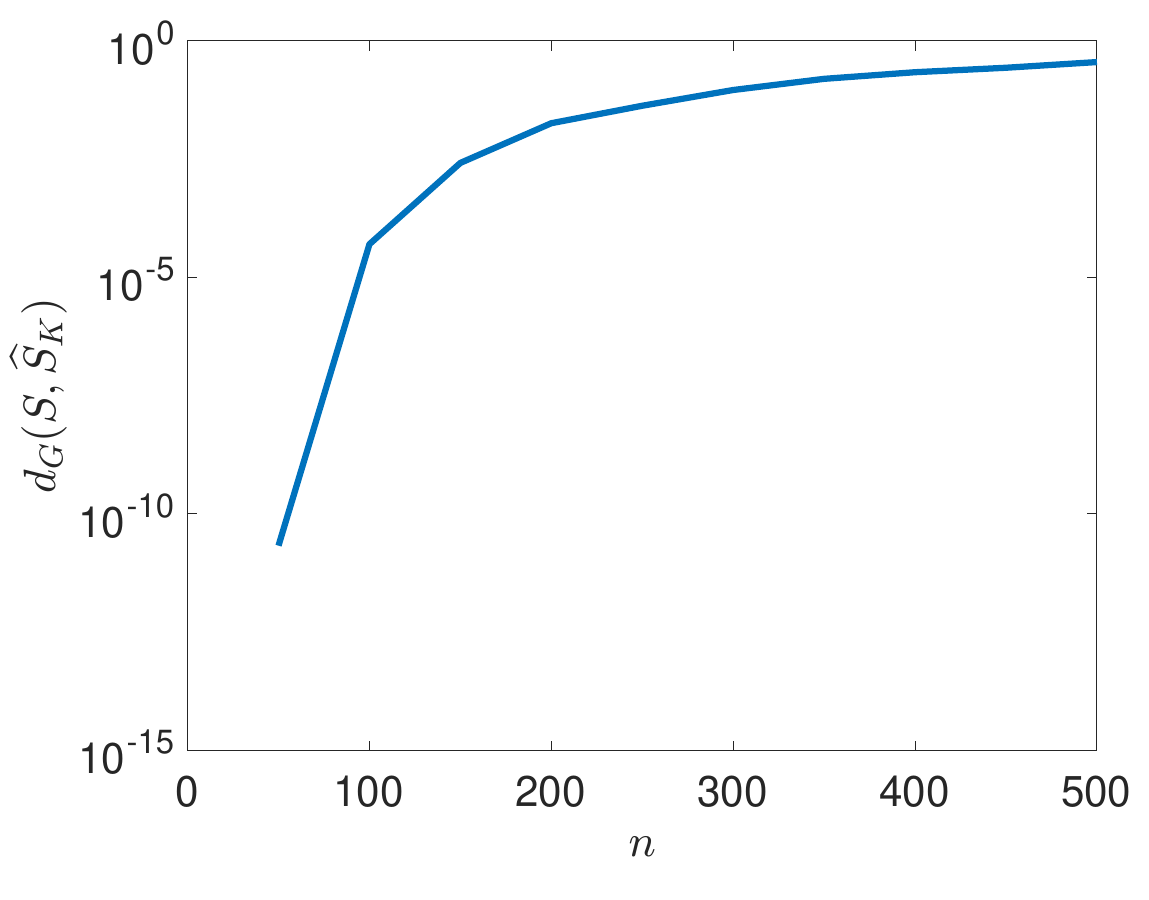}
}
\subfloat[\label{fig:variable block}]{\protect\includegraphics[width=0.49\textwidth]{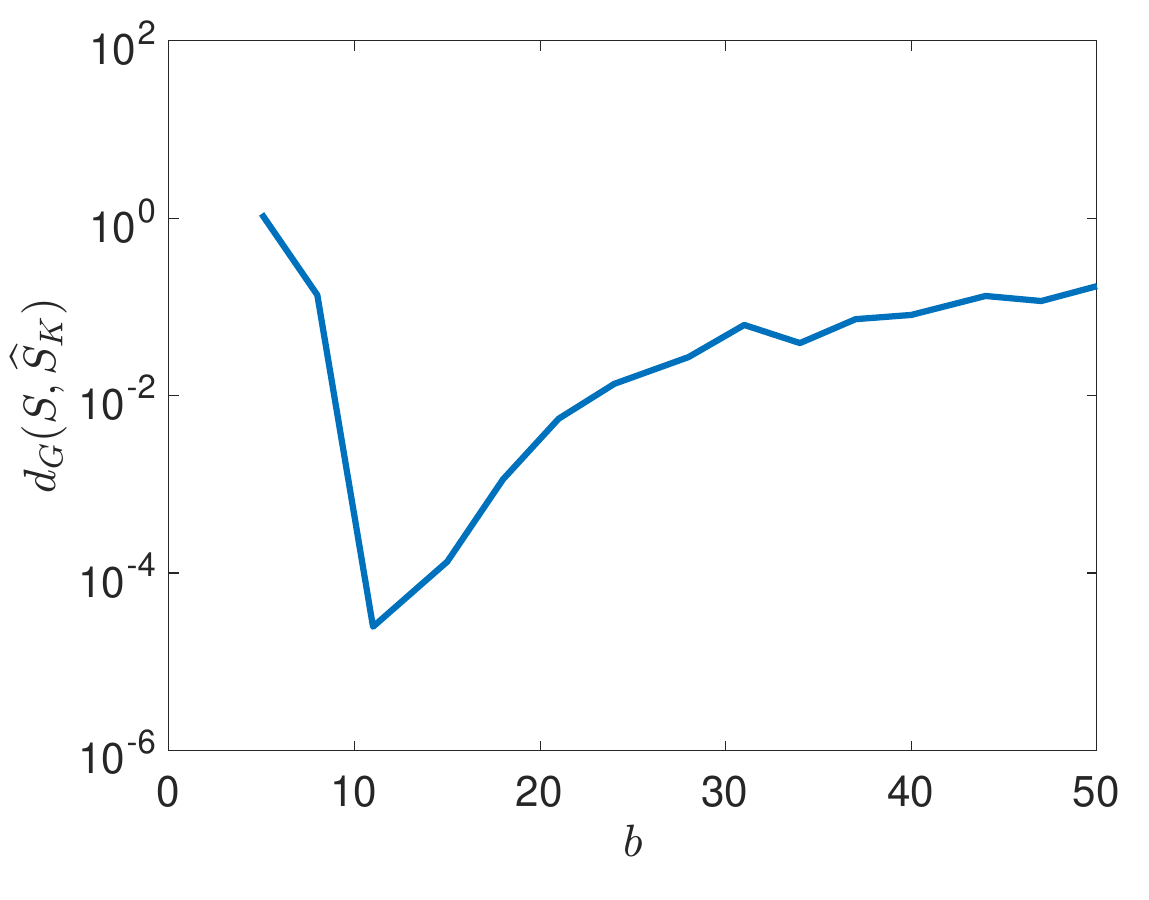}
}

\caption{Performance of $\alg$ as (a) sampling probability $p$, (b) subspace dimension $r$, (c) ambient dimension $n$, (d) block size $b$ vary. $\widehat{\SU}_K$ is the output of $\alg$ and $d_G(\SU,\widehat{\SU}_K)$ is its distance to the true subspace $\SU$, which generated the input of $\alg$. 
See Section \ref{sec:simulations} for details {and note that each panel is generated with a different and random subspace $\SU$.}}
\end{center}
\end{figure}

\paragraph{Comparisons}

Next we empirically compare $\alg$ with $\grouse$ \cite{balzano2010online,zhang2016global}, PETRELS \cite{chi2013petrels}, and the modified power method in \cite{mitliagkas2014streaming}. 
In addition to the version of $\alg$ given in Algorithm 1, we also include comparisons with a simple variant of $\alg$, which we call $\operatorname{SNIPE-overlap}$. Unlike $\alg$ which processes disjoint blocks,  $\operatorname{SNIPE-overlap}$ processes all overlapping blocks of data. More precisely, for a block size $b$, $\operatorname{SNIPE-overlap}$ first processes data columns $t=1,2,...,b$, followed by columns $t=2,3,...,b+1$, and so on, whereas regular $\alg$ processes columns $t=1,2,...,b$ followed by $t=b+1,b+2,...,2b$, etc. The theory developed in this paper does not hold for $\operatorname{SNIPE-overlap}$ because of lack of statistical independence between iterations, but we include the algorithm in the comparisons since it represents a minor modifications of the $\operatorname{SNIPE-overlap}$ framework and appears to have some empirical benefits, as detailed below.

In these experiments, we set $n=100$, $r=5$, $T=5000$, and take $\SU\subset\mathbb{R}^n$ to be a generic $r$-dimensional subspace and simulate noiseless data samples as before. In Figure \ref{fig:comparison} we compare the algorithms for three values of sampling probability $p$, which shows the average over $100$ trials of the estimation error of algorithms (with respect to the metric $d_{\mathbb{G}}$) relative to the number of revealed data columns. For $\alg$, we used the block size of $b=2r$. Having tried to get the best performance from $\grouse$, we used the ``greedy'' step-size as proposed in \cite{zhang2016global}. For \cite{mitliagkas2014streaming}, we set the block size as $b=1000$ which was found empirically to yield the lowest subspace error after $T=5000$ iterations. 

In Table \ref{table1} we also compare the average number of revealed columns needed to reach a given error tolerance for each algorithm (as measured by error metric $d_{\mathbb{G}}$) for various values of the sampling probability $p$. We omit the modified power method from the results since it was unable to reach the given error tolerances in all cases. For the medium/high sampling rates $p=0.45,0.60,0.75$, $\operatorname{SNIPE-overlap}$ is fastest to converge, while regular $\alg$ is competitive with $\grouse$ and $\operatorname{PETRELS}$. For the lower sampling rates $p=0.15,0.30$ we find $\grouse$ yields the fastest convergence, although $\operatorname{SNIPE-overlap}$ is also competitive with $\grouse$ for $p=0.30$.



\begin{table}
\begin{adjustbox}{width=\columnwidth}
\begin{tabular}{c||c|c||c|c||c|c||c|c||c|c}
 & \multicolumn{2}{|c||}{$p=0.15$} & \multicolumn{2}{|c||}{$p=0.30$} & \multicolumn{2}{|c||}{$p=0.45$} & \multicolumn{2}{|c||}{$p=0.60$}  & \multicolumn{2}{|c}{$p=0.75$} \\ 
 & $d_{\mathbb{G}} < 10^{-3}$ & $d_{\mathbb{G}} < 10^{-7}$ & $d_{\mathbb{G}} < 10^{-3}$ & $d_{\mathbb{G}} < 10^{-7}$ & $d_{\mathbb{G}} < 10^{-3}$ & $d_{\mathbb{G}} < 10^{-7}$ & $d_{\mathbb{G}} < 10^{-3}$ & $d_{\mathbb{G}} < 10^{-7}$ & $d_{\mathbb{G}} < 10^{-3}$ & $d_{\mathbb{G}} < 10^{-7}$\\ \hline
GROUSE & {\bf 878.0} (76.1) & {\bf 1852.4} (85.2) & {\bf 294.9} (20.6) & {\bf 646.1} (26.8) & 181.6 (13.8) & 391.2 (18.0) & 130.2 (10.3) & 277.2 (12.8) & 105.1 (11.7) & 213.5 (12.8)  \\ \hline
PETRELS & 1689.0 (1394.1) & 2853.7 (916.9) & 421.8 (31.3) & 1100.5 (64.2) & 262.1 (25.3) & 802.0 (31.1) & 181.8 (20.2) & 671.4 (22.8) & 133.3 (21.3) & 599.1 (24.0) \\ \hline
SNIPE & 1815.7 (137.9) & 3946.9 (182.4) & 537.4 (39.9) & 1236.4 (62.5) & 282.4 (25.8) & 649.6 (41.5) & 171.4 (17.2) & 391.4 (25.4) & 105.5 (9.1) & 241.6 (15.6) \\ \hline
SNIPE-overlap & 1588.3 (183.8) & 3319.5 (232.9) & 318.7 (27.1) & 704.6 (36.8) & {\bf 131.8} (11.6) & {\bf 296.4} (17.4) & {\bf 71.3} (5.7) & {\bf 155.6} (9.9) & {\bf 44.2} (4.2) & {\bf 91.8} (6.8)  \\ \hline
\end{tabular}
\end{adjustbox}
\caption{Average number of revealed data columns needed to reach the indicated subspace recovery error for various sampling probabilities $p$ over 100 random trials. Standard deviations are given in parenthesis.}
\label{table1}
\end{table}

\begin{figure}[h]
\begin{center}
\subfloat[$p=0.30$]{\protect\includegraphics[width=0.49\textwidth]{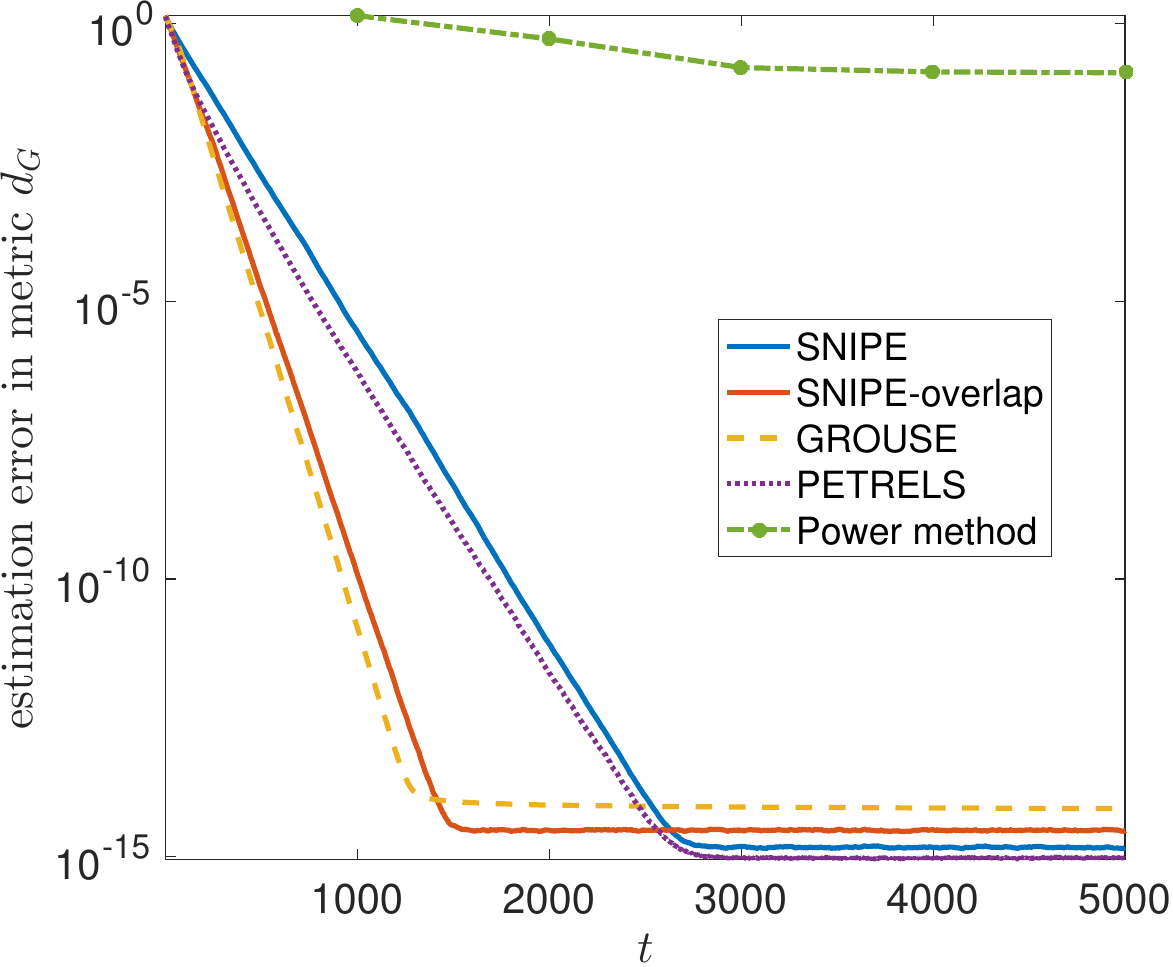}}~
\subfloat[$p=0.45$]{\protect\includegraphics[width=0.49\textwidth]{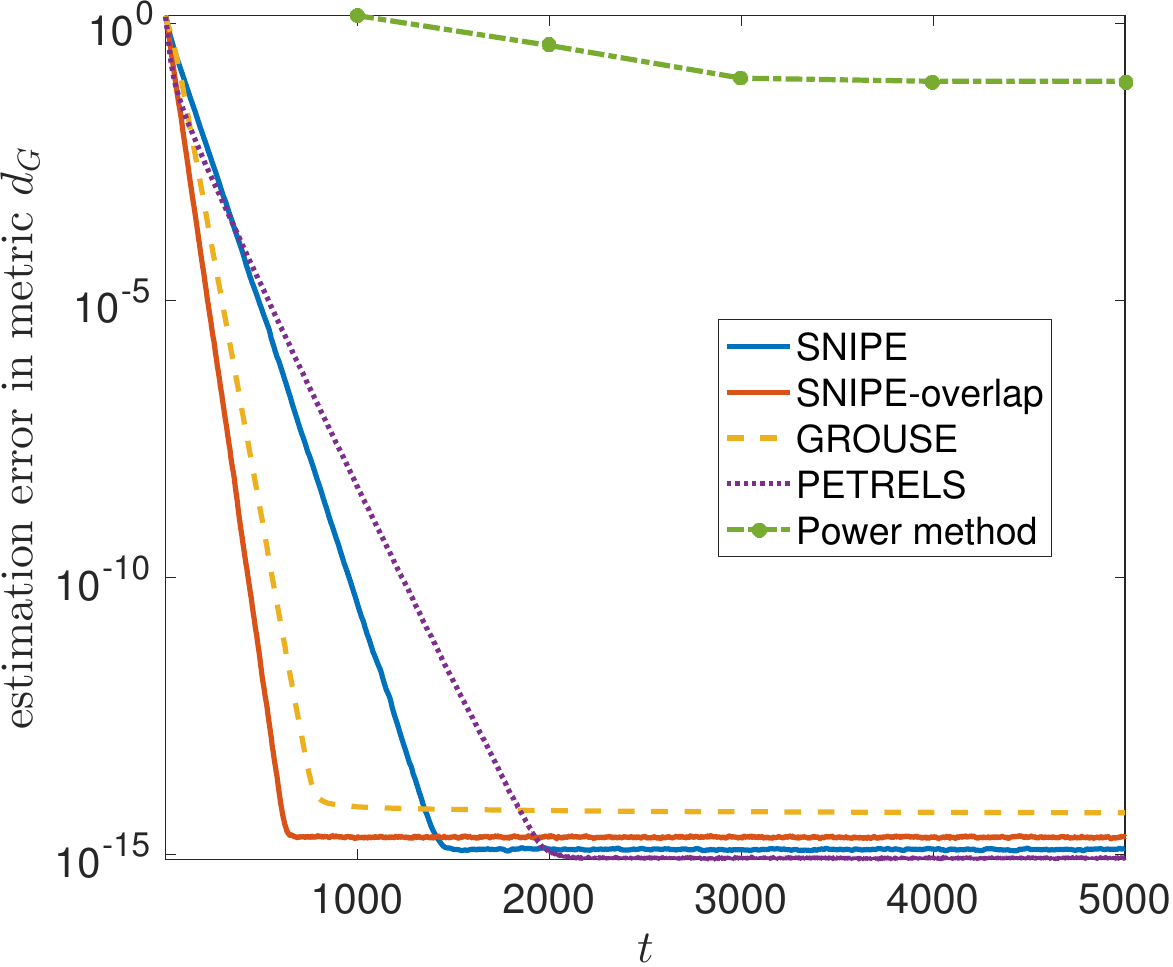}}~

\subfloat[$p=0.60$]{\protect\includegraphics[width=0.49\textwidth]{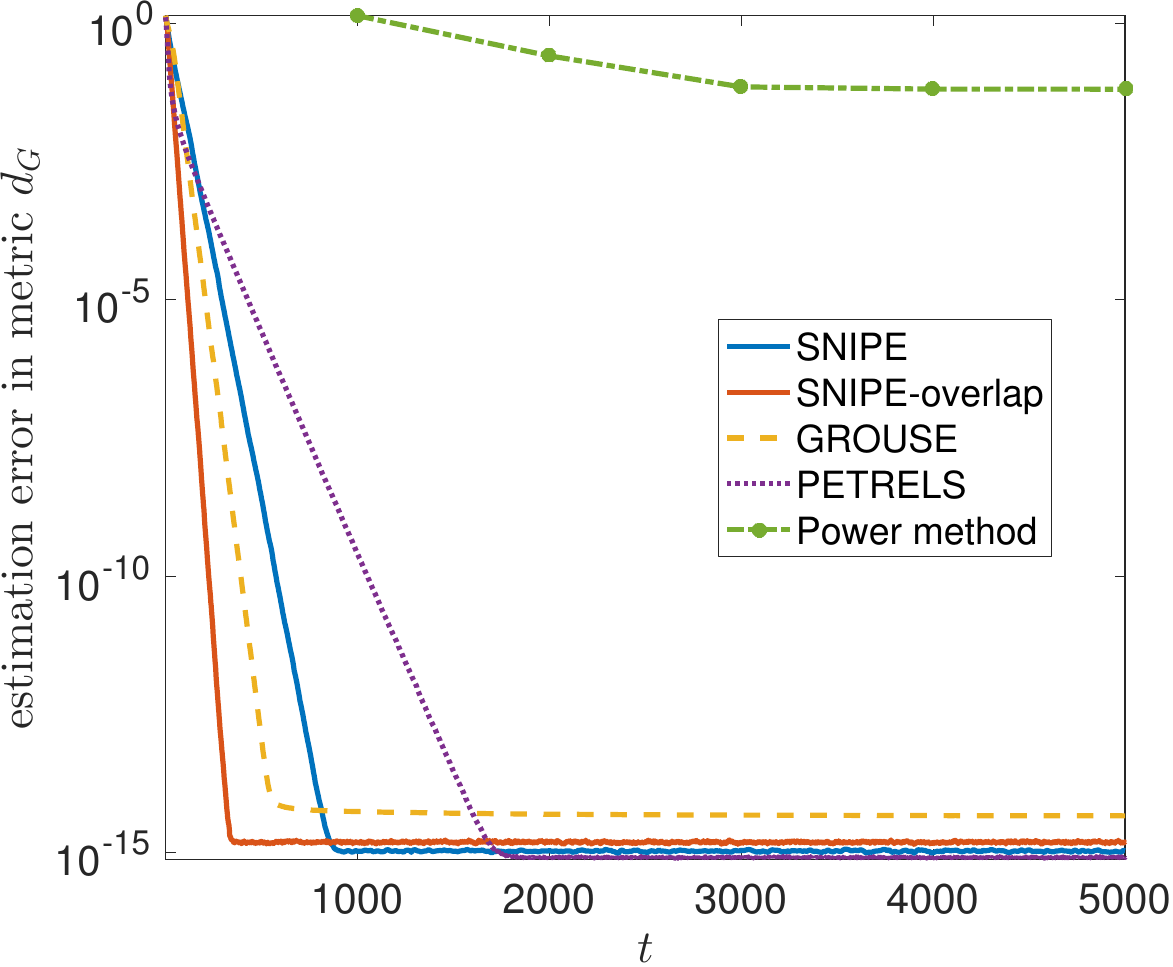}}

\caption{Average subspace estimation error versus number of revealed data columns at the specified sampling probability $p$, see Section \ref{sec:simulations} for details.}
\label{fig:comparison}
\end{center}
\end{figure}





\section{Theory}
\label{sec:analysis}

In this section, we  prove  the technical results presented in Section \ref{sec:theory}. 
A short word on notation is in order first. We will frequently use MATLAB's matrix notation so that, for example, $A[i,j]$ is the $[i,j]$th entry of $A$, and the row-vector $A[i,:]$ corresponds to the $i$th row of $A$. By $\{\epsilon_i\}_i \overset{\operatorname{ind.}}{\sim} \mbox{Bernoulli}(p)$, we mean that $\{\epsilon_i\}_i$ are independent Bernoulli random variables taking one with probability of $p$ and zero otherwise. Throughout, $E_{i,j}$ stands for the $[i,j]$th canonical matrix so that $E_{i,j}[i,j]=1$ is its only nonzero entry. The size of $E_{i,j}$ may be inferred from the context. As usual, $\|\cdot\|$ and $\|\cdot\|_F$ stand for the spectral and Frobenius norms. In addition, $\|A\|_\infty$ and $\|A\|_{2\rightarrow\infty}$ return the largest entry of a matrix $A$ (in magnitude) and the largest $\ell_2$ norm of the rows of $A$, respectively. Singular values of a matrix $A$ are denoted by $\sigma_1(A)\ge \sigma_2(A)\ge \cdots$. 
For purely aesthetic reasons, we will occasionally use the notation $a \vee b = \max(a,b)$ and $a \wedge b = \min(a,b)$.

\subsection{Convergence of SNIPE to a Stationary Point (Proof of Theorem \ref{prop:cvg to critic point}) \label{sec:cvg to crit pnt}}

Consider Program \eqref{eq:non cvx}, namely 
\begin{equation}
\begin{cases}
\min \,\,\, f_\Omega\l(X,\U\r) := \| P_{\U}^\perp X \|_F^2 + \lambda\| P_{\Omega^C} (X)\|_F^2 ,\\
P_{\Omega}(X) = Y,
\end{cases}
\label{eq:main pr}
\end{equation}
where the minimization is over matrix $X\in\mathbb{R}^{n\times T}$ and subspace $\U\in\GR(n,r)$. 
Before proceeding with the rest of the proof, let us for future reference record the partial derivatives of $f_\Omega$ below. Consider a small perturbation to $X$ in the form of  $X+\Delta$, where $\Delta\in\mathbb{R}^{n\times T}$. Let $U\in\mathbb{R}^{n\times r}$ and $U^\perp$ be orthonormal bases for $\U$ and its orthogonal complement $\U^\perp$, respectively. Consider also a small perturbation to $U$ in the form of $U+U^\perp \Delta'\in\mathbb{R}^{n\times r}$, where $\Delta'\in\mathbb{R}^{(n-r)\times r}$. The perturbation to $f_\Omega(X,\U) = f_\Omega(X,U)$ can be written as 
\begin{align}
f_\Omega\l(X+\Delta,U+U^\perp \Delta'\r)  = f_\Omega\l(X,\U \r) + 
\l\langle \Delta , \partial_X f_\Omega\l(X,\U\r) \r\rangle + \l\langle\Delta',\partial_{\U} f_\Omega\l(X,\U \r) \r\rangle +o(\|\Delta\|_F)+o(\|\Delta'\|_F),
\end{align}
where $o(\cdot)$ is the standard little-$o$ notation. The partial derivatives of $f_\Omega$ are listed below and derived in Appendix \ref{lem:derivatives}. 
\begin{lem}\label{lem:stationary pnts}
For $f_\Omega$ in Program (\ref{eq:main pr}), the first-order partial derivatives at $(X,\U)\in\mathbb{R}^{n\times T}\times \GR(n,r)$ are 
\begin{equation}
\partial_X f_\Omega\l(X,\U \r) = 2 P_{\U^\perp} X + 2\lambda P_{\Omega^C}(X) \in\mathbb{R}^{n\times T} ,
\qquad \partial_{\U} f_\Omega\l(X,\U\r) =- 2 (U^\perp)^*XX^*U \in\mathbb{R}^{(n-r)\times r},
\label{eq:first partials of f}
\end{equation}
where $U\in\mathbb{R}^{n\times r}$ and $U^\perp\in\mathbb{R}^{n\times (n-r)}$ are orthonormal bases for $\U$ and its orthogonal complement, respectively.  
\end{lem}

Recall  that $Q_k\in \R^{r\times b}$ is a random matrix with bounded expectation, namely $\E \|\Q_k\|_F< \infty$. As $K\rightarrow\infty$, $\{Q_k\}_{k=1}^K$ therefore has a bounded subsequence. To keep the notation simple and without any loss of generality, we assume that in fact the sequence $\{Q_k\}_k$ is itself bounded. 
As $K\rightarrow\infty$ and for an integer $l$, we can always find an interval of length $l$ over which  the same index set and nearly the same coefficient matrix repeats. More specifically, consider an index set $\h{\Omega}\subseteq [1:n]\times [1:b]$ and a matrix $\h{Q}\in\R^{r\times b}$ in the support of the distributions from which $\{\Omega_k\}_{k=1}^K$ and $\{Q_k\}_{k=1}^K$ are drawn.  For every integer $l$, {as a result of the second Borel-Cantelli lemma \cite[pg.~64]{durrett2010probability}, almost surely} there exists a contiguous interval 
\begin{equation}
\I_l := [k_l-l+1:k_l],
\label{eq:interval 0}
\end{equation}
such that 
\begin{equation}
\Omega_k = \h{\Omega},\qquad  k\in \I_l,
\label{eq:repeated pattern}
\end{equation}
\begin{equation}
\max_{k\in \I_l}\| Q_{k} - \h{Q} \|_F \le \frac{1}{l}. 
\label{eq:repeated coeff}
\end{equation}
As $l\rightarrow\infty$, the measurements corresponding to the interval $\I_{l}$ converge. To be specific, let  $\h{Y} :=  P_{\h{\Omega}}(S\h{Q})$ and note that 
\begin{align}
\lim_{l\rightarrow\infty}\max_{k\in \I_{l}}\,\,\,\| Y_k-\h{Y}\|_F & = \lim_{l\rightarrow\infty} \max_{k\in \I_{l}}\,\,\, \| P_{\Omega_k}(SQ_k )-\h{Y}\|_F \nonumber\\
& =\lim_{l\rightarrow\infty} \max_{k\in \I_{l}} \,\,\, \| P_{\h{\Omega}}(SQ_k)  - \h{Y}\|_F \nonumber\\
& = \lim_{l\rightarrow\infty} \max_{k\in \I_{l}}\,\,\, \|  P_{\h{\Omega}} (S(Q_k- \h{Q}) )\|_F \nonumber\\
& \le \lim_{l\rightarrow\infty} \max_{k\in \I_{l}}\,\,\, \|  Q_k- \h{Q} \|_F \nonumber\\
& = 0.
\label{eq:limit behaviour}
\end{align}
{The above observation encourages us to exchange $(\Omega_k,Y_k)$ with $(\h{\Omega},\h{Y})$ on the interval $\I_{l}$.} Let us therefore study the program 
\begin{equation}
\begin{cases}
\min\,\,\, f_{\h{\Omega}}(X,\U),\\
 P_{\h{\Omega}}(X) = \h{Y} ,
 \end{cases}
 \label{eq:main eps zero}
\end{equation}
where the minimization is over all matrices $X\in\R^{n\times b}$ and subspaces  $\U\in\GR(n,r)$. From a technical viewpoint, it is  in fact more convenient to relax the equality constraint above as 
\begin{equation}
\begin{cases}
\min\,\,\, f_{\h{\Omega}}(X,\U),\\
 \| P_{\h{\Omega}}(X) - \h{Y} \|_F\le \epsilon,
 \end{cases}
 \label{eq:main eps}
\end{equation}
for $\epsilon > 0$. We fix $\epsilon$ for now.  Let us next use alternative minimization to solve Program \eqref{eq:main eps}. More specifically, recall  \eqref{eq:interval 0} and consider the initialization $\h{\SU}_{k_l-l,\epsilon} := \h{\SU}_{k_l-l}$, where $\h{\SU}_{k_l-l}$ is the output of $\alg$ at iteration $k_l-l$, see Algorithm \ref{alg:Alg}. For every $k\in \I_{l}$, consider the program 
\begin{align}
\begin{cases}
\min\,\,\, f_{\h{\Omega}}(X,\h{\SU}_{k-1,\epsilon}),\\
\| P_{\h{\Omega}}(X) - \h{Y} \|_F \le \epsilon,
\end{cases}
\label{eq:def of R_eps}
\end{align}
and let $R_{k,\epsilon}$ be a minimizer of Program \eqref{eq:def of R_eps}. 
We then update the subspace by solving 
\begin{equation}
\underset{\U\in \GR(n,r)}{\min}\,\,\, f_{\h{\Omega}}(R_{k,\epsilon},\U),\\
\label{eq:alt proof subspace}
\end{equation}
and setting $\h{\SU}_{k,\epsilon}$ to be a minimizer of Program \eqref{eq:alt proof subspace}. Recalling the definition of $f_{\h{\Omega}}$ in Program \eqref{eq:main pr} and in light of  the Eckart-Young-Mirsky Theorem,  Program \eqref{eq:alt proof subspace} can be solved by computing top $r$ left singular vectors of $R_{k,\epsilon}$ \cite{eckart,mirsky}. 
For future reference, note that the optimality and hence stationarity of $\h{\SU}_{k,\epsilon}$ in Program~\eqref{eq:alt proof subspace} dictates  that 
\begin{equation}
\partial_{\U} f_{\h{\Omega}}(R_{k,\epsilon},\h{\SU}_{k,\epsilon}) = 0, \qquad k\in \I_l,
\label{eq:st leg 1}
\end{equation}
where $\partial_{\U}f$ was specified in Lemma \ref{lem:stationary pnts}. 
From the above construction of the sequence  $\{ (R_{k,\epsilon},\h{\SU}_{k,\epsilon})\}_{k\in \I_{l}}$, we also observe that 
\begin{equation}
0 \le f_{\h{\Omega}}(R_{k,\epsilon},\h{\SU}_{k,\epsilon}) \le f_{\h{\Omega}}(R_{k,\epsilon},\h{\SU}_{k-1,\epsilon}) \le f_{\h{\Omega}}(R_{k-1,\epsilon},\h{\SU}_{k-1,\epsilon}),
\label{eq:nest}
\end{equation}
for every $k\in [k_l-l+2:k_l] \subset \I_l$, see \eqref{eq:interval 0}. 
That is, $\{f_{\h{\Omega}}(R_{k,\epsilon},\h{\SU}_{k,\epsilon})\}_{k\in \I_{l}}$ is a nonincreasing and nonnegative sequence.
It therefore holds that 
\begin{align}
\lim_{l\rightarrow\infty}  \l| f_{\h{\Omega}}(R_{k_l-1,\epsilon},\h{\SU}_{k_l-1,\epsilon})  - f_{\h{\Omega}}(R_{k_l,\epsilon},\h{\SU}_{k_l-1,\epsilon})  \r| = 0.
\label{eq:other maximizer}
\end{align}
By the feasibility of $R_{k_l-1,\epsilon}$ in Program \eqref{eq:def of R_eps} and by the continuity of $f_{\h{\Omega}}(X,\U)$ in $X$, we conclude in light of \eqref{eq:other maximizer} that  $R_{k_l-1,\epsilon}$ too is a minimizer (and hence also a stationary point) of Program \eqref{eq:def of R_eps} and in the limit of $l\rightarrow\infty$. 
We therefore find  by writing the stationarity conditions of Program \eqref{eq:def of R_eps} at $R_{k_l,\epsilon}$ that 
\begin{equation}
\| P_{\h{\Omega}}(R_{k_l,\epsilon}) - \h{Y} \|_F \le \epsilon, 
\label{eq:st cnd 1.5}
\end{equation}
\begin{equation}
\lim_{l\rightarrow\infty} \l\| \partial_X f_{\h{\Omega}}(R_{k_l,\epsilon},\h{\SU}_{k_l,\epsilon}) + \lambda_{k_l,\epsilon} (P_{\h{\Omega}}(R_{k_l,\epsilon})-\h{Y}) \r\|_F = 0. 
\label{eq:st of fixed Y cnd 2}
\end{equation}
for nonnegative $\lambda_{k_l,\epsilon}$.
Recalling the definition of $f_{\h{\Omega}}$ and that $\lambda>0$ by assumption, we observe that Program~\eqref{eq:def of R_eps} is strongly convex in $P_{\h{\Omega}^C}(X)$ and consequently any pair of minimizers of Program \eqref{eq:def of R_eps} must agree on the index set $\h{\Omega}^C$. Optimality of $R_{k_l,\epsilon}$ and limit optimality of $R_{k_l-1,\epsilon}$ in Program~\eqref{eq:def of R_eps} therefore imply that 
\begin{equation}
\lim_{l\rightarrow\infty } \l\| P_{\h{\Omega}^C}(R_{k_l-1,\epsilon}- R_{k_l,\epsilon})\r \|_F = 0. 
\label{eq:cauchy 1}
\end{equation}
On the index set $\h{\Omega}$, on the other hand, the feasibility of both $R_{k_l-1,\epsilon}$ and $R_{k_l,\epsilon}$ in Program \eqref{eq:def of R_eps} implies that 
\begin{align}
\| P_{\h{\Omega}}(R_{k_l-1,\epsilon}- R_{k_l,\epsilon} ) \|_F & 
\le  \| P_{\h{\Omega}}(R_{k_l-1,\epsilon})- \h{Y}  \|_F + \| P_{\h{\Omega}}(R_{k_l,\epsilon})- \h{Y}  \|_F 
\qquad \mbox{(triangle inequality)}
\nonumber\\
& \le 2\epsilon.
\label{eq:cauchy 2}
\end{align} 
Combining \eqref{eq:cauchy 1} and \eqref{eq:cauchy 2} yields that 
\begin{align}
& \lim_{l\rightarrow\infty } \| R_{k_l-1,\epsilon}- R_{k_l,\epsilon} \|_F   \nonumber\\
& \le \lim_{l\rightarrow\infty } \| P_{\h{\Omega}}(R_{k_l-1,\epsilon}- R_{k_l,\epsilon}) \|_F + \lim_{l\rightarrow\infty } \| P_{\h{\Omega}^C}(R_{k_l-1,\epsilon}- R_{k_l,\epsilon} )\|_F 
\qquad \mbox{(triangle inequality)}
\nonumber\\
& 
\le 2\epsilon, \qquad \mbox{(see (\ref{eq:cauchy 1},\ref{eq:cauchy 2}))}
\label{eq:Rkeps cvgs pre}
\end{align}
In light of \eqref{eq:def of R_eps}, $\{R_{k_l,\epsilon}\}_l$ is bounded and consequently has a convergent subsequence. Without loss of generality and to simplify the notation, we assume that $\{R_{k_l,\epsilon}\}_l$ is itself convergent, namely that there exists $R_\epsilon\in\R^{n\times b}$ for which
\begin{equation}
\lim_{l\rightarrow\infty } \| R_{k_l,\epsilon}- R_{\epsilon} \|_F 
\le 2\epsilon .
\label{eq:implication of lemma}
\end{equation}
Let us now send $\epsilon$ to zero in \eqref{eq:implication of lemma} to obtain that 
\begin{equation}
\lim_{\epsilon\rightarrow 0}\lim_{l\rightarrow\infty } \| R_{k_l,\epsilon}- R_{\epsilon} \|_F 
\le \lim_{\epsilon\rightarrow 0}2\epsilon = 0.
\label{eq:one limit}
\end{equation}
We next show that it is possible to essentially change the order of limits above and also conclude that $(R,\h{\SU})$ coincides with the output of $\alg$ in limit. The following result is proved in Appendix \ref{sec:proof of lemma exchange}.
\begin{lem}\label{lem:exchange order of lims} With the setup above, there exist a sequence $\{\epsilon_i\}_i$ with $\lim_{i\rightarrow\infty} \epsilon_i = 0$ and a matrix $R\in\R^{n\times b}$ such that 
\begin{align}
\lim_{l,i\rightarrow\infty}  \|R_{k_l,\epsilon_i}-R\|_F & 
= \lim_{l\rightarrow\infty}\lim_{i\rightarrow\infty}  \|R_{k_l,\epsilon_i}-R\|_F \nonumber\\
& = \lim_{i\rightarrow\infty}\lim_{l\rightarrow\infty}  \|R_{k_l,\epsilon_i}-R\|_F \nonumber\\
& = 0.
\label{eq:final lim of R_k,i}
\end{align}
Moreover, suppose that the output of $\alg$ in every iteration has a spectral gap in the sense that  there exists $\tau>0$ such that 
\begin{equation}
\frac{\sigma_r(R_{k})}{\sigma_{r+1}(R_{k})} \ge  1+\tau, 
\label{eq:gap designed}
\end{equation}
for every $k$. Let $\h{\SU}_{k_l,\epsilon_i}$ and $\h{\SU}$ be the span of  top $r$ left singular vectors of $R_{k_l,\epsilon_i}$ and $R$, respectively. Then it holds that 
\begin{equation}
\lim_{l,i\rightarrow\infty}  d_{\GR}(\h{\SU}_{k_l,\epsilon_i},\SU) = 0.
\label{eq:final lim of S_ki}
\end{equation}
Lastly, in the limit of $l\rightarrow\infty$, $\alg$ produces $(R,\h{\SU})$ in every iteration, namely 
\begin{equation}
 \lim_{l\rightarrow\infty }   \| R_{k_l} - R\|_F = 
  \lim_{l\rightarrow\infty } d_{\GR}(\h{\SU}_{k_l},\h{\SU}) = 0,
\end{equation}
where $(R_{k_l},\h{\SU}_{k_l})$ is the output of $\alg$ in iteration $k_l$, see Algorithm \ref{alg:Alg}. 
\end{lem}
In fact, the pair $(R,\h{S})$ from Lemma \ref{lem:exchange order of lims} is stationary in limit in the sense described next and proved in Appendix \ref{sec:proof of final blockwise st}.
\begin{lem}\label{lem:final blockwise st}
The pair $(R,\h{S})$ in Lemma \ref{lem:exchange order of lims} is a stationary point of the program 
\begin{equation}
\begin{cases}
\min\,\,\, f_{\Omega_{k_l}}(X,\U),\\
P_{\Omega_{k_l}}(X) = Y_{k_l},
\end{cases}
\label{eq:ideal block}
\end{equation}
as $l\rightarrow\infty$. The minimization above is over all matrices $X\in\R^{n\times b}$ and subspaces $\U\in\GR(n,r)$. More specifically, it holds that 
\begin{equation}
\lim_{l\rightarrow\infty}  \| \partial_{\U} f_{\Omega_{k_l}}(R,\h{\SU}) \|_F = 0,
\label{eq:st cnd pre final 1}
\end{equation}
\begin{equation}
\lim_{l\rightarrow\infty}\| P_{{\Omega}_{k_l}}(R) - {Y}_{k_l} \|_F =0, 
\label{eq:st cnd pre final 2}
\end{equation}
\begin{equation}
\lim_{l\rightarrow\infty}  \| P_{{\Omega}_{k_l}^C}( \partial_X f_{\Omega_{k_l}}(R,\h{\SU}) )\|_F = 0. 
\label{eq:st cnd pre final 3}
\end{equation}
\end{lem}
In words, Lemmas \ref{lem:exchange order of lims} and \ref{lem:final blockwise st} together imply that the output of $\alg$ in limit is a stationary point of Program \eqref{eq:ideal block}. 
This completes the proof of Theorem \ref{prop:cvg to critic point}.

\subsection{Convergence of SNIPE (Proof of Proposition \ref{prop:if lim exists}) \label{sec:proof of if converges}}

In iteration $k$ of $\alg$, we partially observe the data block $S Q_k \in \R^{n\times b}$ on a random index set $\Omega_k\subset[1:n]\times [1:b]$, where $Q_k\in\R^{r\times b}$ is a random coefficient matrix. We collect the observations in $Y_k = P_{\Omega_k}(SQ_k) \in\R^{n\times b}$, see Sections \ref{sec:proposed alg} and \ref{sec:interp} for the detailed setup.  
Note that $R_k$ in \eqref{eq:Rk} can be written as 
\begin{equation}
R_k = Y_k + P_{\Omega_k^C}(\h{S}_{k-1} Q'_k) = P_{\Omega_k}(SQ_k) + P_{\Omega_k^C}(\h{S}_{k-1} Q'_k),
\end{equation}
where 
\begin{equation}
Q'_k := \l[
\begin{array}{ccc}
\cdots & \l( \h{S}_{k-1}^* P_{\omega_t} \h{S}_{k-1} + \lambda I_r  \r)^{\dagger} y_t & \cdots 
\end{array}
\r] \in\R^{r\times b}.
\label{eq:def of Qp}
\end{equation}
By \eqref{eq:limExists}, there exists $Q''_k\in\mathbb{R}^{r\times b}$ and 
\begin{equation}
R'_k :=  P_{\Omega_k}(SQ_k) + P_{\Omega_k^C}(\h{S} Q''_k),
\label{eq:def of Rprime}
\end{equation}
such that 
\begin{equation}
\lim_{k\rightarrow\infty}\|R_k-R'_k\|_F=0. 
\label{eq:RnRprime}
\end{equation}
 In \eqref{eq:def of Rprime} above, $\h{S}\in\mathbb{R}^{n\times r}$ is an orthonormal basis for the subspace $\h{\SU}$. 
 In Algorithm \ref{alg:Alg}, the rank-$r$ truncated SVD of $R_k$ spans $\h{\SU}_k\in\GR(n,r)$, namely the output of $\alg$ in iteration $k$. Let also $\h{\SU}'_k\in\GR(n,r)$ denote the span of rank-$r$ truncated SVD of $R'_k$.
The existence of the reject option in Algorithm \ref{alg:Alg} with positive $\tau$ implies that $R_k$ has a spectral gap and therefore $\widehat{\SU}_k$ is uniquely defined. Combining this with \eqref{eq:RnRprime}, we find that $\widehat{\SU}_k'$ too is uniquely defined in the limit of $k\rightarrow\infty$. Therefore another consequence of \eqref{eq:RnRprime} is that 
\begin{equation}
\lim_{k\rightarrow\infty} d_{\GR}(\h{\SU}_k,\h{\SU}_k') = 0.
\label{eq:RnRprime subspaces}
\end{equation}
  Then we have that 
\begin{align}
\lim_{k\rightarrow\infty}d_{\GR}(\h{\SU}'_k,\h{\SU}) & \le  \lim_{k\rightarrow \infty}d_{\GR}(\h{\SU}_k',\h{\SU}_k) + \lim_{k\rightarrow \infty} d_{\GR}(\h{\SU}_{k},\h{\SU}) \nonumber\\
& \le 0+0, 
\qquad \mbox{(see (\ref{eq:RnRprime subspaces},\ref{eq:limExists}))}
\label{eq:true and Skprime}
\end{align}
namely, $\h{\SU}'_k$ converges to $\h{\SU}$ in the limit too. 
Let us now rewrite $R'_k$ as 
\begin{equation}
R'_k =  P_{\Omega_k}(SQ_k- \h{S} Q''_k) + \h{S} Q''_k,
\qquad \mbox{(see \eqref{eq:def of Rprime})}
\end{equation}
which, together with \eqref{eq:true and Skprime}, implies that 
\begin{equation}
\lim_{k\rightarrow\infty } \| \h{S}^* P_{\Omega_k}( S Q_k - \h{S} Q''_k  ) \|_F = 0. 
\end{equation}
We can rewrite the above limit in terms of the data vectors (rather than data blocks) to obtain that 
\begin{equation}
\lim_{t\rightarrow\infty } \| \h{S}^* P_{\omega_t}( S q_t - \h{S} q''_t  ) \|_F = 0,
\label{eq:limRewritten}
\end{equation}
where $\{q_t,q''_t\}_t$ form the columns of the blocks $\{Q_k,Q''_k\}_k$, and the index sets $\{\omega_t\}_t\subseteq [1:n]$ form $\{\Omega_k\}_k$.  There almost surely exists a subsequence $\{t_i\}_i$ over which $\omega_{t_i} = \{1\}$, namely there is a subsequence where we only observe the first entry of the incoming data vector. Consider a vector $q^1\in\R^r$ in the support of the distribution from which $\{q_t\}_t$ are drawn. Then there also exists a subsequence of $\{t_i\}_i$, denoted by $\{t_{i_j}\}_{i_j}$, such that $\lim_{j\rightarrow\infty}\|q_{t_{i_j}}-q^1\|_2=0$.  Restricted to the subsequence $\{t_{i_j}\}_j$, \eqref{eq:limRewritten} reads as 
\begin{align}
0 & = \lim_{j\rightarrow\infty} \| \h{S}^* P_{\omega_{t_{i_j}}}( S q_{t_{i_j}} - \h{S} q''_{t_{i_j}}  ) \|_F \nonumber\\
& = \lim_{j\rightarrow\infty} \| \h{S}^* P_{\{1\}}( S q_{t_{i_j}} - \h{S} q''_{t_{i_j}}  ) \|_F \nonumber\\
& = \| \h{S}^* P_{\{1\}}( S q^1 - \h{S} q''^1  ) \|_F, 
\end{align} 
where we set $q''^1 := \lim_{j\rightarrow\infty} q''_{t_{i_j}}$; the limit exists by \eqref{eq:def of Qp}.  Likewise, we can show that 
\begin{equation}
 v^l := P_{\{l\}}( S q^l - \h{S} q''^l  ) \in \SU^\perp , 
\qquad l\in [1:n],
\end{equation}
where $\{q''^l\}_l$ are defined similarly. 
Because $\dim(\SU^\perp)=n-r$, at most $n-r$ of the vectors $\{v^l\}_{l=1}^n$ are linearly independent. Because the supports of $\{v^l\}_l$ are disjoint, it follows that there are at most $n-r$ of the vectors $\{v^l\}_{l=1}^n$ are nonzero. Put differently, there exists an index set $I\subset[1:n]$ of size at least $r$ such that 
\begin{equation}
Sq^l = \h{S} q''^l,
\qquad l\in I.
\end{equation}
Almost surely, $\{q^l\}_{l\in I}\subset\R^r$ form a basis for $\R^r$, 
and therefore $\SU \subseteq \h{\SU} $. Because $\h{\SU}\in\GR(n,r)$ by assumption, it follows that $\h{\SU}=\SU$, which completes the proof of Proposition \ref{prop:if lim exists}.

\subsection{Locally Linear Convergence of SNIPE (Proof of Theorems \ref{thm:local cvg expectation} and \ref{thm:main result}) \label{sec:Refinement}}

At iteration $k\in[2:K]$, $\alg$ uses the current estimate $\h{\SU}_{k-1}$ and the new incomplete block $Y_k$ to produce a new estimate $\h{\SU}_k$ of the true subspace $\SU$.  The main challenge here is to compare the new and old principal angles with $\SU$, namely compare $d_{\GR}(\SU,\h{\SU}_k)$ and $d_{\GR}(\SU,\h{\SU}_{k-1})$. Lemma \ref{lem:whp main} below, proved in Appendix~\ref{sec:Proof-of-Lemma whp main},  loosely speaking states that  $d_{\GR}(\SU,\h{\SU}_k)$ reduces by a factor of $1-O(p)$ in expectation in every iteration, when $d_{\GR}(\SU,\h{\SU}_k)\lesssim p^{\frac{5}{2}}$ and ignoring all other parameters in this qualitative discussion. In other words, when sufficiently small, the estimation error of $\alg$ reduces in every iteration, but in expectation. The actual behavior of $\alg$ is  more nuanced. Indeed, Lemma \ref{lem:whp main} below also adds that the  estimation error $d_{\GR}(\SU,\h{\SU}_{k})$ in fact contracts in \emph{some} iterations by a factor of  ${1-Cp}$, namely 
$$
d_{\GR}(\SU,\h{\SU}_{k}) \lesssim  ({1-Cp}) \cdot d_{\GR}(\SU,\h{\SU}_{k-1}),
$$
provided that $d_{\GR}(\SU,\h{\SU}_{k-1})\lesssim p^{\frac{5}{2}}$. That is, when sufficiently small, the estimation error of $\alg$ reduces in some but not all iterations.
In the rest of iterations,  the error does not increase by much, namely  
$$
d_{\GR}(\SU,\h{\SU}_{k}) \approx  d_{\GR}(\SU,\h{\SU}_{k-1}),
$$
with  high probability and provided that $d_{\GR}(\SU,\h{\SU}_{k-1})\lesssim p^{\frac{5}{2}}$. 
\begin{lem}
\label{lem:whp main}
Fix $k\in[2:K]$, $\alpha,\nu\ge 1$, and $c>0$. 
 Let $\ev_{k-1}$ be the event where 
 \begin{equation}
p\gtrsim\alpha^{2}\log^{2} b \log n  \frac{\eta(\widehat{\SU}_{k-1})r}{n},
\label{eq:p large enough lemma}
\end{equation}
\begin{equation}
d_{\GR} (\SU,\h{\SU}_{k-1})\lesssim  
 \frac{p^{\frac{7}{2}}nb}{\alpha   c \log b \sqrt{r \log n}},
\label{eq:activate}
\end{equation} 
and let $\ev'_k$ be the event where $\|Q_k\|\le \nu\cdot \sigma_{\min}$, where  $\sigma_{\min}$ is the reject threshold in $\alg$, see Algorithm \ref{alg:Alg}.
Then it holds that
\begin{align}
\E\l[  d_{\GR} (\SU,\h{\SU}_{k})\, |\, \h{\SU}_{k-1},\ev_{k-1},\ev'_{k} \r]
& \le   \nu \l(1-\frac{p}{2}+\frac{p^3nb}{c}\r)d_{\GR} (\SU,\h{\SU}_{k-1})+ \frac{b^{-C\alpha}}{\sqrt{r}}.
\label{eq:dragon in exp}
\end{align}
Moreover, conditioned on $\h{\SU}_{k-1}$ and  the event $\ev_{k-1}\cap\ev'_k$, it   holds that 
\begin{equation}
d_{\GR} (\SU,\h{\SU}_{k})
\le 
\nu
  \l(1+ \frac{p^3nb}{c } \r)
d_{\GR} (\SU,\h{\SU}_{k-1}),
 \label{eq:nearly contraction}
\end{equation}
except with a probability of at most $b^{-C\alpha}$.
Lastly, a  stronger bound holds conditioned on $\h{\SU}_{k-1}$ and the event $\ev_{k-1}\cap\ev'_k$, namely 
\begin{equation}
d_{\GR} (\SU,\h{\SU}_{k})
\le 
\nu
 \l(1-\frac{p}{4}+ \frac{p^3nb}{c}  \r)
d_{\GR} (\SU,\h{\SU}_{k-1})
 \label{eq:geo fast}
\end{equation}
with a probability of at least 
\begin{equation}
\fail_k(\alpha) := 1-\exp\l(-\frac{\Cl{fail}p^2nb}{\wt{\eta}_{k-1}}\r)-b^{-C\alpha},
\label{eq:failk}
\end{equation}
where 
\begin{equation}
\wt{\eta}_{k} = \wt{\eta}(P_{\h{\SU}_{k-1}^\perp}S_k) := nb\cdot \frac{\|P_{\h{\SU}_{k-1}^\perp} S_k\|^2_{\infty}}{\|P_{\h{\SU}_{k-1}^\perp}S_k\|_F^2}.
\label{eq:eta tilda}
\end{equation}
\end{lem}

 Let us now use Lemma \ref{lem:whp main} to complete the proofs of  Theorems \ref{thm:local cvg expectation} and \ref{thm:main result}.

\subsubsection{Proof of Theorem \ref{thm:local cvg expectation}}

With the choice of $c=4p^2nb$ and $\nu = 1/\sqrt{1-p/4}$, 
\eqref{eq:dragon in exp} reads as
\begin{align}
\E\l[  d_{\GR}(\SU,\h{\SU}_k) \,|\, \h{\SU}_{k-1},\ev_{k-1},\ev'_k \r]
& \le   \nu \l(1-\frac{p}{2}+\frac{p^3nb}{c}\r)
d_{\GR}(\SU,\h{\SU}_{k-1}) + \frac{b^{-C\alpha}}{\sqrt{r}} 
\qquad \mbox{(see \eqref{eq:dragon in exp})}
\nonumber\\
& =  \sqrt{ 1-\frac{p}{4} } d_{\GR}(\SU,\h{\SU}_{k-1})+ \frac{b^{-C\alpha}}{\sqrt{r}}  \nonumber\\
& \le \l( 1- \frac{p}{8}\r) d_{\GR}(\SU,\h{\SU}_{k-1})+ \frac{b^{-C\alpha}}{\sqrt{r}}.
\label{eq:exp proof pre}
\end{align}
With the choice of 
\begin{equation}
\alpha = -\frac{C\log \l( p \sqrt{r} d_{\GR}(\SU,\h{\SU}_{k-1})/16\r)}{\log b},
\label{eq:choice of alpha exp}
\end{equation}
for an appropriate constant $C$ above, the bound in \eqref{eq:exp proof pre} simplifies to 
\begin{align}
\E\l[  d_{\GR}(\SU,\h{\SU}_{k}) \,|\, \h{\SU}_{k-1},\ev_{k-1},\ev'_k \r] & 
\le  \l( 1- \frac{p}{8}\r) d_{\GR}(\SU,\h{\SU}_{k-1}) + \frac{b^{-C\alpha}}{\sqrt{r}} 
\qquad \mbox{(see \eqref{eq:exp proof pre})}
\nonumber\\
& \le  \l( 1- \frac{p}{8}\r) d_{\GR}(\SU,\h{\SU}_{k-1}) +\frac{p}{16} d_{\GR}(\SU,\h{\SU}_{k-1})\nonumber\\
& \le  \l( 1- \frac{p}{16} \r)  d_{\GR}(\SU,\h{\SU}_{k-1}).
\label{eq:base cond}
\end{align}
Lastly we remove the conditioning on $\ev'_k$ above. Using the law of total expectation, we write that 
\begin{align}
& \E\l[ d_{\GR}(\SU,\h{\SU}_k) \,|\, \h{\SU}_{k-1},\ev_{k-1} \r]  \nonumber\\
& = \E\l[  d_{\GR}(\SU,\h{\SU}_k) \,|\, \h{\SU}_{k-1},\ev_{k-1},\ev'_k \r] \cdot \Pr[\ev'_k] +
\E\l[  d_{\GR}(\SU,\h{\SU}_k) \,|\, \h{\SU}_{k-1},\ev_{k-1},\ev'^C_k \r] \cdot \Pr[\ev'^C_k]  \nonumber\\
& \le \E\l[  d_{\GR}(\SU,\h{\SU}_k) \,|\, \h{\SU}_{k-1},\ev_{k-1},\ev'_k \r] +
 \Pr[\ev'^C_k]  
\qquad \mbox{(see \eqref{eq:err metric})} \nonumber\\
& \le  \l( 1- \frac{p}{16} \r)  d_{\GR}(\SU,\h{\SU}_{k-1})+ \Pr[\ev'^C_k]   
\qquad \mbox{(see \eqref{eq:base cond})} \nonumber\\
& \le  \l( 1- \frac{p}{32} \r) d_{\GR}(\SU,\h{\SU}_{k-1}), 
\end{align}
where the last line holds if 
$$
\Pr[\ev'^C_k] \le  \frac{p }{32 } d_{\GR}(\SU,\h{\SU}_{k-1}).
$$
With the choice of $c,\nu,\alpha$ above, let us also rewrite the event $\ev_{k-1}$ in Lemma \ref{lem:whp main}. First, we rewrite \eqref{eq:activate}  as 
\begin{align}
d_{\GR}(\SU,\h{\SU}_{k-1}) & \lesssim  
  \frac{p^{\frac{7}{2}}nb}{\alpha c  \log b \sqrt{r \log n}} \nonumber\\
& =  \frac{ Cp^{\frac{3}{2}}}{\alpha  \log b \sqrt{r\log n}}
\nonumber\\
& = -\frac{Cp^{\frac{3}{2}}}{ \log \l(p\sqrt{r} d_{\GR}(\SU,\h{\SU}_{k-1}) /16 \r)  \sqrt{r \log n}}.
\qquad \mbox{(see \eqref{eq:choice of alpha exp})}
\label{eq:activate exp}
\end{align} 
Second, we replace the coherence $\eta(\h{\SU}_{k-1})$ in \eqref{eq:p large enough lemma} with the simpler quantity $\eta(\SU)$. We can do so thanks to Lemma \ref{lem:coherence perturbation} which roughly speaking states that a pair of subspaces $\A$ and $\B$ with a  small principal angle  have similar coherences, namely $ \theta_1(\A,\B) \approx 0 \Longrightarrow \eta(\A)\approx \eta(\B)$. More concretely, note that 
\begin{align}
\sqrt{\eta(\h{\SU}_{k-1} )} & \le \sqrt{\eta\l(\SU \r)}+ d_{\GR}(\SU,\h{\SU}_{k-1})   \sqrt{n}
\qquad \mbox{(see Lemma \ref{lem:coherence perturbation})}
\nonumber\\
& \le \sqrt{\eta(\SU)} + C p^{\frac{3}{2}} \sqrt{\frac{n}{r\log n}} \qquad \mbox{(see \eqref{eq:activate exp})} \nonumber\\
& \le \sqrt{\eta(\SU)} + 1
\qquad \l( \mbox{if } p \lesssim \frac{1}{\sqrt{nb}} \r) \nonumber\\
& \le 2\sqrt{\eta(\SU)}. \qquad \mbox{(see \eqref{eq:coh is bounded})}
\label{eq:coh relation 1}
\end{align}
This completes the proof of Theorem \ref{thm:local cvg expectation}. 

\subsubsection{Proof of Theorem \ref{thm:main result}}

For $K_0\in[1:K]$, we condition on $\h{\SU}_{K_0}$. For positive  $c$ to be set later,  suppose that 
\begin{equation}
d_{\GR}(\SU,\h{\SU}_{K_0}) \lesssim  
 \frac{ e^{-\frac{Cp^3 n b}{\widetilde{\eta}}} p^{\frac{7}{2}}nb}{\alpha c  \log b \sqrt{\log n}}.
\label{eq:close enough}
\end{equation}
In particular, \eqref{eq:close enough} implies that  the error at iteration $K_0$ is small enough to activate Lemma \ref{lem:whp main}, see \eqref{eq:activate}. For $\nu\ge 1$ to be set later, we condition for now on the event 
\begin{equation}
\ev' := \cap_{k=K_0+1}^K \ev'_k,
\label{eq:cap of all evs}
\end{equation}
where the event $\ev'_k$ was defined in Lemma \ref{lem:whp main}.  
Suppose also that \eqref{eq:p large enough lemma} holds for every $k\in [K_0+1:K]$, namely 
\begin{equation}
p\gtrsim \max_{k\in [K_0:K-1]} \eta(\widehat{\SU}_{k}) \cdot \frac{r\log^2 b \log n }{n},
\label{eq:lwr bnd on p complete}
\end{equation}
which will next allow us to apply Lemma \ref{lem:whp main} repeatedly to all iterations in the interval $[K_0+1:K]$. With the success probability $\phi_k(\alpha)$ defined in Lemma \ref{lem:whp main}, let us also define 
\begin{equation}
\phi(\alpha) := \min_{k\in [K_0+1:K]} \phi_k(\alpha) = 1- \exp\l(-\frac{\Cr{fail} p^2 nb}{\wt{\eta}} \r) -b^{-C\alpha},
\qquad \wt{\eta} := \max_{k\in [K_0+1:K]} \wt{\eta}_{k} \ge 1,
\label{eq:def of phi smallest}
\end{equation}
where the inequality above follows because $\wt{\eta}_k\ge 1$ for every $k$, see \eqref{eq:eta tilda}.
We now partition $[K_0+1:K]$ into (non-overlapping) intervals $\{\Int_i\}_i$,  each with the length 
\begin{equation}
\len = \frac{\Cl{len}\log b \log(K-K_0) }{\phi(\alpha)},
\label{eq:length defined}
\end{equation}
except possibly the last interval which might be shorter. 
Consider one of these intervals, say $\Int_1$. Then by Lemma \ref{lem:whp main} and the union bound,  \eqref{eq:nearly contraction} holds for every iteration $k\in\Int_1$ except with a probability of at most $\len \cdot b^{-C\alpha}$ because the length of $\Int_1$ is $l$. That is, the estimation error does not increase by much in every iteration in the interval $\Int_1$. In some of these iterations, the error in fact reduces. More specifically, 
\eqref{eq:geo fast} holds in iteration $k$ with a probability of at least $\fail_k(\alpha) $, see \eqref{eq:failk}. While $b^{-C\alpha}$ in \eqref{eq:failk} can be made arbitrary small by increasing the tuning parameter $\alpha$, this of course would  not necessarily make $\fail_k(\alpha)$ arbitrary close to one. That is, there is a sizable chance that the estimation error does not contract in iteration $k$.
However, \eqref{eq:geo fast} holds at least in one iteration in the interval $\Int_1$ except with a probability of at most 
$$
(1-\phi(\alpha))^{\len}\le e^{-\phi(\alpha) \len} = b^{-\Cr{len}\log (K-K_0)}.
\qquad \l(  1+a \le e^a\r)
$$
Therefore, except with a probabilty of at most $lb^{-C\alpha}+b^{-\Cr{len}\log(K-K_0)}$, \eqref{eq:geo fast} holds at least once and \eqref{eq:nearly contraction} holds for all iterations in the interval $\Int_1$. 
It immediately follows that 
\begin{align}
& \frac{ d_{\GR}(\SU,\widehat{\SU}_{K_0+l}) }{ d_{\GR}(\SU,\widehat{\SU}_{K_0}) }
\nonumber\\
& \le   \nu \l(1-\frac{p}{4} +\frac{p^3nb}{c }\r) \cdot \nu^{l-1} \l(1+\frac{p^3nb}{c } \r)^{\len-1} 
\qquad \mbox{(see (\ref{eq:nearly contraction},\ref{eq:geo fast}))}
\nonumber\\
&  \le \nu^l \exp\l(   -\frac{p}{4}+\frac{p^3np}{c } + (l-1) \frac{p^3nb}{c }   \r) 
\qquad \l( 1+a \le e^a\r)
\nonumber\\
& =  \nu^l\exp\l(  -\frac{p}{4}+ \frac{l p^3nb}{c }  \r),
\label{eq:first interval cvg rate}
 \end{align}
except with a probability of at most 
\begin{equation}
lb^{-C\alpha}+b^{-\Cr{len}\log(K-K_0)}. 
\label{eq:fail prob for one int}
\end{equation}
In particular, suppose that 
\begin{equation}
c \ge 4lp^2 nb,
\label{eq:exp will be neg}
\end{equation}
so that the exponent in the last line of \eqref{eq:first interval cvg rate} is negative. Let 
\begin{equation*}
i_{\max}:= \l\lfloor \frac{K-K_0}{l} \r\rfloor,
\end{equation*}
and note that 
\begin{equation}
\frac{ d_{\GR}(\SU,\widehat{\SU}_{K}) }{d_{\GR}(\SU,\widehat{\SU}_{K_0}) }
 = \prod_{i=1}^{i_{\max}} \frac{ d_{\GR}(\SU,\widehat{\SU}_{K_0+il}) }{  d_{\GR}(\SU,\widehat{\SU}_{K_0+(i-1)l}) }  \cdot \frac{ d_{\GR}(\SU,\widehat{\SU}_{K}) }{ d_{\GR}(\SU,\widehat{\SU}_{K_0+i_{\max}l}) }. 
 \label{eq:prod telescope}
\end{equation}
By applying the  bound in \eqref{eq:first interval cvg rate} to all intervals $\{I_i\}_i$, we then conclude that
\begin{align}
& \frac{ d_{\GR}(\SU,\widehat{\SU}_{K}) }{d_{\GR}(\SU,\widehat{\SU}_{K_0}) }
\nonumber\\
& = \prod_{i=1}^{i_{\max}} \frac{ d_{\GR}(\SU,\widehat{\SU}_{K_0+il}) }{  d_{\GR}(\SU,\widehat{\SU}_{K_0+(i-1)l}) }  \cdot \frac{ d_{\GR}(\SU,\widehat{\SU}_{K}) }{ d_{\GR}(\SU,\widehat{\SU}_{K_0+i_{\max}l}) }
\qquad \mbox{(see \eqref{eq:prod telescope})} 
\nonumber\\
& \le \l(\nu^{l}\exp\l(-\frac{p}{4}+\frac{lp^3nb}{c} \r) \r)^{\l\lfloor\frac{K-K_0}{\len} \r\rfloor} \cdot \nu^l\l(1+ \frac{p^3nb}{c} \r)^{l}
\qquad \mbox{(see (\ref{eq:first interval cvg rate},\ref{eq:nearly contraction}))}
\nonumber\\
& \le \l(  \nu^{l}\exp\l(-\frac{p}{4}+\frac{lp^3nb}{c} \r) \r)^{\l\lfloor\frac{K-K_0}{\len} \r\rfloor} \cdot \nu^l \exp\l( \frac{lp^3nb}{c} \r)
\qquad \l(1+a\le e^a \r)
\nonumber\\
& \le \nu^{K-K_0} \exp\l( \l\lfloor \frac{K-K_0}{2l} \r\rfloor \l(-\frac{p}{4}+\frac{lp^3nb}{c} \r) \r) 
\cdot \exp\l( \frac{lp^3nb}{c} \r) \nonumber\\
& \le \nu^{K-K_0} \exp\l( \frac{K-K_0}{2l} \l(-\frac{p}{4}+\frac{lp^3nb}{c} \r) \r) 
\cdot \exp\l( \frac{lp^3nb}{c} \r)
\qquad \l(\mbox{if } K-K_0 \ge l \mbox{ and \eqref{eq:exp will be neg} holds} \r) 
\nonumber\\
& 
\le  \nu^{K-K_0} \exp\l( \frac{K-K_0}{2l} \l(-\frac{p}{4}+\frac{3lp^3nb}{c} \r) \r),
\qquad \l(\mbox{if } K-K_0 \ge l  \r) 
\label{eq:cvg rate 1}
 \end{align}
except with a probability of at most
\begin{align}
\l\lceil\frac{K-K_0}{l}\r\rceil \l( lb^{-C\alpha}+b^{-\Cr{len}\log(K-K_0)}\r) & \le \frac{2(K-K_0)}{l} \l( lb^{-C\alpha}+b^{-\Cr{len}\log(K-K_0)}\r) 
\qquad \l(\mbox{if } K-K_0 \ge l \r)
\nonumber\\
& \le (K-K_0)e^{-C\alpha}+ b^{-C\Cr{len}\log (K-K_0)},
\label{eq:simplified failure prob -1}
\end{align}
which follows from an application of the union bound to the failure probability in \eqref{eq:fail prob for one int}. With the choice of $\alpha = \alpha' \log b \log(K-K_0)$ with sufficiently large $\alpha'$, the failure probability in \eqref{eq:simplified failure prob -1} simplifies to 
\begin{align}
(K-K_0)e^{-C\alpha}+b^{-C\log(K-K_0)} &
 = (K-K_0)\cdot (K-K_0)^{-C\alpha' \log b}+b^{-C\log(K-K_0)} \nonumber\\
 &\le (K-K_0)^{-C\alpha' \log b }+b^{-C\log(K-K_0)} \nonumber\\
& = b^{-C\alpha'\log (K-K_0)}+b^{-C\log(K-K_0)} \nonumber\\
& 
\le b^{-C\log(K-K_0)}. 
\label{eq:simplified failure prob}
\end{align}
The next step involves elementary bookkeeping to upper-bound the last line of \eqref{eq:cvg rate 1}. Suppose that 
\begin{equation}
\alpha \gtrsim \frac{\log\l( \frac{C \wt{\eta}}{p^2 n b} \r)}{\log b},
\label{eq:cnd on alpha}
\end{equation} 
\begin{equation}
p \lesssim \frac{1}{\sqrt{nb}}.
\label{eq:upp bnd on p}
\end{equation}
Using \eqref{eq:cnd on alpha} and \eqref{eq:upp bnd on p} with appropriate constants replacing $\gtrsim$ and $\lesssim$ above, we may verify that 
\begin{equation}
b^{-C\alpha} \le \frac{\Cr{fail}p^2 nb }{4 \wt{\eta}},
\qquad \mbox{(see \eqref{eq:cnd on alpha})}
\label{eq:bnd on b}
\end{equation}
\begin{equation}
\frac{ \Cr{fail} p^2 nb}{ \wt{\eta}} \lesssim 2,
\qquad \l( \mbox{\eqref{eq:upp bnd on p} and } \widetilde{\eta} \ge 1 \r)
\label{eq:exp trivial bnd}
\end{equation}
\begin{align}
\phi(\alpha) & = 1- \exp\l(-\frac{\Cr{fail} p^2 nb}{\wt{\eta}} \r) -b^{-C\alpha} 
\qquad \mbox{(see \eqref{eq:def of phi smallest})}
\nonumber\\
& \ge \frac{1}{2}\cdot \frac{\Cr{fail}p^2 nb }{ \wt{\eta}} -b^{-C\alpha}
\qquad \l(\mbox{\eqref{eq:exp trivial bnd} and } e^{-a}\le 1- \frac{a}{2}\mbox{ for } a\lesssim 2\r)
\nonumber\\
& 
\ge \frac{\Cr{fail} p^2nb}{4 \wt{\eta}}.
\qquad \mbox{(see (\ref{eq:bnd on b}))}
\label{eq:lwr bnd on phi}
\end{align}
\begin{align}
l & = \frac{\Cr{len} \log b \log(K-K_0) }{\phi(\alpha)}
\qquad \mbox{(see \eqref{eq:length defined})} \nonumber\\
& \le \frac{4\Cr{len} \wt{\eta} \log b \log(K-K_0)}{\Cr{fail}p^2 nb }.
\qquad \mbox{(see \eqref{eq:lwr bnd on phi})}
\label{eq:upp bnd on l}
\end{align}
Now with the choice of 
\begin{equation}
c = \frac{96 \Cr{len} }{\Cr{fail}} \wt{\eta} \log b \log(K-K_0),
\label{eq:choice of c factor}
\end{equation}
we may verify that 
\begin{equation}
-\frac{p}{4} + \frac{3l p^3nb}{c} \le -\frac{p}{8},
\qquad \mbox{(see (\ref{eq:upp bnd on l},\ref{eq:choice of c factor}))}
\label{eq:exponent bounded}
\end{equation}
and, revisiting \eqref{eq:cvg rate 1}, we find that 
\begin{align}
& \frac{ d_{\GR}(\SU,\h{\SU}_{K} )}{  d_{\GR}(\SU,\h{\SU}_{K_0} )}
\nonumber\\
& \le   \nu^{K-K_0} \exp\l( \frac{K-K_0}{2\len } \l( -\frac{p}{4}+ \frac{3l p^3nb}{c } \r)  \r) \qquad \mbox{(see \eqref{eq:cvg rate 1})} \nonumber\\
& \le  \nu^{K-K_0}\exp\l( -\frac{(K-K_0)p}{16l}  \r) 
\qquad \mbox{(see \eqref{eq:exponent bounded})}
\nonumber\\
& = \l( \nu \exp\l(-\frac{p}{16l}\r) \r)^{K-K_0} \nonumber\\
& \le  \l(\nu \exp\l(-\frac{Cp^3 nb}{\wt{\eta}  \log b\log(K-K_0)}\r) \r)^{K-K_0} 
\qquad \mbox{(see (\ref{eq:upp bnd on l}))} \nonumber\\
& \le \nu^{K-K_0}\l( 1- \frac{Cp^3 nb}{\wt{\eta}  \log b\log(K-K_0)}   \r)^{K-K_0} 
\qquad \l(\mbox{\eqref{eq:upp bnd on p} and } \wt{\eta} \ge 1 \r)  \nonumber\\
& \le \l( 1- \frac{Cp^3 nb}{\wt{\eta}  \log b\log(K-K_0)}   \r)^{K-K_0}, 
\label{eq:final bnd agg cvg rate}
 \end{align}
where we set
 \begin{equation}
 \nu =  1+ \frac{Cp^3 nb}{\wt{\eta}  \log b\log(K-K_0)}  ,
 \label{eq:choice of nu}
 \end{equation}
 for an appropriate choice of constant $C$. 
To reiterate, conditioned on the event $\ev'$ in \eqref{eq:cap of all evs}, \eqref{eq:final bnd agg cvg rate} is valid provided that (\ref{eq:close enough},\ref{eq:lwr bnd on p complete},\ref{eq:cnd on alpha},\ref{eq:upp bnd on p}) hold and except with the probability of at most $ b^{-C\log(K-K_0)}$, see \eqref{eq:simplified failure prob}.
In particular, to better interpret \eqref{eq:lwr bnd on p complete}, we next replace the coherence $\eta(\h{\SU}_k)$ therein with the simpler quantity $\eta(\SU)$. We can do so thanks to Lemma \ref{lem:coherence perturbation} which roughly speaking states that a pair of subspaces $\A$ and $\B$ with a  small principal angle  have similar coherences, namely $ \theta_1(\A,\B) \approx 0 \Longrightarrow \eta(\A)\approx \eta(\B)$. More concretely, Lemma \ref{lem:coherence perturbation} implies that 
\begin{align}
\sqrt{\eta(\h{\SU}_{k} )} & \le \sqrt{\eta\l(\SU \r)}+ d_{\GR}(\SU,\h{\SU}_k)  \sqrt{n}.
\label{eq:coh relation 1}
\end{align}
for every $k$. To bound the distance in the last line above, we observe that \eqref{eq:final bnd agg cvg rate} holds also after replacing $K$ with any $k\in [K_0+l:K]$, implying in particular that  
\begin{equation}
d_{\GR}(\SU,\h{\SU}_k) \le d_{\GR}(\SU,\h{\SU}_{K_0}),
\qquad k\in [K_0+l:K]. 
\label{eq:far}
\end{equation}
When $k\in [K_0+1:K_0+l-1]$ however, we cannot guarantee that the error reduces and all we can say is that that the error does not increase by much. That is, for every $k\in [K_0+1:K_0+l-1]$, we have that 
\begin{align}
d_{\GR}(\SU,\h{\SU}_k) & \le \nu^{k-K_0} \l(1+\frac{p^3nb}{c} \r)^{k-K_0} d_{\GR}(\SU,\h{\SU}_{K_0})\qquad \mbox{(see \eqref{eq:nearly contraction})} \nonumber\\
& \le \nu^{l} \l(1+\frac{p^3nb}{c} \r)^{l} d_{\GR}(\SU,\h{\SU}_{K_0})
 \nonumber\\
& = \nu^{l} 
\l(1+ \frac{Cp^3nb}{\wt{\eta}\log b \log(K-K_0)} \r)^l
d_{\GR}(\SU,\h{\SU}_{K_0}) \nonumber\\
& \le 
\l(1+ \frac{Cp^3nb}{\wt{\eta}\log b \log(K-K_0)} \r)^{2l}
d_{\GR}(\SU,\h{\SU}_{K_0}), 
\qquad \mbox{(see \eqref{eq:choice of nu})}
\label{eq:near pre}
\end{align}
with an appropriate choice of $C$ in \eqref{eq:choice of nu}. We continue by writing that 
\begin{align}
d_{\GR}(\SU,\h{\SU}_k) 
& \le \l(1+ \frac{Cp^3nb}{\wt{\eta}\log b \log(K-K_0)} \r)^{2l}
d_{\GR}(\SU,\h{\SU}_{K_0}) 
\qquad \mbox{(see \eqref{eq:near pre})}
\nonumber\\
& \le \exp \l( \frac{Cp^3nb}{\wt{\eta}\log b \log(K-K_0)} \cdot l \r)  
 d_{\GR}(\SU,\h{\SU}_{K_0})
\qquad \l(1+a\le e^a \r) 
 \nonumber\\
 & \le e^{Cp} d_{\GR}(\SU,\h{\SU}_{K_0}) 
 \qquad \mbox{(see \eqref{eq:upp bnd on l})} \nonumber\\
 & \lesssim d_{\GR}(\SU,\h{\SU}_{K_0}). 
 \qquad \l( p\le 1\r)
\label{eq:near}
\end{align}
Combining \eqref{eq:far} and \eqref{eq:near}, we arrive at 
\begin{equation}
d_{\GR}(\SU,\h{\SU}_{k}) \lesssim   d_{\GR}(\SU,\h{\SU}_{K_0}),
\label{eq:near n far together}
\end{equation}
for every $k\in [K_0+1:K]$, provided that (\ref{eq:close enough},\ref{eq:lwr bnd on p complete},\ref{eq:cnd on alpha},\ref{eq:upp bnd on p}) hold and except with a probability of at most $b^{-C\log(K-K_0)}$, see \eqref{eq:simplified failure prob}. Substituting the above bound into \eqref{eq:coh relation 1} yields that 
\begin{align}
\sqrt{\eta(\h{\SU}_{k} )} 
& \le \sqrt{\eta\l(\SU\r)} +d_{\GR}(\SU,\h{\SU}_{k}) \sqrt{n}\qquad \mbox{(see \eqref{eq:coh relation 1})} \nonumber\\
& \le \sqrt{\eta\l(\SU\r)} + C   d_{\GR}(\SU,\h{\SU}_{K_0}) \sqrt{n}  \qquad \mbox{(see \eqref{eq:near n far together})} \nonumber\\
& \le  
\sqrt{\eta\l(\SU\r)} + \frac{Cp^{\frac{7}{2}}n^{\frac{3}{2}}b}{\alpha c \log b \sqrt{\log n}} 
\qquad \mbox{(see \eqref{eq:close enough})} \nonumber\\
& \le  
\sqrt{\eta\l(\SU\r)} + \frac{Cp^{\frac{7}{2}}n^{\frac{3}{2}}b}{\alpha \wt{\eta} \log^2 b  \sqrt{\log n} \log(K-K_0)} 
\qquad \mbox{(see \eqref{eq:choice of c factor})} \nonumber\\
& \le  \sqrt{\eta\l(\SU\r)} + 1.\qquad \l(\mbox{\eqref{eq:upp bnd on p} and } \wt{\eta} \ge1 \r) \nonumber\\
& \le 2 \sqrt{\eta(\SU)}.
\qquad \l(\eta(\SU) \ge 1 \r)
\label{eq:coh relation 2}
\end{align}
%
%
%
Plugging back the bound above into \eqref{eq:lwr bnd on p complete} yields that 
\begin{equation}
p \gtrsim  \log^2 b \log n  \frac{\eta(\SU) r}{n}.
\label{eq:req p final}
\end{equation}
To summarize, conditioned on the event $\ev'$, we established that \eqref{eq:final bnd agg cvg rate} is valid under (\ref{eq:close enough},\ref{eq:cnd on alpha},\ref{eq:upp bnd on p},\ref{eq:req p final})  and except with a probability of at most $ b^{-C\log(K-K_0)}$. The event $\ev'=\cap_{k=K_0+1}^K \ev'_k$ itself holds except with a probability of at most $\sum_{k=K_0+1}^K \Pr[\ev'^C_k]$ by the union bound. With an application of the law of total probability, \eqref{eq:final bnd agg cvg rate} is therefore valid except with a probability of at most $b^{-C\log(K-K_0)}+ \sum_{k=K_0+1}^K \Pr[\ev'^C_k]$. This completes the proof of Theorem \ref{thm:main result}.

\section*{Acknowledgements}
AE would like to thank Anand Vidyashankar and Chris Williams for separately pointing out the possibility of a statistical interpretation of $\alg$, as discussed at the end of Section \ref{sec:interp}.  AE is also extremely grateful to Raphael Hauser for his brilliant insights. Lastly, the authors would like to acknowledge and thank Dehui Yang for his involvement in the early phases of this project. 
AE is supported by the Alan
Turing Institute under the EPSRC grant EP/N510129/1 and partially by the Turing Seed Funding grant SF019. GO and LB are supported by ARO Grant W911NF-14-1-0634. MBW is  partially supported by NSF grant CCF-1409258 and NSF CAREER grant CCF-1149225.

\bibliographystyle{unsrt}
\bibliography{ReferencesSNIPE}

\appendix

\section{Toolbox}
\label{sec:Toolbox}

This section collects a number of standard results for the reader's convenience. We begin with the following large-deviation bounds that are repeatedly used in the rest of the appendices
\cite{gross2011recovering,tropp2012user}. Throughout, $C$ is a universal constant the value of which might change in every appearance. 
\begin{lem}
\textbf{\textup{{[}Hoeffding inequality{]}
}}\textup{\emph{Let $\{z_{i}\}_i$ be a
finite sequence of zero-mean independent random variables and assume that almost surely every $z_i$ belongs to a compact interval of length $l_i$ on the real line. Then, for positive $\alpha$ and except with a probability of at most 
$
e^{{-C\alpha^2}/{\sum_i l_i^2}},
$ 
it holds that $\sum_i z_i \le \alpha$.}}
\label{lem:Bernie for fro}
\end{lem}

\begin{lem}
\textbf{\emph{{[}Matrix Bernstein inequality for spectral norm{]} }}\textup{\emph{Let
$\{Z_{i}\}_i\subset\mathbb{R}^{n\times b}$ be a finite sequence of
zero-mean independent random matrices, and set
\[
\beta:=\max_{i}\|Z_{i}\|,
\]
\[
\sigma^{2}:=\left\Vert \sum_{i}\mathbb{E}\left[Z_{i}^{*}Z_{i}\right]\right\Vert \vee\left\Vert \sum_{i}\mathbb{E}\left[Z_{i}Z_{i}^{*}\right]\right\Vert .
\]
Then, for $\alpha\ge1$ and except with a probability of at most $e^{-C\alpha}$,
it holds that
\[
\left\Vert \sum_{i}Z_{i}\right\Vert \lesssim\alpha\cdot \max\left(\log\left(n\vee b\right)\cdot \beta,\sqrt{\log\left(n\vee b\right)}\cdot\sigma\right).
\]
}}
\label{lem:Bernie for spec}
\end{lem}

For two $r$-dimensional subspaces $\A$ and $\B$ with principal angles $\theta_1(\A,\B)\ge \theta_2(\A,\B)\ge \cdots \ge \theta_r(\A,\B)$, recall the following useful identities about the principal angles between them: 
\begin{equation}
\sin\l(\theta_1\l(\A,\B\r)\r) = \l\| P_{\A^\perp} P_{\B} \r\| = \l\| P_{\A}-P_{\B} \r\|,
\label{eq:basic id about principal angle 1}
\end{equation}
\begin{equation}
\sqrt{\sum_{i=1}^r \sin^2\l(\theta_i\l(\A,\B\r)\r)}
 =   \l\| P_{\A^\perp} P_{B} \r\|_F = \frac{1}{\sqrt{2}}\l\| P_{\A}-P_{\B} \r\|_F.
 \label{eq:basic id about principal angle 2}
\end{equation}
Note also the following perturbation bound that is slightly stronger 
than the standard ones, but proved similarly nonetheless \cite{wedin1972perturbation}.
\begin{lem}
\textbf{\emph{{[}Perturbation bound{]} }}\textup{\label{lem:perturbation lemma}
}\textup{\emph{Fix a rank-$r$ matrix $A$ and let  $\A=\Span(A)$.
For matrix $B$, let $B_{r}$ be a rank-$r$ truncation of $B$ obtained via SVD and set $\B_{r}=\Span(B_{r})$.
Then, it holds that
\[
\left\Vert P_{\A}-P_{\B_{r}}\right\Vert =\left\Vert P_{\A^{\perp}}P_{\B_{r}}\right\Vert \le\frac{\left\Vert P_{\A^{\perp}}B\right\Vert }{\sigma_{r}\left(B\right) }\le\frac{\left\Vert B-A\right\Vert }{\sigma_{r}\left(B\right)}
\le \frac{\|B-A\|}{\sigma_r(A)-\|B-A\|}
,
\]
}}\textup{
\[
\frac{1}{\sqrt{2}}\left\Vert P_{\A}-P_{\B_{r}}\right\Vert _{F}=\left\Vert P_{\A^{\perp}}P_{\B_{r}}\right\Vert _{F}\le
\frac{\left\Vert P_{\A^{\perp}}B\right\Vert _{F}}{\sigma_{r}\left(B\right)}
\le\frac{\left\Vert B-A\right\Vert _{F}}{\sigma_{r}\left(B\right)}
,
\]
where $\sigma_{r}(A)$ is the $r$ largest singular value of 
$A$. }\end{lem}
\begin{proof}
Let $B_{r^{+}}:=B-B_{r}$ denote the residual and note that
\begin{align*}
\left\Vert P_{\A^{\perp}}P_{\B_{r}}\right\Vert  & =\left\Vert P_{\A^{\perp}}B_{r}B_{r}^{\dagger}\right\Vert \qquad\left(\B_{r}=\mbox{span}(B_{r})\right)\\
 & =\left\Vert P_{\A^{\perp}}\left(B-B_{r^{+}}\right)B_{r}^{\dagger}\right\Vert \qquad\left(B=B_{r}+B_{r^{+}}\right)\\
 & =\left\Vert P_{\A^{\perp}}BB_{r}^{\dagger}\right\Vert \qquad\left(\mbox{span}\left(B_{r^{+}}^{*}\right)\perp\mbox{span}\left(B_{r}^{*}\right)=\mbox{span}\left(B_{r}^{\dagger}\right)\right)\\
 & \le\left\Vert P_{\A^{\perp}}B\right\Vert \cdot\Vert B_{r}^{\dagger}\Vert \\
 & =\frac{\left\Vert P_{\A^{\perp}}B\right\Vert }{\sigma_{r}\left(B_{r}\right)} \\
 & = \frac{\left\Vert P_{\A^{\perp}}(B-A)\right\Vert }{\sigma_{r}\left(B_{r}\right)} \\
 & \le 
 \frac{\left\Vert B-A\right\Vert }{\sigma_{r}\left(B_{r}\right)}
\\
 & \le\frac{\|B-A\|}{\sigma_{r}\left(A\right)-\left\Vert B-A\right\Vert }.\qquad\mbox{(Weyl's inequality)}
\end{align*}
The proof is identical for the claim with the Frobenius norm and is therefore omitted.
\end{proof}

Lastly, let us record what happens to the coherence of a subspace under a small perturbation, see \eqref{eq:def of coh}.
\begin{lem}\textbf{\emph{[Coherence under perturbation]}}
\label{lem:coherence perturbation}
Let $\A,\B$ be  two $r$-dimensional subspaces in $\mathbb{R}^n$, and let $d_{\GR}(\A,\B)$ denote their distance, see (\ref{eq:err metric}). Then their coherences are related as 
\begin{equation*}
\sqrt{\eta\l( \B \r)}
\le 
\sqrt{\eta\l( \A \r)} + d_{\GR}(\A,\B) \sqrt{n}. 
\end{equation*}
\end{lem}

\begin{proof}
Let $\theta_i=\theta_i(\A,\B)$ be the shorthand for the $i$th principal angle between the subspaces $\A$ and $\B$. It is well-known \cite{golub2013matrix} that there exist orthonormal bases $A,B\in\mathbb{R}^{n\times r}$ for the subspaces $\A$ and $\B$, respectively, such that
\begin{equation}
A^* B = \mbox{diag}\l(
\l[
\begin{array}{cccc}
\cos \theta_1 &
\cos \theta_2 &
\cdots &
\cos \theta_r
\end{array}
\r]
  \r)
=: \Gamma \in \mathbb{R}^{r\times r},
\label{eq:cosines}
\end{equation}
where $\mbox{diag}(a)$ is the diagonal matrix formed from vector $a$.
There also exists $A'\in\mathbb{R}^{n\times r}$ with orthonormal columns such that
\begin{equation}
\l(A'\r)^*B = \mbox{diag}\l(
\l[
\begin{array}{cccc}
\sin \theta_1 &
\sin \theta_2 &
\cdots &
\sin \theta_r
\end{array}
\r]
  \r)
=: \Sigma \in\mathbb{R}^{r\times r },
\qquad \l(A'\r)^* A = 0,
\label{eq:sines}
\end{equation}
and, moreover,
\begin{equation}
\mbox{span}\l(
\l[
\begin{array}{cc}
A & B
\end{array}
\r]
 \r)
 =
 \mbox{span}\l(
 \l[
\begin{array}{cc}
A & A'
\end{array}
\r]
 \r).
 \label{eq:decomp proof}
\end{equation}
With $\A'=\mbox{span}(A')$, it follows that
\begin{align}
B & = P_{\A} B + P_{\A'} B
\nonumber\\
& = AA^* B + A' (A')^* B
\nonumber\\
& = A \Gamma + A' \Sigma.
\qquad \mbox{(see \eqref{eq:cosines} and \eqref{eq:sines})}
\label{eq:princ decomposition}
\end{align}
Consequently,
\begin{align}
\sqrt{\eta\l( \B\r)} &
= \sqrt{\frac{n}{r}}\max_i \l\| B[i,:] \r\|_2
\qquad \mbox{(see \eqref{eq:def of coh})}
\nonumber\\
& \le  \sqrt{\frac{n}{r}}
\max_i \l\|  A[i,:] \cdot \Gamma \r\|_2 +
\sqrt{\frac{n}{r}}
\max_i \l\| A'[i,:] \cdot \Sigma  \r\|_2
\qquad \mbox{(\eqref{eq:princ decomposition} and triangle inequality)}
\nonumber\\
& \le \sqrt{\frac{n}{r}}
\max_i \l\|  A[i,:] \r\|_2 \l\| \Gamma  \r\| +
\sqrt{\frac{n}{r}}
\max_i \l\| A'[i,:] \r\|_2 \l\| \Sigma  \r\| \nonumber\\
& = \sqrt{\eta \l( \A \r)} \l\| \Gamma \r\|+
\sqrt{\eta( \A' )} \l\| \Sigma \r\|
\qquad \mbox{(see \eqref{eq:def of coh})}
\nonumber\\
& \le \sqrt{\eta \l( \A \r)} \l\| \Gamma\r\| +
\sqrt{\frac{n}{r}} \l\| \Sigma \r\|
\qquad \l(  \eta(\A ') \le \frac{n}{r}\r)
\nonumber\\
& \le \sqrt{\eta \l( \A \r)}  +
\sqrt{\frac{n}{r}} \sin \theta_1
\qquad
\mbox{(see \eqref{eq:cosines} and \eqref{eq:sines})}\nonumber\\
& = \sqrt{\eta \l( \A \r)}  +
\sqrt{\frac{n}{r}}  \l\|P_{\A^\perp}P_{\B} \r\|
\qquad \mbox{(see \eqref{eq:basic id about principal angle 1})} \nonumber\\
& \le \sqrt{\eta \l( \A \r)}  +
\sqrt{\frac{n}{r}}  \l\|P_{\A^\perp}P_{\B} \r\|_F \nonumber\\
& = \sqrt{\eta \l( \A \r)}  +
 d_{\GR}(\A,\B) \sqrt{n},
\qquad \mbox{(see \eqref{eq:err metric})}
\end{align}
which completes the proof of Lemma \ref{lem:coherence perturbation}.
\end{proof}

\section{Supplement to Section \ref{sec:interp}}
\label{sec:interp verifying}
In this section, we verify that
\begin{align}
R_k = \begin{cases}
\arg \min\,\,\, \|P_{\h{\SU}_{k-1}^\perp}X_k \|_F^2 +\lambda\| P_{\Omega^C_k}(X_k)\|_F^2 \\
P_{\Omega_k}\l(X_k\r) = Y_k,
\end{cases}
\label{eq:Rk interp app}
\end{align}
when $k\ge 2$. The optimization above is over $X_k\in\R^{n\times b}$.  First note that Program \eqref{eq:Rk interp app} is separable and equivalent to the following $b$ programs:
\begin{equation}
R_k[:,j] =
\begin{cases}
\arg\min \,\,\, \| P_{\h{\SU}_{k-1}^\perp} x\|_2^2 + \lambda\| P_{\omega_t^C} x \|_F^2 \\
 P_{\omega_t} \cdot  x
 = y_t,
\end{cases}
\qquad t = (k-1)b+j,
\,\,\, j \in [1:b].
\label{eq:Rk separated}
\end{equation}
Above, $R_k[:,j]\in\mathbb{R}^n$ is the $j$th column of $R_k$ in MATLAB's matrix notation and the optimization is over $x\in \R^n$. To solve the $j$th program in \eqref{eq:Rk separated}, we make the change of variables $x=y_t+W_{\omega_t^C} \cdot z$. Here, $z\in\mathbb{R}^{n-m}$ and  $W_{\omega_t^C}\in\{0,1\}^{n\times (n-m)}$ is defined naturally so that $P_{\omega_t^C}=W_{\omega_t^C}W_{\omega_t^C}^*$.   We now rewrite \eqref{eq:Rk separated} as the following $b$ unconstrained programs:
\begin{align}
z_j & := \arg \min \,\,\, \| P_{\h{\SU}_{k-1}^\perp} y_t
+ P_{\h{\SU}_{k-1}^\perp} W_{\omega_t^C}   z
\|_2^2  +\lambda  \| P_{\omega_t^C} W_{\omega_t^C} z  \|_F^2\nonumber\\
&  = \arg \min \,\,\, \| (\h{S}_{k-1}^\perp)^* y_t
+ (\h{S}_{k-1}^\perp)^* W_{\omega_t^C}   z
\|_2^2 + \lambda \|z\|_F^2, 
\qquad t = (k-1)b+j,
\,\,\, j \in [1:b].
\label{eq:ls prs listed}
\end{align}
Above, $\h{S}^\perp_{k-1}\in\R^{n\times (n-r)}$ is an orthonormal basis for the subspace $\h{\SU}^\perp_{k-1} \in \GR(n,n-r)$ and in particular $P_{\h{\SU}^\perp_{k-1}} = \h{S}^\perp_{k-1} (\h{S}^\perp_{k-1})^*$. The optimization above is over $z\in\R^{n-m}$. 
Note that 
\begin{equation}
z_j
= -\l( W_{\omega^C_t}^* P_{\h{\SU}_{k-1}^\perp} W_{\omega^C_t} + \lambda I_{n-m} \r)^{-1} W_{\omega_t^C}^*P_{\h{\SU}_{k-1}^\perp} y_t ,
\qquad j\in [1:b],
\end{equation}
are solutions of the least squares programs in \eqref{eq:ls prs listed} when $m$ is large enough. 
For fixed $j$, we simplify the expression for $z_j$ as follows:
\begin{align}
z_j & =
-\l( \lambda I_{n-m}+ W_{\omega^C_t}^* P_{\h{\SU}_{k-1}^\perp} W_{\omega^C_t}  \r)^{-1} W_{\omega_t^C}^* P_{\h{\SU}_{k-1}^\perp} y_t  \nonumber\\
& = - \l( (1+\lambda)I_{n-m} - W_{\omega^C_t}^* P_{\h{\SU}_{k-1}} W_{\omega_t^C}  \r)^{-1} W_{\omega_t^C}^* P_{\h{\SU}_{k-1}^\perp} y_t
\qquad \l(P_{\h{\SU}_{k-1}^\perp} = I_n - P_{\h{\SU}_{k-1}}  \r) \nonumber\\
& = \frac{-1}{1+\lambda} \l( I_{n-m} + W_{\omega^C_t}^* \h{S}_{k-1}
\l((1+\lambda) I_r - \h{S}_{k-1}^*P_{\omega_t^C}\h{S}_{k-1} \r)^{-1} \h{S}_{k-1}^*W_{\omega_t^C}
  \r)
 W_{\omega_t^C}^* P_{\h{\SU}_{k-1}^\perp} y_t
 \qquad \mbox{(inversion lemma)} \nonumber\\
 & = \frac{-1}{1+\lambda}\l( I_{n-m} + W_{\omega^C_t}^* \h{S}_{k-1}
\l((1+\lambda) I_r - \h{S}_{k-1}^*P_{\omega^C_t}\h{S}_{k-1} \r)^{-1} \h{S}_{k-1}^*W_{\omega^C_t}
  \r)
 W_{\omega_t^C}^* P_{\h{\SU}_{k-1}^\perp} P_{\omega_t} y_t
 \qquad \l( y_t = P_{\omega_t}y_t \r) \nonumber\\
 & = \frac{1}{1+\lambda} \l( I_{n-m} + W_{\omega_t^C}^* \h{S}_{k-1}
\l( (1+\lambda) I_r - \h{S}_{k-1}^*P_{\omega^C_t}\h{S}_{k-1} \r)^{-1} \h{S}_{k-1}^*W_{\omega_t^C}
  \r)
 W_{\omega_t^C}^* P_{\h{\SU}_{k-1}}  y_t
 \qquad \l( W_{\omega_t^C} P_{\omega_t}=0 \r) \nonumber\\
 & = \frac{W_{\omega_t^C}P_{\h{\SU}_{k-1}}y_t}{1+\lambda}
 + \frac{W_{\omega_t^C}^* \h{S}_{k-1}}{1+\lambda}
\l( (1+\lambda) I_r - \h{S}_{k-1}^*P_{\omega^C_t}\h{S}_{k-1} \r)^{-1} \h{S}_{k-1}^*P_{\omega_t^C} P_{\h{\SU}_{k-1}} y_t
\qquad \l( P_{\omega^C_t} = W_{\omega_t^C} W_{\omega_t^C}^*\r)
\nonumber\\
 & = 
\frac{W_{\omega_t^C}P_{\h{\SU}_{k-1}}y_t}{1+\lambda}
 + \frac{W_{\omega_t^C}^* \h{S}_{k-1}}{1+\lambda}
\l( (1+\lambda) I_r - \h{S}_{k-1}^*P_{\omega^C_t}\h{S}_{k-1} \r)^{-1}
\l( \h{S}_{k-1}^*P_{\omega_t^C}\h{S}_{k-1}-(1+\lambda) I_r\r) \h{S}_{k-1}^*y_t
\nonumber\\
& \qquad + W_{\omega_t^C}^* \h{S}_{k-1} \l( (1+\lambda) I_r- \h{S}_{k-1}^* P_{\omega_t^C} \h{S}_{k-1} \r)^{-1} \h{S}_{k-1}^* y_t \nonumber\\
& = W_{\omega_t^C}^* \h{S}_{k-1} \l( (1+\lambda) I_r- \h{S}_{k-1}^* P_{\omega_t^C} \h{S}_{k-1} \r)^{-1} \h{S}_{k-1}^* y_t \nonumber\\
& = W_{\omega_t^C}^* \h{S}_{k-1} \l( \lambda I_r+ \h{S}_{k-1}^* P_{\omega_t} \h{S}_{k-1} \r)^{-1} \h{S}_{k-1}^* y_t,
\end{align}
which means that
\begin{equation}
y_t + W_{\omega_t^C} z_j =
y_t + P_{\omega_t^C} \h{S}_{k-1}\l(\lambda I_r + \h{S}_{k-1}^* P_{\omega_t} \h{S}_{k-1}  \r)^{-1} \h{S}_{k-1}^* y_t,
\end{equation}
is a solution of the $j$th program in  \eqref{eq:Rk separated} which indeed matches the $j$th column of $R_k$ defined in \eqref{eq:Rk}.

\section{Proof of Proposition \ref{lem:init}}
\label{sec:proof of lemma init}


Let us form the blocks $Q_1\in\R^{b_1\times r}$, $S_1\in\R^{n\times b_1}$, and $\Omega_1\subseteq[1:n]\times [1:b_1]$ as usual, see Section \ref{sec:interp}. As in that section, we also write the measurement process as $Y_1=P_{\Omega_1}(S_1)$, where $P_{\Omega_1}(\cdot)$ projects onto the index set $\Omega_1$. Let us fix $Q_1$ for now. Also let $Y_{1,r}\in \mathbb{R}^{n\times b_1}$ be a rank-$r$ truncation of $Y_1$ obtained via SVD.
$\alg$ then  sets $\h{\SU}_1=\mbox{span}(Y_{1,r})$. Our objective here is to control $\|P_{\SU^\perp} P_{\h{\SU}_1}\|_F$. Since
\begin{equation}
\|P_{\SU^\perp} P_{\h{\SU}_1}\|_F \le \sqrt{r} \|P_{\SU^\perp} P_{\h{\SU}_1}\|,
\qquad \l( \h{\SU}_1\in\mathbb{G}(n,r) \r)
\label{eq:rank-r adv}
\end{equation}
it suffices to bound the spectral norm.  Conditioned on $Q_1$, it is easy to  verify that $\mathbb{E}[Y_1]=\mathbb{E}[P_{\Omega_1}(S_1)]= p\cdot S_1$, suggesting that we might consider $Y_1$ as a perturbed copy of $p\cdot S_1$ and perhaps consider  $\h{\SU}_1=\mbox{span}(Y_{1,r})$ as a perturbation of $\SU=\mbox{span}(p\cdot S_1)$. Indeed, the perturbation bound in Lemma \ref{lem:perturbation lemma} dictates that
\begin{align}
\| P_{\SU^{\perp}} P_{\h{\SU}_1} \|
& \le \frac{\| Y_1-p  S_1 \|}{p\cdot \sigma_r\l(   S_1 \r)-\| Y_1-p S_1 \|} \nonumber\\
& = \frac{\l\| Y_1-p S_1 \r\|}{p \cdot \sigma_r\l( Q_1 \r)-\l\| Y_1-p S_1 \r\|}
\qquad \l( S_1 = S Q_1,\,\, S^*S =I_r \r)
\nonumber\\
& \le \frac{2}{p}\cdot  \frac{\l\| Y_1-p S_1 \r\|}{\sigma_r\l( Q_1\r)}.
\qquad \l(   \mbox{if } \l\| Y_1-p S_{Q_1}\r\|\le  \frac{p}{2} \cdot \sigma_r(Q_1)\r)
\label{eq:perturb init}
\end{align}
It remains to bound the norm in the last line above.
To that end,
we study the concentration of $Y_1$ about its expectation by writing that
\begin{align}
Y_1 - p  S_1
&= P_{\Omega_1} \l(S_1\r) -p S_{1}
\nonumber\\
& = \sum_{i,j} \l(\epsilon_{i,j}-p\r)  S_1[i,j]\cdot E_{i,j}
\nonumber\\
& =: \sum_{i,j } Z_{i,j},
\label{eq:Bernie prep initialization}
\end{align}
where  $\{\epsilon_{i,j}\}_{i,j}\overset{\operatorname{ind.}}{\sim} \mbox{Bernoulli}(p)$ and $E_{i,j}\in\mathbb{R}^{n\times b_1} $ is the $[i,j]$th canonical matrix.  Additionally, $\{Z_{i,j}\}_{i,j}$ are independent zero-mean random matrices. In order to appeal to the matrix  Bernstein inequality (Lemma~\ref{lem:Bernie for spec}), we compute the $\beta$ and $\sigma$ parameters below, starting with $\beta$:
\begin{align}
\l\| Z_{i,j} \r\| & = \l\| \l(\epsilon_{i,j}-p\r) S_1[i,j]\cdot E_{i,j} \r\| \nonumber\\
& = \l| \l(\epsilon_{i,j}-p\r) S_1[i,j] \r|\qquad  \l( \l\| E_{i,j} \r\|=1\r)\nonumber\\
& \le \l|S_1[i,j] \r| \qquad \l( \epsilon_{i,j}\in\{0,1\}  \r) \nonumber\\
& \le \l\| S_{1}\r\|_\infty
\nonumber\\
& \le \l\| S\r\|_{2\rightarrow\infty} \l\|Q_1 \r\|_{2\rightarrow\infty}
\qquad \l(S_{1} = SQ_1,\,\, \|AB^*\|_\infty\le \|A\|_{2\rightarrow\infty} \|B\|_{2\rightarrow\infty} \r)\nonumber\\
& \le \sqrt{\frac{\eta\l(\SU \r)r}{n}}\cdot
\sqrt{\frac{\eta\l( \Q_1 \r)r}{b_1}} \cdot \|Q_1\|
\qquad \l( \mbox{see \eqref{eq:def of coh}}
\r) \nonumber\\
& =: \beta.
\label{eq:beta in init}
\end{align}
Above, $\|A\|_\infty$ and $\|A \|_{2\rightarrow\infty}$ return the largest entry of $A$ in magnitude and the largest $\ell_2$ norm of the rows of matrix $A$, respectively.  As for $\sigma$, we write that
\begin{align}
\l\| \mathbb{E} \l[ \sum_{i,j} Z_{i,j} Z_{i,j}^* \r] \r\|&
 = \l\| \sum_{i,j} \E \l[ \l(\epsilon_{i,j}-p \r)^2 \r] S_1[i,j]^2 \cdot E_{i,i} \r\|
\nonumber\\
& = \l\| \sum_{i,j} p(1-p) S_1[i,j]^2 \cdot E_{i,i} \r\|
\qquad \l( \epsilon_{i,j} \sim \mbox{Bernoulli}(p) \r) \nonumber\\
 & \le p \l\|\sum_{i,j} S_1[i,j]^2 \cdot E_{i,j}\r\| \nonumber\\
 & = p \l\| \sum_i \l\|S_1[i,:]\r\|^2_2 \cdot E_{i,i} \r\| \nonumber\\
 & = p \max_i \l\| S_1[i,:] \r\|_2^2 \nonumber\\
 & = p \l\| S_1 \r\|_{2\rightarrow\infty}^2 \nonumber\\
 & \le p \l\|S \r\|_{2\rightarrow\infty}^2 \cdot \l\| Q_1\r\|^2
 \qquad \l( S_1 = SQ_1, \,\,  \|AB\|_{2\rightarrow\infty} \le \|A\|_{2\rightarrow\infty}\|B\|\r)
\nonumber\\
& = p \cdot \frac{\eta\l( \SU\r)r}{n} \cdot \|Q_1\|^2.
\qquad \l(\mbox{see \eqref{eq:def of coh}}\r)
\label{eq:sigma leg 1 spec norm}
\end{align}
In a similar fashion, we find that
\begin{align}
\l\| \mathbb{E} \l[ \sum_{i,j} Z_{i,j}^* Z_{i,j} \r] \r\|
& \le  p  \l\| \sum_j \l\| S_1[:,j]\r\|_2^2 E_{j,j} \r\| \nonumber\\
& = p\l\| S_1^* \r\|_{2\rightarrow\infty}^2 \nonumber\\
& \le p\cdot \|S\|^2 \cdot \l\|Q_1\r\|_{2\rightarrow\infty}^2
\qquad \l(S_{1}=SQ_1,\,\, \l\| AB\r\|_{2\rightarrow\infty} \le \|A\|_{2\rightarrow\infty} \|B\|  \r)
\nonumber\\
& \le p \cdot \frac{\eta\l( \Q_1\r)r}{b_1} \cdot \|Q_1\|^2, \qquad \l(\|S\|=1,\,\,\mbox{see \eqref{eq:def of coh}} \r)
\label{eq:sigma leg 2 spect norm}
\end{align}
and eventually
\begin{align}
\sigma^2 & = \l\| \mathbb{E}\l[ \sum_{i,j} Z_{i,j}^*Z_{i,j}\r]\r\|
\vee \l\| \mathbb{E}\l[ \sum_{i,j} Z_{i,j} Z_{i,j}^*\r]\r\|
\nonumber\\
& \le \frac{pr}{n}  \l( 1\vee \frac{n}{b_1} \r)
\l( \eta\l( \SU\r)\vee \eta\l(\Q_1 \r)\r)\|Q_1\|^2.
\qquad \l( \mbox{see \eqref{eq:sigma leg 1 spec norm} and \eqref{eq:sigma leg 2 spect norm}}\r)
\label{eq:sigma in initialization}
\end{align}
Lastly,
\begin{align}
& \max\l( \log(n\vee b_1) \cdot \beta , \sqrt{\log(n\vee b_1)} \cdot \sigma \r)
\nonumber\\
&
\lesssim \max\l(  \log(n\vee b_1)\cdot  \frac{r}{n} ,
\sqrt{\log(n\vee b_1)} \cdot
 \sqrt{\frac{{pr}}{n}}
 \r)
\sqrt{ 1\vee \frac{n}{b_1}}
\cdot
\sqrt{\eta\l(\SU\r) \vee  \eta\l( \Q_1\r)}\cdot  \|Q_1\|
\qquad \mbox{(see \eqref{eq:beta in init} and \eqref{eq:sigma in initialization})}
\nonumber\\
& \le \sqrt{\log(n\vee b_1)} \cdot \sqrt{\frac{{pr}}{n}}
\sqrt{ 1\vee \frac{n}{b_1}} \cdot
\sqrt{\eta\l(\SU\r)\vee  \eta\l( \Q_1\r)}\cdot  \|Q_1\|.
\qquad \l(\mbox{if } p\ge \frac{\log\l( n\vee b_1\r)r}{n} \r)
\label{eq:Bernie init another prep }
\end{align}
The Bernstein inequality now dictates that
\begin{align}
\l\| Y_1 - p S_{1} \r\| & = \l\| \sum_{i,j} Z_{i,j}\r\|
\qquad \mbox{(see \eqref{eq:Bernie prep initialization})}
 \nonumber\\
& \lesssim \alpha  \max\l(\log(n\vee b_1)\cdot \beta,
\sqrt{\log(n\vee b_1} \cdot \sigma  \r)
\qquad \mbox{(see Lemma \ref{lem:Bernie for spec})}
\nonumber\\
& \lesssim
\alpha \sqrt{\log(n\vee b_1)} \cdot \sqrt{\frac{{rp}}{n}} \sqrt{1 \vee \frac{n}{b_1} }
\cdot
\sqrt{\eta\l(\SU \r) \vee \eta\l(\Q_1 \r)} \cdot \|Q_1\|,
\qquad \mbox{(see \eqref{eq:Bernie init another prep })}
\label{eq:norm bound in init}
\end{align}
except with a probability of at most $e^{-\alpha}$.
 In particular, suppose that
\begin{equation}
 p \gtrsim \alpha^2      \nu(Q_1)^2
  \l( 1\vee \frac{n}{b_1} \r)\frac{\l(\eta\l( \SU \r) \vee \eta\l( \Q_1\r) \r)r \log(n\vee b_1) }{n},
  \qquad \l( \nu(Q_1) = \frac{\|Q_1\|}{\sigma_r(Q_1)} \r)
\end{equation}
so that \eqref{eq:perturb init} holds. Then, by  substituting \eqref{eq:norm bound in init} back into \eqref{eq:perturb init} and then applying \eqref{eq:rank-r adv}, we find that
\begin{align}
\frac{\| P_{\SU^\perp}P_{\h{\SU}_1} \|_F}{\sqrt{r}} &
\le  \| P_{\SU^\perp} P_{\h{\SU}_1} \| \qquad \mbox{(see \eqref{eq:rank-r adv})}
\nonumber\\
& \le \frac{2}{p} \cdot
\frac{\l\| Y_1-pS_{1} \r\|}{\sigma_{r}(Q_1)}
\qquad \mbox{(see \eqref{eq:perturb init})} \nonumber\\
& \lesssim \alpha  \sqrt{\log(n\vee b_1) \cdot \frac{r}{pn} \l( 1\vee \frac{n}{b_1} \r)  \l( {\eta\l(\SU \r)\vee \eta\l( \Q_1\r) }\r) }\cdot \nu\l( Q_1\r)
\qquad \mbox{(see \eqref{eq:norm bound in init})}\nonumber\\
& =: \delta_1\l( \nu(Q_1),\eta(Q_1)\r),
\label{eq:bnd on frob iter 1 pre final}
\end{align}
except with a probability of at most $e^{-\alpha}$ and for fixed $Q_1$. In order to remove the conditioning on $Q_1$, fix $\nu\ge 1$, $1\le \eta_{1} \le \frac{b_1}{r}$, and recall the following inequality for events $\mathcal{A}$ and $\mathcal{B}$:
\begin{align}
\Pr\l[ \mathcal{A} \r]
& = \Pr\l[ \mathcal{A} | \mathcal{B}\r] \cdot \Pr\l[ \mathcal{B}\r] + \Pr\l[ \mathcal{A} | \mathcal{B}^C\r] \cdot \Pr\l[ \mathcal{B}^C\r]\nonumber\\
& \le \Pr\l[ \mathcal{A} | \mathcal{B}\r]  +\Pr\l[ \mathcal{B}^C\r].
\label{eq:useful event ineq}
\end{align}
Set $\Q_1=\mbox{span}(Q_1^*)$ and let $\mathcal{E}$ be the event where both $\nu(Q_1) \le \nu$ and $\eta(\Q_1)\le \eta_{1}$.
Thanks to the inequality above, we find that
\begin{align}
& \Pr\l[ \frac{\|P_{\SU^{\perp}} P_{\h{\SU}_1}\|_F}{\sqrt{r}}
\gtrsim
 \delta_1\l( \nu, \eta_{1}\r)
\r]
 \nonumber\\
&
\le \Pr\l[ \frac{\|P_{\SU^{\perp}} P_{\h{\SU}_1}\|_F}{\sqrt{r}}
\gtrsim
\delta\l( \nu,\eta_{1}\r) \mid
\mathcal{E}
\r]
+ \Pr\l[ \mathcal{E}^C \r]
\qquad \mbox{(see \eqref{eq:useful event ineq})}
\nonumber\\
& \le e^{-\alpha} + \Pr\l[ \nu(Q_1) > \nu \r]+ \Pr\l[ \eta(Q_1)> \eta_{1} \r],
\qquad \mbox{(see \eqref{eq:bnd on frob iter 1 pre final})}
\end{align}
which completes the proof of Proposition \ref{lem:init}.

\section{Proof of Lemma \ref{lem:stationary pnts} \label{lem:derivatives}}
 We conveniently define the orthonormal basis 
$$
B=\l[
\begin{array}{cc}
U & U^\perp
\end{array}
\r] \in\mathbb{R}^{n\times n}.
$$
Then the perturbation from $\U$ can be written more compactly as 
$$
U+U^\perp \Delta' = B \l[ 
\begin{array}{c}
I_r \\
\Delta'
\end{array}
\r] \in\mathbb{R}^{n\times r}.
$$
In particular, orthogonal projection onto $\mbox{span}(U+U^\perp \Delta')$ is 
\begin{align}
(U+U^\perp \Delta') (U+U^\perp \Delta')^\dagger & = 
B\l[
\begin{array}{c}
I_r \\
\Delta'
\end{array}
\r] \cdot 
\l[
\begin{array}{c}
I_r \\
\Delta'
\end{array}
\r]^\dagger B^* \nonumber\\
&  = B
\l[
\begin{array}{c}
I_r \\
\Delta'
\end{array}
\r] 
\l(I_r+ \Delta'^*\Delta'\r)^{-1}
\l[
\begin{array}{cc}
I_r & \Delta'^*
\end{array}
\r]
O^* \nonumber\\
& = B
\l[
\begin{array}{cc}
I_r  & \Delta'^* \\
\Delta' & 0_{n-r}
\end{array}
\r] 
B^* + o(\|\Delta'\|_F),
\label{eq:proj diff}
\end{align}
where $0_{n-r}\in \R^{(n-r)\times (n-r)}$ is the matrix of zeros of size $n-r$. 
From \eqref{eq:main pr}, note also that 
$$
f_\Omega(X,\U) = \l\langle P_{\U^\perp}, XX^* \r\rangle + \lambda \|P_{\Omega^C}(X)\|_F^2=  \l\langle I_n - UU^\dagger,XX^*  \r\rangle + \lambda\|P_{\Omega^C}(X)\|_F^2.
$$
 We can now write that 
\begin{align}
& f_{\Omega}\l(X+\Delta,U+U^\perp \Delta' \r) \nonumber\\
& = \l\langle I_n - B \l[ 
\begin{array}{cc}
I_r& \Delta'^*\\
\Delta' & 0_{n-r}
\end{array}
\r] B^*,
\l(X+\Delta \r)\l(X+\Delta \r)^*
\r\rangle + 2\lambda P_{\Omega^C}(X)+ o(\|\Delta\|_F) + o(\|\Delta'\|_F)
\qquad \mbox{(see \eqref{eq:proj diff})}
\nonumber\\
& = \l\langle B \l[ 
\begin{array}{cc}
0_r & -\Delta'^*\\
-\Delta' & I_{n-r}
\end{array}
\r] B^*,
\l(X+\Delta \r)\l(X+\Delta \r)^*
\r\rangle +2\lambda P_{\Omega^C}(X)+  o(\|\Delta\|_F) + o(\|\Delta'\|_F) \nonumber\\
& = \l\langle  \l[ 
\begin{array}{cc}
0_r & -\Delta'^*\\
-\Delta' & I_{n-r}
\end{array}
\r],
B^*\l(XX^* + \Delta X^* + X\Delta^*  \r) B
\r\rangle + 2\lambda P_{\Omega^C}(X)+ o(\|\Delta\|_F) + o(\|\Delta'\|_F) \nonumber\\
& = \l\langle 
\l[ 
\begin{array}{cc}
0_r & 0_{r\times (n-r)}\\
0_{(n-r)\times r } & I_{n-r}
\end{array}
\r]
, B^* XX^* B
\r\rangle \nonumber\\
& +
\l\langle
\l[
\begin{array}{cc}
0_r & 0_{r\times (n-r)}\\
0_{(n-r)\times r} & I_{n-r}
\end{array}
\r]
,B^*\Delta X^* B+B^* X\Delta^* B
\r\rangle \nonumber\\
& - \l\langle
\l[ 
\begin{array}{cc}
0_r & \Delta'^*\\
\Delta' & 0_{n-r}
\end{array}
\r],
B^* XX^* B \r\rangle
 +2\lambda P_{\Omega^C}(X)+  o(\|\Delta\|_F) +o(\|\Delta'\|_F).
\end{align} 
We can further simplify the above expansion as
\begin{align}
& f_{\Omega}\l(X+\Delta,U+U^\perp \Delta'\r) \nonumber\\
& = f_{\Omega}\l( X,\U \r)  + 2 \l\langle
\Delta, P_{\U^{\perp}} X 
\r\rangle - 2\l\langle  
\Delta',(U^{\perp})^* XX^* U
\r\rangle + 2\lambda P_{\Omega^C}(X)+  o(\|\Delta\|_F) +o(\|\Delta'\|_F). 
\end{align}
Therefore the partial derivatives of $f$ are  
\begin{equation}
\partial_X f_\Omega \l(X,\U \r) = 2 P_{\U^\perp} X + 2\lambda P_{\Omega^C}(X)\in\mathbb{R}^{n\times T},
\qquad \partial_{\U} f_\Omega\l(X,\U\r) = -2 (U^\perp)^*XX^*U \in\mathbb{R}^{(n-r)\times r},
\end{equation}
which completes the proof of Lemma \ref{lem:stationary pnts}.

\section{Proof of Lemma \ref{lem:exchange order of lims} \label{sec:proof of lemma exchange}}
 By definition in \eqref{eq:def of R_eps}, $\{R_{\epsilon'}\}_{\epsilon'\le \epsilon}$ is a bounded set, see also Program \eqref{eq:main pr} for the definition of $f_{\h{\Omega}}$. Therefore there exist a subsequence $\{R_{\epsilon_i}\}_i$ and a matrix $R\in \R^{n\times b}$ such that  
\begin{equation}
\lim_{i\rightarrow \infty} \epsilon_i = 0,
\label{eq:vanishing}
\end{equation}
\begin{equation} 
\lim_{i\rightarrow \infty} \|R_{\epsilon_i}-R\|_F=0.
\label{eq:cvg of Rs over eps}
\end{equation}
That is, $\{R_{\epsilon_i}\}_i$ converges to $R$.  On the other hand, for every $\delta\ge 0$, \eqref{eq:implication of lemma} implies that there exists an integer $l_\delta$  that depneds on $\delta$ and 
\begin{equation}
\| R_{k_l,\epsilon}- R_{\epsilon} \|_F 
\le 2\epsilon l'+\delta ,
\qquad  l\ge l_\delta.
\label{eq:that above ineq}
\end{equation}
Restricted to the sequence $\{\epsilon_i\}_i$ above, \eqref{eq:that above ineq} reads as 
\begin{equation}
\| R_{k_l,\epsilon_i}- R_{\epsilon_i} \|_F 
\le 2\epsilon_i l'+\delta ,
\qquad  l\ge l_\delta,
\end{equation}
which, in the limit of $i\rightarrow \infty$, yields that
\begin{equation}
\lim_{i\rightarrow\infty} 
\| R_{k_l,\epsilon_i}- R_{\epsilon_i} \|_F 
\le \lim_{i\rightarrow \infty}2\epsilon_i +\delta = \delta,
\qquad l\ge l_\delta.
\label{eq:limit raw}
\end{equation}
We used \eqref{eq:vanishing} to obtain the identity above.
An immediate consequence of \eqref{eq:limit raw} is that 
\begin{equation}
\lim_{l\rightarrow \infty}\lim_{i\rightarrow\infty} 
\| R_{k_l,\epsilon_i}- R_{\epsilon_i} \|_F 
 =  0.
 \label{eq:two limit}
\end{equation}
Invoking \eqref{eq:cvg of Rs over eps}, it then follows that 
\begin{align}
\lim_{l\rightarrow \infty}\lim_{i\rightarrow\infty} 
\| R_{k_l,\epsilon_i}- R \|_F 
 & = \lim_{l\rightarrow \infty}\lim_{i\rightarrow\infty} 
\| R_{k_l,\epsilon_i}- R_{\epsilon_i} \|_F 
\qquad \mbox{(see \eqref{eq:cvg of Rs over eps})} \nonumber\\
& = 0. \qquad \mbox{(see \eqref{eq:two limit})}
\label{eq:final cvg of Rs}
\end{align}
Exchanging the order of limits above yields that 
\begin{align}
& \lim_{i\rightarrow\infty} \lim_{l\rightarrow\infty} \| R_{k_l,\epsilon_i}-R \|_F \nonumber\\
& 
\le \lim_{i\rightarrow\infty} \lim_{l\rightarrow\infty}  \| R_{k_l,\epsilon_i}-R_{\epsilon_i} \|_F + \lim_{i\rightarrow \infty} \|R_{\epsilon_i}-R\|_F 
\qquad \mbox{(triangle inequality)}
\nonumber\\
& =  \lim_{i\rightarrow\infty} \lim_{l\rightarrow\infty}  \| R_{k_l,\epsilon_i}-R_{\epsilon_i} \|_F 
\qquad \mbox{(see \eqref{eq:cvg of Rs over eps})} \nonumber\\
& = 0. \qquad \mbox{(see \eqref{eq:one limit})}
\label{eq:final cvg of Rs 2}
\end{align}
Therefore, \eqref{eq:final cvg of Rs} and \eqref{eq:final cvg of Rs 2} together state that $R_{k,\epsilon_i}$ converges to $R$ as $l,i\rightarrow\infty$, namely
\begin{equation}
\lim_{l,i\rightarrow\infty}  \|R_{k_l,\epsilon_i}-R\|_F = 0,\label{eq:final lim of R_k,i proof}
\end{equation}
thereby proving the first claim in Lemma \ref{lem:exchange order of lims}, see  \eqref{eq:final lim of R_k,i}. 
In order to prove the claim about subspaces in Lemma \ref{lem:exchange order of lims}, we proceed as follows. 
Recall  the output of $\alg$, namely $\{(R_k,\h{\SU}_k)\}_k$ constructed in Algorithm \ref{alg:Alg}. In light of Section \ref{sec:interp}, $R_k$ is also the unique minimizer of Program \eqref{eq:Rk interp}.  Recall that $\I_l = [k_l-l+1:k_l]$ from \eqref{eq:interval 0}. Recall also the construction of sequence $\{(R_{k,\epsilon},\h{\SU}_{k,\epsilon} )\}_{k\in \I_{l}}$ in the beginning of Section \ref{sec:cvg to crit pnt} and note that both procedures are initialized identically at the beginning of interval $\I_l$, namely  $\h{\SU}_{k_l-l,\epsilon} = \h{\SU}_{k_l-1}$.  Therefore, for  fixed $l$, observe that\footnote{To verify \eqref{eq:snipe out n st pnts pre}, note that for every feasible $X_\epsilon$ in Program \eqref{eq:main eps}, $X= Y+P_{\widehat{\Omega}^C}(X_\epsilon)$ is feasible for Program \eqref{eq:main eps zero}. Moreover, as $\epsilon\rightarrow0$, $\|X_\epsilon-X\|_F\rightarrow0$ and consequently, by continuity, $| f_{\widehat{\Omega}}(X_\epsilon,\U)- f_{\widehat{\Omega}}(X,\U) | \rightarrow0$. On the other hand, by definition, $R_k$ is the unique minimizer of Program \eqref{eq:main eps zero}, namely $f_{\widehat{\Omega}}(R_k,\U) < f_{\widehat{\Omega}}(X,\U)$ for any $X\ne R_k$ feasible for Program \eqref{eq:main eps zero}. For sufficiently small $\epsilon$, it follows that $f_{\widehat{\Omega}}(R_k,\U) < f_{\widehat{\Omega}}(X_\epsilon,\U)$ for any $X_\epsilon$ that is feasible for Program \eqref{eq:main eps} and $\lim\inf_{\epsilon\rightarrow 0} \| X_\epsilon-R_k\| >0$. That is, the unique minimizer of Program \eqref{eq:main eps} approaches $R_k$ in the limit, namely $\lim_{\epsilon\rightarrow 0} \|R_{k,\epsilon}-R_k\|_F = 0$. } 
\begin{equation}
\lim_{\epsilon\rightarrow 0} \| R_{k,\epsilon} - R_k\|_F = 0, \qquad k\in \I_l,
\label{eq:snipe out n st pnts pre}
\end{equation}
which, when restricted to $\{\epsilon_i\}_i$, reads as 
\begin{equation}
\lim_{i\rightarrow \infty} \| R_{k,\epsilon_i} - R_k\|_F = 0, \qquad k\in \I_l.
\label{eq:snipe out n st pnts pre restricted}
\end{equation}
By design, every $R_k$ has a spectral gap in the sense that  there exists $\tau>0$ such that 
\begin{equation}
\frac{\sigma_r(R_{k})}{\sigma_{r+1}(R_{k})} \ge 1+\tau, 
\label{eq:gap designed}
\end{equation}
for every $k$. Recall that $\h{\SU}_k$ is the span of  top $r$ left singular vectors of $R_k$.  An immediate consequence of \eqref{eq:gap designed} is that  $\h{\SU}_{k}$ is uniquely defined, namely there are no ties in the spectrum of $R_{k}$.  By \eqref{eq:snipe out n st pnts pre restricted}, there are no ties in the spectrum of $R_{k,\epsilon_i}$  as well for sufficiently large $i$, namely 
\begin{equation}
\lim_{i\rightarrow \infty} \frac{\sigma_r(R_{k,\epsilon_i})}{\sigma_{r+1}(R_{k,\epsilon_i})} \ge 1+\tau,\qquad k\in \I_l.
\label{eq:fixed k and limit eps}
\end{equation} 
By sending $l$ to infinity above, we find that 
\begin{align}
1+\tau & \le \lim_{l\rightarrow \infty}\lim_{i\rightarrow \infty} \frac{\sigma_r(R_{k_l,\epsilon_i})}{\sigma_{r+1}(R_{k_l,\epsilon_i})} 
\qquad \mbox{(see \eqref{eq:fixed k and limit eps})}
\nonumber\\
& = \lim_{l,i\rightarrow\infty}   \frac{\sigma_r(R_{k_l,\epsilon_i})}{\sigma_{r+1}(R_{k_l,\epsilon_i})}  
\qquad \mbox{(see \eqref{eq:final lim of R_k,i proof})} \nonumber\\
& = \frac{\sigma_r(R)}{\sigma_{r+1}(R)}.
\qquad \mbox{(see \eqref{eq:final lim of R_k,i proof})}
\label{eq:gap double hit}
\end{align}
Recall that $\h{\SU}_{k,\epsilon_i}$ is by definition the span of  leading $r$ left singular vectors of $R_{k,\epsilon_i}$. Likewise, let $\h{\SU}$  be the span of  leading $r$ left singular vectors of  $R$. An immediate consequence of the second line of \eqref{eq:gap double hit} is that   $\h{\SU}_{k_l,\epsilon_i}$ is uniquely defined, namely no ties in the spectrum of $R_{k_l,\epsilon_i}$ in the limit of $l,i\rightarrow\infty$.  The third line of \eqref{eq:gap double hit} similarly implies that $\h{\SU}$ is uniquely defined. 
Given the uniqueness of these subspaces, another implication of  \eqref{eq:final lim of R_k,i proof} is that 
\begin{equation}
\lim_{l,i\rightarrow\infty} d_{\GR}(\h{\SU}_{k_l,\epsilon_i},\h{\SU}) = 0,
\label{eq:final lim of S_ki proof}
\end{equation}
where $d_{\GR}$ is the metric on Grassmannian defined in \eqref{eq:err metric}. 
Lastly we show that $\alg$ in the limit produces copies of $(R,\h{\SU})$ on the interval $\I_{l,l'}$. This is done by simply noting that 
\begin{align}
\lim_{l\rightarrow\infty } \| R_{k_l} - R \|_F & 
= \lim_{l\rightarrow\infty } \lim_{i\rightarrow\infty } \| R_{k_l,\epsilon_i} - R \|_F
\qquad \mbox{(see \eqref{eq:snipe out n st pnts pre restricted})} \nonumber\\
& = \lim_{l,i\rightarrow\infty } \| R_{k_l,\epsilon_i} - R \|_F
\qquad \mbox{(see \eqref{eq:final lim of R_k,i proof})} \nonumber\\
& = 0, \qquad \mbox{(see \eqref{eq:final lim of R_k,i proof})}
\end{align}
\begin{align}
\lim_{l\rightarrow\infty } d_{\GR}(\h{\SU}_{k_l},\h{\SU})  & 
= \lim_{l\rightarrow\infty } \lim_{i\rightarrow\infty } d_{\GR}(\h{\SU}_{k_l,\epsilon_i},\h{\SU}) 
\qquad \mbox{(see \eqref{eq:snipe out n st pnts pre restricted})} \nonumber\\
& = \lim_{l,i\rightarrow\infty }  \lim_{i\rightarrow\infty } d_{\GR}(\h{\SU}_{k_l,\epsilon_i},\h{\SU})
\qquad \mbox{(see \eqref{eq:final lim of S_ki proof})} \nonumber\\
& = 0, \qquad \mbox{(see \eqref{eq:final lim of S_ki proof})}
\end{align}
which completes the proof of Lemma \ref{lem:exchange order of lims}.

\section{Proof of Lemma \ref{lem:final blockwise st} \label{sec:proof of final blockwise st}}

By restricting \eqref{eq:st leg 1} to $\{\epsilon_i\}_i$, we find that 
\begin{align}
0 & = \partial_{\U} f_{\h{\Omega}}(R_{k_l,\epsilon_i},\h{\SU}_{k,\epsilon_i})  
\qquad \mbox{(see \eqref{eq:st leg 1})}
\nonumber\\
& = \partial_{\U} f_{\Omega_{k_l}}(R_{k_l,\epsilon_i},\h{\SU}_{k,\epsilon_i}),
\qquad \mbox{(see \eqref{eq:repeated pattern})}
\label{eq:st leg 1 recall}
\end{align}
for  every $l,i$. By the joint continuity of $\partial_{\U}f_{\Omega_{k_l}}$ in Lemma \ref{lem:stationary pnts}, it follows that 
\begin{align}
0 & = \lim_{l,i\rightarrow \infty} \l\| \partial_{\U} f_{\Omega_{k_l}}(R_{k_l,\epsilon_i},\h{\SU}_{k_l,\epsilon_i})\r\|_F  \nonumber\\
& = \lim_{l\rightarrow\infty} \l\| \partial_{\U} f_{\Omega_{k_l}}(R,\h{\SU})\r\|_F, 
\qquad \mbox{(see Lemma \ref{lem:exchange order of lims})}
\end{align}
%
%
which establishes \eqref{eq:st cnd pre final 1}. To establish \eqref{eq:st cnd pre final 2}, we restrict \eqref{eq:st cnd 1.5} to $\{\epsilon_i\}_i$  and then send $i,l$ to infinity to find that 
\begin{align}
0 & = \lim_{i\rightarrow\infty} \lim_{l\rightarrow\infty}\| P_{\h{\Omega}}(R_{k_l,\epsilon_i}) - \h{Y} \|_F 
\qquad \mbox{(see \eqref{eq:st cnd 1.5})}
\nonumber\\
& = \lim_{i\rightarrow\infty} \lim_{l\rightarrow\infty}\| P_{{\Omega}_k}(R_{k_l,\epsilon_i}) - Y_{k_l} \|_F 
\qquad \mbox{(see (\ref{eq:repeated pattern},\ref{eq:limit behaviour}))} \nonumber\\
& = \lim_{l\rightarrow\infty}  \| P_{{\Omega}_{k_l}}(R) - {Y}_{k_l} \|_F. 
\qquad \mbox{(see   Lemma \ref{lem:exchange order of lims})}
\end{align}
To establish \eqref{eq:st cnd pre final 3}, we restrict  \eqref{eq:st of fixed Y cnd 2} to $\{\epsilon_i\}_i$ and then send $i$ to infinity to find that 
\begin{align}
0 & =\lim_{i\rightarrow \infty} \lim_{l\rightarrow\infty} \l\| \partial_X f_{\h{\Omega}}(R_{k_l,\epsilon_i},\h{\SU}_{k_l,\epsilon_i}) + \lambda_{k_l,\epsilon_i} (P_{\h{\Omega}}(R_{k_l,\epsilon_i})-\h{Y}) \r\|_F^2 \nonumber\\
& = \lim_{i\rightarrow \infty} \lim_{l\rightarrow\infty} 
\l\| P_{\h{\Omega}^C}\l( \partial_X f_{\h{\Omega}}(R_{k_l,\epsilon_i},\h{\SU}_{k_l,\epsilon_i}) \r)   \r\|_F^2 +  \lim_{i\rightarrow \infty} \lim_{l\rightarrow\infty}   \l\|  P_{\h{\Omega}}\l( \partial_X f_{\h{\Omega}}(R_{k_l,\epsilon},\h{\SU}_{k_l,\epsilon_i}) 
+ \lambda_{k,\epsilon_i} (R_{k_l,\epsilon_i}-\h{Y}) 
\r)  \r\|_F^2 \nonumber\\
& \ge  \lim_{i\rightarrow \infty} \lim_{l\rightarrow\infty} 
\l\| P_{\h{\Omega}^C}\l( \partial_X f_{\h{\Omega}}(R_{k_l,\epsilon_i},\h{\SU}_{k_l,\epsilon_i}) \r)   \r\|_F^2  \nonumber\\
&  =  \lim_{l\rightarrow\infty} \l\| P_{{\Omega}_{k_l}^C}\l( \partial_X f_{\Omega_{k_l}}(R,\h{\SU}) \r)   \r\|_F^2.
\qquad \mbox{(see \eqref{eq:repeated pattern} and Lemma \ref{lem:exchange order of lims})}
\end{align}
This completes the proof of Lemma \ref{lem:final blockwise st}.

\section{Proof of Lemma \ref{lem:whp main} \label{sec:Proof-of-Lemma whp main}}


Recall that by construction in Section \ref{sec:interp},  the rows of the coefficient matrix $Q_k\in\mathbb{R}^{b\times r}$ are independent copies of the random vector $q\in\mathbb{R}^r$.
Setting   $S_k=S\cdot Q_k^*\in\mathbb{R}^{n\times b}$,  we  observe in iteration $k$ each entry of  the data block $S_k$ independently with a probability of $p$, collect the observed entries in $Y_k\in\mathbb{R}^{n\times b}$, supported on the index set $\Omega_k\subseteq [1:n]\times [1:b]$. 
 We write this as $Y_k=P_{\Omega_k}(S_k)$, where the linear operator $P_{\Omega_k}$ retains the entries on the index set $\Omega_k$ and sets the rest to zero. Recall also that $Q_k$ is obtained by concatenating the coefficient vectors $\{q_t\}_{t=(k-1)b+1}^{kb} \subset \R^r$. To unburden the notation, we enumerate these vectors as $\{q_j\}_{j=1}^b$. Likewise, we use the indexing $\{s_j,y_j,\omega_j\}_{j=1}^b$ for the data vectors, incomplete data vectors, and their supports, respectively. 
%

Given the new incomplete block $Y_k$ at iteration $k$, we update our estimate of the true subspace $\SU$ from the old $\h{\SU}_{k-1}$ as follows. We calculate the random matrix
\begin{equation}
\mathbb{R}^{n\times b}\ni R_k := Y_k+ P_{\Omega_k^{C}}\left(O\left(Y_k\right)\right),\label{eq:main}
\end{equation}
where the linear operator $P_{\Omega_k^C}$ projects onto the complement of index set $\Omega_k$, and
\begin{equation}
O(Y_k)=\left[\cdots \h{S}_{k-1} (P_{\omega_{j}} \h{S}_{k-1})^{\dagger} y_j\cdots\right]\in\mathbb{R}^{n\times b}.
\label{eq:def of O recall}
\end{equation}
Above, $\h{S}_{k-1}\in\mathbb{R}^{n\times r}$ is an orthonormal basis for the $r$-dimensional subspace $\h{\SU}_{k-1}$. If $\sigma_r(R_k) < \sigma_{\min}$, reject this iteration, see Algorithm \ref{alg:Alg}. Otherwise, let $R_{k,r}$ denote a rank-$r$ truncation of $R_k$ obtained via SVD. Then our updated estimate is $\h{\SU}_{k}=\mbox{span}(R_{k,r})$. 

We condition on the subspace $\h{\SU}_{k-1}$ and the coefficient matrix $Q_k$ for now. To control the estimation error $d_{\GR}(\SU,\h{\SU}_{k})$, our strategy is to treat $R_k$ as a perturbed copy of $S_k=SQ_k^*$ and in turn treat $\h{\SU}_{k}=\mbox{span}(R_{k,r})$ as a perturbed copy of  $\SU=\mbox{span}(S_k)$.
Indeed, an application of the perturbation bound in Lemma \ref{lem:perturbation lemma}
yields that
\begin{align}
\Vert P_{\SU^{\perp}}P_{\h{\SU}_{k}}\Vert _{F}& \le\frac{\left\Vert P_{\SU^{\perp}}R_k\right\Vert _{F}}{\sigma_r\left(R_k\right)}. 
\label{eq:objective}
\end{align}
To control the numerator above, we begin with some
preparation. First, recalling the definition of $O(Y_k)$ from \eqref{eq:def of O recall}, we observe that
\begin{align}
O\left(Y_k\right) & = O\l( P_{\Omega_k}\l(S_k\r) \r)
\qquad \l(Y_k=P_{\Omega_k}\l( S_k\r)  \r)
\nonumber\\
 & =O(P_{\Omega_k}( P_{\h{\SU}_{k-1}}S_{k}) ) +O(P_{\Omega_k}( P_{\h{\SU}_{k-1}^{\perp}}S_k ))\qquad\l( \mbox{linearity of } O\r)\nonumber \\
 & =P_{\h{\SU}_{k-1}}S_k+O(P_{\Omega_k}(P_{\h{\SU}_{k-1}^{\perp}}S_k)).\qquad\mbox{(see (\ref{eq:def of O recall}))}\label{eq:basic decomp}
\end{align}
The above decomposition allows us to rewrite $R$ in \eqref{eq:main} as
\begin{align}
R_k & = Y_k+ P_{\Omega_k^{C}}\left(O\left(Y_k\right)\right)\qquad\mbox{(see (\ref{eq:main}))}\nonumber \\
& =  P_{\Omega_k}(S_k)+ P_{\Omega_k^{C}}\left(O\left(Y_k\right)\right)\qquad
\l( Y_k=P_{\Omega_k}(S_k) \r)\nonumber\\
 & =P_{\Omega_k}\left(S_{k}\right)+P_{\Omega_k^{C}}(P_{\h{\SU}_{k-1}}S_{k})+[P_{\Omega_k^{C}}\circ O\circ P_{\Omega_k}](P_{\h{\SU}_{k-1}^{\perp}}S_{k})
 \qquad\left(\mbox{see (\ref{eq:basic decomp}), }f\circ g(x):=f(g(x))\right)\nonumber \\
 & =S_k-P_{\Omega_k^{C}}\left(S_{k}\right)+P_{\Omega_k^{C}}(P_{\h{\SU}_{k-1}}S_{k})+[P_{\Omega_k^{C}}\circ  O \circ P_{\Omega_k} ](P_{\h{\SU}_{k-1}^{\perp}}S_{k})
 \nonumber \\
 & = S_{k}-P_{\Omega_k^{C}}(P_{\h{\SU}_{k-1}^{\perp}}S_{k})+[P_{\Omega_k^{C}}\circ O \circ P_{\Omega_k} ](P_{\h{\SU}_{k-1}^{\perp}}S_{k})
 \nonumber \\
 & =
 S_{k}-P_{\Omega_k^{C}}(P_{\h{\SU}_{k-1}^{\perp}}S_{k})
- \left[P_{\Omega_k}\circ O\circ P_{\Omega_k}\right](P_{\h{\SU}_{k-1}^{\perp}}S_{k}) 
 +\l[ O\circ P_{\Omega_k}  \r](P_{\h{\SU}_{k-1}^\perp}S_{k}).\label{eq:exp for R}
\end{align}
Since $P_{\SU^{\perp}}S_{k}=P_{\SU^{\perp}}SQ_k^{*}=0$, it immediately  follows that 
\[
P_{\SU^{\perp}}R_k= -
P_{\SU^\perp} \cdot
P_{\Omega_k^{C}}(P_{\h{\SU}_{k-1}^{\perp}}S_{k})
- P_{\SU^\perp} \cdot
[P_{\Omega_k}\circ O \circ P_{\Omega_k} ](P_{\h{\SU}_{k-1}^{\perp}}S_{k})
+
P_{\SU^\perp} \cdot
\l[ O \circ P_{\Omega_k} \r](P_{\h{\SU}_{k-1}^{\perp}}S_{k}).
\]
In particular, with an application of the triangle inequality above, we find that
\begin{align}
\left\Vert P_{\SU^{\perp}}R_k\right\Vert _{F} &
\le
\left\Vert P_{\Omega_k^{C}}(P_{\h{\SU}_{k-1}^{\perp}}S_{k})
+ [P_{\Omega_k}\circ O \circ P_{\Omega_k} ](P_{\h{\SU}_{k-1}^{\perp}}S_{k})
\right\Vert _{F}
+\left\Vert P_{\SU^{\perp}}\cdot \l[ O \circ P_{\Omega_k}\r] (P_{\h{\SU}_{k-1}^{\perp}}S_{k})\right\Vert _{F}.\label{eq:2nd decomp}
\end{align}
We proceed by controlling each norm on the right-hand side above using
the next two technical lemmas, proved in Appendices \ref{sec:Proof-of-Lemma t3} and \ref{sec:Proof-of-Lemma t2}, respectively. 
\begin{lem}
\label{lem:t3} It holds that 
\begin{equation}
\E \l[ \left\Vert P_{\Omega_k^{C}}(P_{\h{\SU}_{k-1}^{\perp}}S_{k})
+ \left[P_{\Omega_k}\circ O \circ P_{\Omega_k} \right](P_{\h{\SU}_{k-1}^{\perp}}S_{k})
\r\Vert _{F} \,|\, \h{\SU}_{k-1},Q_k\r] \le 
\sqrt{1-p} \cdot  \| P_{\SU^\perp} P_{\h{\SU}_{k-1}} \|_F \cdot   \|Q_k\|. 
\label{eq:contraction expectation}
\end{equation}
For fixed $\h{\SU}_{k-1}$ and $Q_k$, we also have  
\begin{equation}
\left\Vert P_{\Omega_k^{C}}(P_{\h{\SU}_{k-1}^{\perp}}S_{k})
+ \left[P_{\Omega_k}\circ O \circ P_{\Omega_k} \right](P_{\h{\SU}_{k-1}^{\perp}}S_k )
\right\Vert _{F}
\le \| P_{\SU^{\perp}}P_{\h{\SU}_{k-1}}\Vert _{F}\cdot\|Q_k\|.
\label{eq:contraction}
\end{equation}
and the  stronger bound
\begin{equation}
\left\Vert P_{\Omega_k^{C}}(P_{\h{\SU}_{k-1}^{\perp}}S_k )
+ \left[P_{\Omega_k}\circ O \circ P_{\Omega_k} \right](P_{\h{\SU}_{k-1}^{\perp}}S_k )
\right\Vert _{F}
\le 
\sqrt{1-p/2} \cdot \Vert  P_{\SU^{\perp}}P_{\h{\SU}_{k-1}}\Vert _{F}\cdot\|Q_k\|,
\label{eq:strong contraction}
\end{equation}
except with a probability of at
most 
\begin{equation}
\exp\l(-\frac{\Cr{fail}p^2nb}{\wt{\eta}_k}\r),
\end{equation}
where 
\begin{equation}
\wt{\eta}_k = \wt{\eta}(P_{\h{\SU}_{k-1}^\perp}S_k) = nb\cdot  \frac{\|P_{\h{\SU}_{k-1}^\perp} S_k\|^2_{\infty}}{\|P_{\h{\SU}_{k-1}^\perp}S_k\|_F^2}.
\end{equation}
\end{lem}
\begin{lem}
\label{lem:t2}For fixed $\o, Q_k$ and $\alpha \ge1$, it holds that
\begin{equation}
\left\Vert P_{\SU^{\perp}}\cdot \l[ O\circ P_{\Omega_k} \r](P_{\h{\SU}_{k-1}^{\perp}}S_k )\r\Vert _{F}\lesssim  \alpha \log b \sqrt{\frac{ \log n}{p}}   \cdot 
\Vert P_{\SU^{\perp}}P_{\h{\SU}_{k-1}}\Vert \cdot\Vert P_{\SU^{\perp}}P_{\h{\SU}_{k-1}}\Vert _{F}\cdot\|Q_k\|,
\label{eq:neighborhood}
\end{equation}
 except with a probability of at
most $b^{-C\alpha }$ and 
provided that $p\gtrsim\alpha^{2}\log^{2} b \log n\cdot\eta(\h{\SU}_{k-1})r/n$. 
\end{lem}
We next use Lemmas \ref{lem:t3} and \ref{lem:t2} to derive two  bounds for the numerator of \eqref{eq:objective}, the weaker bound holds with high probability but the stronger bound holds with  only some probability. More specifically, substituting \eqref{eq:contraction} and \eqref{eq:neighborhood}  into (\ref{eq:2nd decomp}) yields that
\begin{align}
&\left\Vert P_{\SU^{\perp}}R_k\right\Vert _{F}\nonumber\\
& \le
\left\Vert P_{\Omega_k^{C}}(P_{\h{\SU}_{k-1}^{\perp}}S_k )
+ \left[P_{\Omega_k}\circ O \circ P_{\Omega_k} \right](P_{\h{\SU}_{k-1}^{\perp}}S_k )
\right\Vert _{F}+\left\Vert P_{\SU^{\perp}}\cdot \l[ O \circ P_{\Omega_k}\r] (P_{\h{\SU}_{k-1}^{\perp}}S_k )\right\Vert _{F}
\qquad \mbox{(see \eqref{eq:2nd decomp})} \nonumber\\
&
\le 
\left(1+ C\alpha \log b \sqrt{\frac{ \log n}{p}} \| P_{\SU^\perp}P_{\h{\SU}_{k-1}}\|   
\right)
\Vert P_{\SU^{\perp}}P_{\h{\SU}_{k-1}}\Vert _{F}\|Q_k\|,\label{eq:numerator -1}
\end{align}
except with a probability of at most $b^{-C\alpha}$ and provided that $p\gtrsim \alpha^2 \log^2 b \log n \cdot \eta(\h{\SU}_{k-1})r/n$. For positive $c$ to be set later, let us further assume that 
\begin{equation}
\Vert P_{\SU^{\perp}}P_{\h{\SU}_{k-1}}\Vert
\le \Vert P_{\SU^{\perp}}P_{\h{\SU}_{k-1}}\Vert_F
\lesssim 
\frac{ p^{\frac{7}{2}}nb}{ c \alpha   \log b \sqrt{\log n} }.
\label{eq:initial 2}
\end{equation}
With an appropriate constant replacing $\lesssim$ above,  \eqref{eq:numerator -1} simplifies to 
\begin{align}
\left\Vert P_{\SU^{\perp}}R_k\right\Vert _{F}
\lesssim
\left(1+  \frac{p^3nb}{c }\right)\Vert P_{\SU^{\perp}}P_{\h{\SU}_{k-1}}\Vert _{F}\|Q_k\|.\label{eq:numerator 0}
\end{align}
A stronger bound is obtained by substituting (\ref{eq:strong contraction}, \ref{eq:neighborhood}) into \eqref{eq:2nd decomp}, namely  
\begin{align}
&\left\Vert P_{\SU^{\perp}}R_k\right\Vert _{F}\nonumber\\
& \le
\l\Vert P_{\Omega_k^{C}}(P_{\h{\SU}_{k-1}^{\perp}}S_k )
+ [P_{\Omega_k}\circ O \circ P_{\Omega_k} ](P_{\h{\SU}_{k-1}^{\perp}}S_k )
\r\Vert _{F}+\l\Vert P_{\SU^{\perp}}\cdot [ O \circ P_{\Omega_k}] (P_{\h{\SU}_{k-1}^{\perp}}S_k )\right\Vert _{F}
\qquad \mbox{(see \eqref{eq:2nd decomp})} \nonumber\\
&
\le \l( \sqrt{1-\frac{p}{2}} +  \frac{p^3nb}{c }\r) \| P_{\SU^\perp } P_{\h{\SU}_{k-1}} \|_F  \|Q_k\|
\qquad \mbox{(see (\ref{eq:strong contraction},\ref{eq:neighborhood},\ref{eq:initial 2}))}
\nonumber\\
& \le 
\l( {1-\frac{p}{4}} +  \frac{p^3nb}{c } \r)\| P_{\SU^\perp } P_{\h{\SU}_{k-1}} \|_F  \|Q_k\|, 
\label{eq:numerator}
\end{align}
provided that $p\gtrsim \alpha^2 \log^2 b \log n \cdot \eta(\h{\SU}_{k-1})r/n$
and except with a probability of at most 
$$
\exp\l(-\frac{\Cr{fail}p^2nb}{\wt{\eta}_k}\r)+b^{-C\alpha}.
$$
Note that \eqref{eq:numerator 0} and \eqref{eq:numerator} offer two alternative bounds for the numerator in the last line  \eqref{eq:objective}, which we will next use to complete the proof of Lemma \ref{lem:whp main}. 
%
%

Fix the subspace $\h{\SU}_{k-1}$ for now. Let $\ev_{k-1}$ be the event where $p\gtrsim \alpha^2 \log^2 b\log n \cdot \eta(\h{\SU}_{k-1})r/n$ and \eqref{eq:initial 2} holds. For $\nu\ge 1$ to be set later, let $\ev'_k$ be the event where 
\begin{equation}
\|Q_k\| \le \nu \cdot \sigma_{\min}.
\label{eq:bnd on Qk}
\end{equation}
 Conditioned on the event $\ev_{k-1}\cap \ev'_k$, we write that 
\begin{align}
 \Vert P_{\SU^{\perp}}P_{\n}\Vert _{F} 
 & \le\frac{\left\Vert P_{\SU^{\perp}}R_k\right\Vert _{F}}{\sigma_{r}\left(R_k\right)}
\qquad \mbox{(see \eqref{eq:objective})}
 \nonumber \\
 & \le \frac{\left\Vert P_{\SU^{\perp}}R_k\right\Vert _{F}}{\sigma_{\min}} \nonumber\\
 & \lesssim  \nu \left( 1+ \frac{p^3nb}{c } \right)\Vert P_{\SU^{\perp}}P_{\h{\SU}_{k-1}}\Vert _{F},
\qquad \mbox{(see (\ref{eq:numerator 0}))} 
 \label{eq:dragon  contraction only}
\end{align}
except with a probability of at most $b^{-C\alpha}$.
A  stronger bound is obtained from \eqref{eq:numerator}, namely 
\begin{align}
 \Vert P_{\SU^{\perp}}P_{\n}\Vert _{F}
 & \le\frac{\left\Vert P_{\SU^{\perp}}R_k\right\Vert _{F}}{\sigma_{r}\left(R_k\right)}
\qquad \mbox{(see \eqref{eq:objective})}
 \nonumber \\
 & \le  \nu \l(1- \frac{p}{4}+ \frac{p^3nb}{c }\r) \Vert P_{\SU^{\perp}}P_{\h{\SU}_{k-1}}\Vert _{F},
 \qquad \mbox{(see (\ref{eq:numerator}))} 
 \label{eq:dragon}
\end{align}
conditioned on the event $\ev_{k-1}\cap \ev'_k$ and except with a probability of at most 
\begin{equation}
\exp\l(-\frac{\Cr{fail}p^2nb}{\wt{\eta}_k}\r)+b^{-C\alpha}.
\end{equation}
This completes the proof of  the probabilistic claims in Lemma \ref{lem:whp main}, namely \eqref{eq:nearly contraction} and \eqref{eq:geo fast}. To complete the proof of Lemma \ref{lem:whp main}, we next derive a bound for the conditional expectation of $\|P_{\SU^\perp} P_{\h{\SU}_k}\|_F$. 
Let  $\ev''_k$ be the event where  $p\gtrsim \alpha^2\log^{2}b \log n \cdot \eta(\h{\SU}_{k-1})r / n$ and  
\begin{equation}
\label{eq:def of eventprime -1}
\left\Vert P_{\SU^{\perp}}\cdot \l[ O\circ P_\Omega \r]\left(P_{\h{\SU}_{k-1}^{\perp}}S_k \right)\right\Vert _{F}\lesssim  \frac{p^3nb}{c}
\Vert P_{\SU^{\perp}}P_{\h{\SU}_{k-1}}\Vert _{F}\|Q_k\|.
\end{equation}
In light of Lemma \ref{lem:t2}, we have that 
\begin{equation}
\Pr[\ev''_k|\h{\SU}_{k-1},\ev_k,\ev'_k] \ge 1- b^{-C\alpha}.
\label{eq:chance of evprime}
\end{equation} 
Using the law of total expectation, we now write that 
\begin{align}
& \E \l[\Vert P_{\SU^{\perp}}P_{\n}\Vert _{F}\, |\, \h{\SU}_{k-1},\ev_{k-1},\ev'_k\r] \nonumber\\
& 
= \E \l[ \Vert P_{\SU^{\perp}}P_{\n}\Vert _{F} \,|\, \h{\SU}_{k-1},\ev_{k-1},\ev'_k,\ev''_k\r] \cdot \Pr[\ev''_k|\h{\SU}_{k-1},\ev_{k-1},\ev'_k] \nonumber\\
&\qquad  +
 \E \l[ \Vert P_{\SU^{\perp}}P_{\n}\Vert _{F} \,|\, \h{\SU}_{k-1},\ev_{k-1},\ev'_k,\ev''^C_k\r] \cdot  \Pr[\ev''^C_k|\h{\SU}_{k-1},\ev_{k-1},\ev'_k] \nonumber\\
& \le  \E \l[ \Vert P_{\SU^{\perp}}P_{\n}\Vert _{F} \,|\, \h{\SU}_{k-1},\ev_{k-1},\ev'_k,\ev''_k\r]  + \sqrt{r}  \Pr[\ev''^C_k| \h{\SU}_{k-1},\ev_{k-1},\ev'_k]
\qquad \l(\h{\SU}_k \in \GR(n,r) \r) \nonumber\\
&  \le \E \l[ \Vert P_{\SU^{\perp}}P_{\n}\Vert _{F} \,|\, \h{\SU}_{k-1},\ev_{k-1},\ev'_k,\ev''_k\r] +  \sqrt{r} b^{-C\alpha} \qquad \mbox{(see \eqref{eq:chance of evprime})} \nonumber\\
& \le \E \l[ \frac{\l\| P_{\SU^\perp}R_k\r\|_F}{\sigma_r(R_k)} \,|\, \h{\SU}_{k-1},\ev_{k-1},\ev'_{k},\ev''_k\r] +   b^{-C\alpha} \qquad \l(\mbox{\eqref{eq:objective} and }b\ge r\r) \nonumber\\
& \le \sigma_{\min}^{-1} \E \l[ \l\| P_{\SU^\perp}R_k\r\|_F \,|\, \h{\SU}_{k-1},\ev_{k-1},\ev'_{k},\ev''_k\r] +   b^{-C\alpha}. 
\label{eq:exp decomp 1}
\end{align}
We next bound the remaining expectation above by writing that 
\begin{align}
& \E \l[ {\l\| P_{\SU^\perp}R_k\r\|_F} | \h{\SU}_{k-1},\ev_{k-1},\ev'_k,\ev''_k\r]  \nonumber\\
& \le \E\l[{\left\Vert P_{\Omega_k^{C}}(P_{\h{\SU}_{k-1}^{\perp}}S_k)
+ \left[P_{\Omega_k}\circ O \circ P_{\Omega_k} \right](P_{\h{\SU}_{k-1}^{\perp}}S_k )
\right\Vert _{F}+\left\Vert P_{\SU^{\perp}}\cdot \l[ O \circ P_\Omega\r] (P_{\h{\SU}_{k-1}^{\perp}}S_k )\right\Vert _{F}}\, |\, \h{\SU}_{k-1},\ev_{k-1},\ev'_k,\ev''_k\r] 
\qquad \mbox{(see \eqref{eq:2nd decomp})} \nonumber\\
& \le \E\l[ {\left\Vert P_{\Omega_k^{C}}(P_{\h{\SU}_{k-1}^{\perp}}S_k )
+ \left[P_{\Omega_k}\circ O \circ P_{\Omega_k} \right](P_{\h{\SU}_{k-1}^{\perp}}S_k )
\right\Vert _{F}+\frac{p^3nb}{c}\Vert P_{\SU^{\perp}}P_{\h{\SU}_{k-1}}\Vert _{F}\|Q_k\|}\, |\, \h{\SU}_{k-1},\ev_{k-1},\ev'_k,\ev''_k \r] 
\qquad  \mbox{(see (\ref{eq:def of eventprime -1}))} \nonumber\\
&  = \E\l[ {\E\l[  \left\Vert P_{\Omega_k^{C}}(P_{\h{\SU}_{k-1}^{\perp}}S_k )
+ \left[P_{\Omega_k}\circ O \circ P_{\Omega_k} \right](P_{\h{\SU}_{k-1}^{\perp}}S_k )
\right\Vert _{F} |\h{\SU}_{k-1},Q_k \r]+\frac{p^3nb}{c}\Vert P_{\SU^{\perp}}P_{\h{\SU}_{k-1}}\Vert _{F}\|Q_k\|}  | \h{\SU}_{k-1},\ev_{k-1},\ev'_k,\ev''_k\r] \nonumber\\
& \le 
\E\l[ { \sqrt{1-p} \| P_{\SU^\perp} P_{\h{\SU}_{k-1}}\|_F  \|Q_k\| +\frac{p^3nb}{c} \Vert P_{\SU^{\perp}}P_{\h{\SU}_{k-1}}\Vert _{F}\|Q_k\|} | \h{\SU}_{k-1},\ev_{k-1},\ev'_k,\ev''_k \r]
\qquad \mbox{(see \eqref{eq:contraction expectation})}
\nonumber\\
&= \E\l[  \nu\cdot \sigma_{\min}\l(  \sqrt{1-p}+\frac{p^3nb}{c} \r) \Vert P_{\SU^{\perp}}P_{\h{\SU}_{k-1}}\Vert _{F} | \h{\SU}_{k-1},\ev_{k-1},\ev'_k,\ev''_k\r]  
\qquad \mbox{(see \eqref{eq:bnd on Qk})}
\nonumber\\
& = \nu \cdot \sigma_{\min}\l(  \sqrt{1-p }+\frac{p^3nb}{c} \r) \Vert P_{\SU^{\perp}}P_{\h{\SU}_{k-1}}\Vert _{F} \nonumber\\
& \le \nu \cdot \sigma_{\min}\l(  1-\frac{p}{2}+\frac{p^3nb}{c} \r) \Vert P_{\SU^{\perp}}P_{\h{\SU}_{k-1}}\Vert _{F}.
\end{align}
Plugging the bound above back into \eqref{eq:exp decomp 1} yields that
\begin{align}
\E \l[ \Vert P_{\SU^{\perp}}P_{\n}\Vert _{F} | \h{\SU}_{k-1},\ev_{k-1},\ev'_{k} \r]&  \le \sigma_{\min}^{-1}
\E \l[ {\l\| P_{\SU^\perp}R_k\r\|_F} \,|\, \h{\SU}_{k-1},\ev_{k-1},\ev'_k,\ev''_k\r] +   b^{-C\alpha} \qquad \mbox{(see \eqref{eq:exp decomp 1})}\nonumber\\
& = \nu \l(1-\frac{p}{2}+\frac{p^3nb}{c}\r)\Vert P_{\SU^{\perp}}P_{\h{\SU}_{k-1}}\Vert _{F} + b^{-C\alpha},
\end{align}
which proves \eqref{eq:dragon in exp} and completes the proof of Lemma \ref{lem:whp main}.

\section{Proof of Lemma \ref{lem:t3} \label{sec:Proof-of-Lemma t3}}

Throughout, $\o,Q_k$ is fixed. Recalling the definition of operator $O(\cdot)$ from \eqref{eq:def of O recall},
we write that
\begin{align}
\l[ 
P_{\Omega_k} \circ O \circ P_{\Omega_k}
\r] (P_{\h{\SU}^\perp_{k-1}} S_k  ) & 
= \l[
\begin{array}{ccc}
\cdots & 
\l( P_{\omega_j} \h{S}_{k-1} \r) \l(P_{\omega_j} \h{S}_{k-1}\r)
\cdot P_{\omega_j} P_{\h{\SU}^\perp_{k-1}} S q_j
& \cdots 
\end{array}
\r]
\qquad \mbox{(see \eqref{eq:def of O recall})}
\nonumber\\
& = \l[
\begin{array}{ccc}
\cdots & 
\l( P_{\omega_j} \h{S}_{k-1} \r) \l(P_{\omega_j} \h{S}_{k-1} \r)
\cdot  P_{\h{\SU}^\perp_{k-1}} S q_j
& \cdots 
\end{array}
\r] \nonumber\\
& = 
\l[
\begin{array}{ccc}
\cdots & 
P_{\h{\SU}_{k-1,j}} 
\cdot  P_{\h{\SU}^\perp_{k-1}} S q_j
& \cdots 
\end{array}
\r] .
\qquad 
\l( \h{\SU}_{k-1,j} := \mbox{span}(P_{\omega_j} \h{S}_{k-1}) \r)
 \label{eq:brk 9.99}
\end{align}
Let also $\h{\SU}_{k-1,j}^C:=\mbox{span}(P_{\omega_j^C}\h{S}_{k-1})$. Note that $\h{S}_{k-1} = P_{\omega_j}\h{S}_{k-1}+P_{\omega_j^C}\h{S}_{k-1}$ and $(P_{\omega_j}\h{S}_{k-1})^* (P_{\omega_j^C}\h{S}_{k-1})=0$. Consequently, $\h{\SU}_{k-1,j} \perp \h{\SU}_{k-1,j}^C$  and then $P_{\h{\SU}_{k-1}}=P_{\h{\SU}_{k-1,j}}+P_{\h{\SU}_{k-1,j}^C}$. Put differently, $\h{\SU}_{k-1} = \h{\SU}_{k-1,j} \oplus \h{\SU}_{k-1,j}^C$. Using this decomposition, we simplify \eqref{eq:brk 9.99} as 
\begin{align}
\l[ 
P_{\Omega_k} \circ O \circ P_{\Omega_k}
\r] (P_{\h{\SU}^\perp_{k-1}} S_k  ) & 
= 
\l[
\begin{array}{ccc}
\cdots & 
P_{\h{\SU}_{k-1,j}} 
\cdot  P_{\h{\SU}^\perp_{k-1}} S q_j
& \cdots 
\end{array}
\r]  
\qquad \mbox{(see \eqref{eq:brk 9.99})}
\nonumber\\
& = 
\l[
\begin{array}{ccc}
\cdots & 
(P_{\h{\SU}_{k-1}} - P_{\h{\SU}_{k-1,j}^C}) 
\cdot   P_{\h{\SU}^\perp_{k-1}} S  q_j
& \cdots 
\end{array}
\r] 
\qquad \l( \h{\SU}_{k-1} = \h{\SU}_{k-1,j}\oplus \h{\SU}_{k-1,j}^C \r)
 \nonumber\\
& = -\l[
\begin{array}{ccc}
\cdots & 
P_{\h{\SU}_{k-1,j}^C}
\cdot   P_{\h{\SU}^\perp_{k-1}} S q_j
& \cdots 
\end{array}
\r].  
\end{align}
It immediately follows that 
\begin{align}
& P_{\Omega_k^C}(P_{\h{\SU}_{k-1}^\perp}S_k ) + \l[ 
P_{\Omega_k} \circ O \circ P_{\Omega_k}
\r] (P_{\h{\SU}^\perp_{k-1}} S_k  ) \nonumber\\
& = 
\l[
\begin{array}{ccc}
\cdots & 
P_{\omega_j^C}  \cdot P_{\h{\SU}_{k-1}^\perp} S q_j
-P_{\h{\SU}_{k-1,j}^C}
\cdot   P_{\h{\SU}^\perp_{k-1}} Sq_j
& \cdots 
\end{array}
\r] \nonumber\\
& =
\l[
\begin{array}{ccc}
\cdots & 
P_{\omega_j^C}  ( I_n 
-P_{\h{\SU}_{k-1,j}^C}
) P_{\omega_j^C}   \cdot P_{\h{\SU}^\perp_{k-1}} S q_j
& \cdots 
\end{array}
\r], 
\qquad \l( \h{\SU}_{k-1,j}^C = \mbox{span}(P_{\omega_j^C} \h{S}_{k-1}) \r)
\label{eq:columns of matrix of interest}
\end{align}
and consequently 
\begin{align}
& \l\|  P_{\Omega_k^C}(P_{\h{\SU}_{k-1}^\perp}S_k ) + \l[ 
P_{\Omega_k} \circ O \circ P_{\Omega_k}
\r] (P_{\h{\SU}^\perp_{k-1}} S_k  ) \r\|_F^2 \nonumber\\
& = \sum_{j=1}^b 
\l\| 
P_{\omega_j^C}  ( I_n 
-P_{\h{\SU}_{k-1,j}^C}
) P_{\omega_j^C}   \cdot P_{\h{\SU}^\perp_{k-1}} S q_j
\r\|_2^2 
\qquad \mbox{(see \eqref{eq:columns of matrix of interest})} \nonumber\\
& \le 
\sum_{j=1}^b 
\l\| 
 P_{\omega_j^C}   \cdot P_{\h{\SU}^\perp_{k-1}} S q_j
\r\|_2^2  \nonumber\\
& = \l\| P_{\Omega_k^C} (P_{\h{\SU}_{k-1}^\perp} S Q_k ) \r\|_F^2 \nonumber\\
& = \l\| P_{\Omega_k^C} (P_{\h{\SU}_{k-1}^\perp} S_k  ) \r\|_F^2.
\qquad \l(S_k  = SQ_k \r)
\label{eq:pre pre main bnd of 1st aux lemma}
\end{align}
Note that
\begin{align}
& \E \l[ \l\|  P_{\Omega_k^C}(P_{\h{\SU}_{k-1}^\perp}S_k ) + \l[ 
P_{\Omega_k} \circ O \circ P_{\Omega_k} 
\r] (P_{\h{\SU}^\perp_{k-1}} S_k  ) \r\|_F \,|\, \h{\SU}_{k-1},Q_k\r]  \nonumber\\
& \le \sqrt{ \E \l[ \l\|  P_{\Omega_k^C}(P_{\h{\SU}_{k-1}^\perp}S_k ) + \l[ 
P_{\Omega_k} \circ O \circ P_{\Omega_k}
\r] (P_{\SU^\perp_{k-1}} S_k  ) \r\|_F^2 \,|\, \h{\SU}_{k-1},Q_k\r]} 
\qquad \mbox{(Jensen's inequality)}
 \nonumber\\
& \le \sqrt{\E \l[ \l\| P_{\Omega_k^C} (P_{\h{\SU}_{k-1}^\perp} S_k  ) \r\|_F^2\, |\, \h{\SU}_{k-1},Q_k \r]} 
\qquad \mbox{(see \eqref{eq:pre pre main bnd of 1st aux lemma})}
\nonumber\\
& = \sqrt{\E\l[ \sum_{i,j} (1-\epsilon_{i,j}) \l| (P_{\h{\SU}_{k-1}^\perp} S_k )[i,j] \r|^2 \,|\, \h{\SU}_{k-1},Q_k\r]} \nonumber\\
& = \sqrt{1-p} \| P_{\h{\SU}_{k-1}^\perp} S_k \|_F \nonumber\\
& \le \sqrt{1-p} \| P_{\h{\SU}_{k-1}^\perp} S \|_F  \|Q_k\| 
\qquad \l( S_k  = SQ_k^*, \qquad \|AB\|_F \le \|A\|_F\cdot  \|B\|\r)
\label{eq:expectation Hoeffding}
\end{align} 
where $\{\epsilon_{i,j}\}_{i,j}\overset{\operatorname{ind.}}{\sim} \mbox{Bernoulli}(p)$. We thus proved the first claim in Lemma \ref{lem:t3}, namely \eqref{eq:contraction expectation}.
Then note that 
\begin{align}
& \l\|  P_{\Omega_k^C}(P_{\h{\SU}_{k-1}^\perp}S_k ) + \l[ 
P_{\Omega_k} \circ O \circ P_{\Omega_k}
\r] (P_{\SU^\perp_{k-1}} S_k  ) \r\|_F^2 \nonumber\\
& \le  \l\| P_{\Omega_k^C} (P_{\h{\SU}_{k-1}^\perp} S_k  ) \r\|_F^2
\qquad \mbox{(see \eqref{eq:pre pre main bnd of 1st aux lemma})}
\nonumber\\
& \le\| P_{\h{\SU}_{k-1}^\perp} S_k \|_F^2 \nonumber\\
& \le \| P_{\h{\SU}_{k-1}^\perp} S\|_F^2 \|Q_k\|^2,
\qquad \l( S_k  = SQ_k \r)
\label{eq:pre main bnd of 1st aux lemma}
\end{align}
which proves the second claim, namely \eqref{eq:contraction}. In fact, with some probability, a  stronger bound can be derived by controlling the deviation from the expectation in \eqref{eq:expectation Hoeffding} using the Hoeffding inequality (Lemma \ref{lem:Bernie for fro}). With $\alpha = \frac{p}{2}\|P_{\h{\SU}_{k-1}^\perp} S_k \|_F^2$ in  Lemma \ref{lem:Bernie for fro} and recalling that $\o,Q_k$ are fixed for now, we find that 
\begin{align}
\l\| P_{\Omega_k^C} (P_{\h{\SU}_{k-1}^\perp}S_k ) \r\|_F^2 & \le \E \l[\l\| P_{\Omega_k^C} (P_{\h{\SU}_{k-1}^\perp}S_k ) \r\|_F^2 \,|\, \h{\SU}_{k-1},Q_k\r]+ \alpha \nonumber\\
& = (1-p/2) \| P_{\h{\SU}_{k-1}^\perp}S_k  \|_F^2 
\qquad \mbox{(see \eqref{eq:expectation Hoeffding})}
\nonumber\\
& = (1-p/2) \| P_{\h{\SU}_{k-1}^\perp}S Q_k \|_F^2 
\qquad \l( S_k  = SQ_k\r)
\nonumber\\
& \le (1-p/2) \| P_{\h{\SU}_{k-1}^\perp}S \|_F^2  \l\| Q_k \r\|^2, 
\label{eq:Hoeff applied}
\end{align}
except with a probability of at most 
\begin{align}
\exp\l({-\frac{\Cr{fail}p^2  \|P_{\h{\SU}_{k-1}^\perp} S_k \|_F^4}{\sum_{i,j} | ( P_{\h{\SU}_{k-1}^\perp} S_k )[i,j] |^4  }}\r) &  \le \exp\l( {-\frac{\Cr{fail}p^2  \|P_{\h{\SU}_{k-1}^\perp} S_k \|_F^4}{ \| P_{\h{\SU}_{k-1}^\perp} S_k \|_{\infty}^2\cdot \| P_{\h{\SU}_{k-1}^\perp} S_k \|_{F}^2 } }\r) 
\nonumber\\
& = \exp\l( {-\frac{\Cr{fail}p^2  \|P_{\h{\SU}_{k-1}^\perp} S_k \|_F^2}{ \| P_{\h{\SU}_{k-1}^\perp} S_k \|_{\infty}^2 }}\r) \nonumber\\
& =: \exp\l(\frac{-\Cr{fail}p^2nb}{\wt{\eta}(P_{\h{\SU}_{k-1}^\perp}S_k )}\r),
\end{align}
where $\|A\|_{\infty}$ returns the largest entry of $A$ in magnitude. 
Substituting \eqref{eq:Hoeff applied} back into \eqref{eq:pre main bnd of 1st aux lemma} yields that 
\begin{align}
& \l\|  P_{\Omega_k^C}(P_{\h{\SU}_{k-1}^\perp}S_k ) + \l[ 
P_{\Omega_k} \circ O \circ P_{\Omega_k}
\r] (P_{\h{\SU}^\perp_{k-1}} S_k  ) \r\|_F
\nonumber\\
& 
\le \l\| P_{\Omega_k^C} (P_{\h{\SU}_{k-1}^\perp} S_k ) \r\|_F
\qquad \mbox{(see \eqref{eq:pre main bnd of 1st aux lemma})} \nonumber\\
&  \le  \sqrt{1-p/2} \cdot \Vert P_{\h{\SU}_{k-1}^{\perp}}S\Vert _{F}\|Q_k\|,
\end{align}
which proves the last claim in Lemma \ref{lem:t3}, namely \eqref{eq:strong contraction}.

\section{Proof of Lemma \ref{lem:t2}\label{sec:Proof-of-Lemma t2}}

We fix $\o,Q_k$ throughout. We begin by bounding the target quantity as 
\begin{align}
\left\Vert P_{\SU^{\perp}}\cdot \l[ O \circ P_{\Omega_k}  \r](P_{\h{\SU}_{k-1}^{\perp}}S_k )\right\Vert _{F} & =\left\Vert P_{\SU^{\perp}}P_{\h{\SU}_{k-1}}\cdot \l[ O\circ P_{\Omega_k} \r](P_{\h{\SU}_{k-1}^{\perp}}S_k )\right\Vert _{F}\qquad\mbox{(see (\ref{eq:def of O recall}))}\nonumber \\
 & \le\Vert P_{\SU^{\perp}}P_{\h{\SU}_{k-1}}\Vert \cdot\left\Vert \l[O \circ P_{\Omega_k} \r](P_{\h{\SU}_{k-1}^{\perp}}S_k )\right\Vert _{F}\qquad\left(\|AB\|_{F}\le\|A\|\cdot\|B\|_{F}\right).\label{eq:brk .39}
\end{align}
We bound the random norm in the last line above in Appendix \ref{sec:Proof-of-Lemma O}.
\begin{lem}
\label{lem:O norm}For $\alpha\ge1$ and except with a probability of
at most $b^{-C\alpha}$, it holds that
\begin{equation}
\left\Vert \l[ O\circ P_{\Omega_k} \r](P_{\h{\SU}_{k-1}^{\perp}}S_k )\right\Vert _{F}\lesssim { \alpha\log b \sqrt{\frac{{\log n}}{{p}}}} \Vert P_{\SU^{\perp}}P_{\h{\SU}_{k-1}}\Vert _{F}\cdot\|Q_k\|,
\label{eq:O norm whp}
\end{equation}
provided that $p\gtrsim\alpha^{2}\log^{2}b \log n\cdot\eta(\h{\SU}_{k-1})r/n$. 
\end{lem}
In light of the above lemma, we conclude
that
\begin{align}
\left\Vert P_{\SU^{\perp}}\cdot \l[O \circ P_{\Omega_k} \r](P_{\h{\SU}_{k-1}^{\perp}}S_k )\right\Vert _{F}
& \le\Vert P_{\SU^{\perp}}P_{\h{\SU}_{k-1}}\Vert \cdot\left\Vert \l[ O \circ P_{\Omega_k} \r](P_{\h{\SU}_{k-1}^{\perp}}S_k )\right\Vert _{F}
\qquad \mbox{(see \eqref{eq:brk .39})}
\nonumber\\
& \lesssim {\alpha\log b \sqrt{\frac{ \log n}{p}}}\Vert P_{\SU^{\perp}}P_{\h{\SU}_{k-1}}\Vert \cdot \Vert P_{\SU^{\perp}}P_{\h{\SU}_{k-1}}\Vert _{F}\cdot\|Q_k\|,
\qquad \mbox{(see \eqref{eq:O norm whp})} 
\label{eq:high prob bnd }
\end{align}
except with a probability of at most $b^{-C\alpha}$ and provided
that $p\gtrsim\alpha^2 \log^2 b \log n \cdot\eta(\h{\SU}_{k-1})r/n$. 
This completes the proof of Lemma \ref{lem:t2}.

\section{Proof of Lemma \ref{lem:O norm}\label{sec:Proof-of-Lemma O}}

Using the definition of operator $O$ in \eqref{eq:def of O recall}, we write that
\begin{align}
\left\Vert \l[O \circ P_{\Omega_k}\r] (P_{\h{\SU}_{k-1}^{\perp}}S_k )\right\Vert _{F}^{2} &
 =\sum_{j=1}^{b}\left\Vert \h{S}_{k-1}(P_{\omega_{j}}\h{S}_{k-1})^{\dagger}P_{\omega_j} \cdot P_{\h{\SU}_{k-1}^{\perp}}Sq_j \right\Vert _{2}^{2}
 \qquad\left(\mbox{(\ref{eq:def of O recall}) and }S_k =SQ_k\right)\nonumber \\
 & =\sum_{j=1}^{b}\left\Vert \h{S}_{k-1}(P_{\omega_{j}}\h{S}_{k-1})^{\dagger} P_{\h{\SU}_{k-1}^{\perp}}Sq_j\right\Vert _{2}^{2}
 \nonumber \\
 & =\sum_{j=1}^{b}\left\Vert (P_{\omega_{j}}\h{S}_{k-1})^{\dagger}P_{\h{\SU}_{k-1}^{\perp}}Sq_j \right\Vert _{2}^{2}.\qquad\left(\h{S}_{k-1}^{*}\h{S}_{k-1}=I_{r}\right)\label{eq:brk .4}
\end{align}
For fixed $j\in [1:b]$, consider the summand in the last line above:
\begin{align}
\left\Vert (P_{\omega_{j}}\h{S}_{k-1})^{\dagger} P_{\h{\SU}_{k-1}^{\perp}}S\cdot q_j \right\Vert _{2} 
& \le\left\Vert (P_{\omega_{j}}\h{S}_{k-1})^{\dagger}P_{\h{\SU}_{k-1}^{\perp}}\right\Vert \cdot \Vert P_{\h{\SU}_{k-1}^{\perp}}Sq_j\Vert _{2}\nonumber \\
 & =\left\Vert (P_{\omega_{j}}\h{S}_{k-1})^{\dagger}\h{S}_{k-1}^{\perp}\right\Vert \cdot\Vert P_{\h{\SU}_{k-1}^{\perp}}Sq_j\Vert _{2}\qquad\left(P_{\h{\SU}_{k-1}^{\perp}}= \h{S}_{k-1}^{\perp} (\h{S}_{k-1}^{\perp})^* \right) \nonumber \\
 & =:\Vert \widehat{Z}_{j} \Vert \cdot \Vert P_{\h{\SU}_{k-1}^{\perp}}S q_j\Vert _{2}.\label{eq:brk .5}
\end{align}
Above,  $\h{S}_{k-1}^\perp$ is as usual an orthonormal basis for the subspace $\h{\SU}_{k-1}^\perp$. We can now revisit (\ref{eq:brk .4}) and write that
\begin{align}
\left\Vert O (P_{\h{\SU}_{k-1}^{\perp}}S_k )\right\Vert _{F}^{2} & =\sum_{j=1}^{b}\left\Vert (P_{\omega_{j}}\h{S}_{k-1} )^{\dagger} P_{\h{\SU}_{k-1}^{\perp}}Sq_j )\right\Vert _{2}^{2}\qquad\mbox{(see (\ref{eq:brk .4}))}\nonumber \\
 & \le\max_{j}\Vert \widehat{Z}_{j}\Vert ^{2}\cdot\sum_{j=1}^{b}\Vert P_{\h{\SU}_{k-1}^{\perp}}Sq_j\Vert _{2}^{2}\qquad\mbox{(see (\ref{eq:brk .5}))}\nonumber \\
 & =\max_{j}\Vert \widehat{Z}_{j}\Vert ^{2}\cdot \Vert P_{\h{\SU}_{k-1}^{\perp}}SQ_k \Vert _{F}^{2}\nonumber \\
 & \le\max_{j}\Vert \widehat{Z}_{j}\Vert ^{2}\cdot \Vert P_{\h{\SU}_{k-1}^{\perp}}S \Vert _{F}^{2}\|Q_k\|^{2}.\qquad\left(\|AB\|_{F}\le\|A\|_{F}\cdot\|B\|\right)\label{eq:brk .51}
\end{align}
It remains to control the maximum in the last line above. We first focus on controlling $\|\widehat{Z}_{j}\|$ for
fixed $j\in[1:b]$. Observe
that $\widehat{Z}_{j}$ is a solution of the least-squares problem
\[
\min_{Z\in\mathbb{R}^{n\times(n-r)}}\left\Vert \h{S}_{k-1}^{\perp}-(P_{\omega_{j}}\h{S}_{k-1})Z\right\Vert _{F}^{2},
\]
and therefore satisfies the \emph{normal equation}
\[
(P_{\omega_{j}}\h{S}_{k-1})^{*}\left((P_{\omega_{j}}\h{S}_{k-1})\widehat{Z}_{j}-\h{S}_{k-1}^{\perp}\right)=0,
\]
which is itself equivalent to
\begin{equation}
(\h{S}_{k-1}^{*}P_{\omega_{j}}\h{S}_{k-1})\widehat{Z}_{j}=\h{S}_{k-1}^{*}P_{\omega_{j}}\h{S}_{k-1}^{\perp}.\qquad\left(P_{\omega_{j}}^{2}=P_{\omega_{j}}\right)\label{eq:norma}
\end{equation}
In fact, since
\[
\mathbb{E}\left[\h{S}_{k-1}^{*}P_{\omega_{j}}\h{S}_{k-1}^{\perp}\right]=p\cdot \h{S}_{k-1}^{*}\h{S}_{k-1}^{\perp}=0,
\]
\[
\mathbb{E}\left[\h{S}_{k-1}^{*}P_{\omega_{j}}\h{S}_{k-1}\right]=p\cdot I_{r},
\qquad \l( \h{S}_{k-1}^* \h{S}_{k-1} = I_r \r)
\]
we can rewrite (\ref{eq:norma}) as
\[
\left(\h{S}_{k-1}^{*}P_{\omega_{j}}\h{S}_{k-1}-\mathbb{E}\left[\h{S}_{k-1}^{*}P_{\omega_{j}}\h{S}_{k-1}\right]\right)\widehat{Z}_{j}+p\cdot\widehat{Z}_{j}=\h{S}_{k-1}^{T}P_{\omega_{j}}\h{S}_{k-1}^{\perp}-\mathbb{E}\left[\h{S}_{k-1}^{T}P_{\omega_{j}}\h{S}_{k-1}^{\perp}\right].
\]
An application of the triangle inequality above immediately implies that
\begin{align}
p \Vert \widehat{Z}_{j} \Vert  & \le\left\Vert \h{S}_{k-1}^{*}P_{\omega_{j}}\h{S}_{k-1}-\mathbb{E}\left[\h{S}_{k-1}^{*}P_{\omega_{j}}\h{S}_{k-1}\right]\right\Vert \cdot\Vert \widehat{Z}_{j}\Vert +\left\Vert \h{S}_{k-1}^{*}P_{\omega_{j}}\h{S}_{k-1}^{\perp}-\mathbb{E}\left[\h{S}_{k-1}^{*}P_{\omega_{j}}\h{S}_{k-1}^{\perp}\right]\right\Vert .\label{eq:key}
\end{align}
To control $\|\widehat{Z}_{j}\|$, we therefore need to derive large
devation bounds for the two remaining norms on the right-hand side
above. For the first spectral norm, we write that
\begin{equation}
\left\Vert \h{S}_{k-1}^{*}P_{\omega_{j}}\h{S}_{k-1}-\mathbb{E}\left[\h{S}_{k-1}^{*}P_{\omega_{j}}\h{S}_{k-1}\right]\right\Vert =\left\Vert \sum_{i}\l(\epsilon_i-p \r)\cdot \h{S}_{k-1}^{*}E_{i,i}\h{S}_{k-1}\right\Vert =:\left\Vert \sum_{i}A_{i}\right\Vert ,
\label{eq:def of As}
\end{equation}
where $\{\epsilon_i\}_i \overset{\operatorname{ind.}}{\sim} \mbox{Bernoulli}(p)$ and  $E_{i,i}\in\mathbb{R}^{n\times n}$ is the $[i,i]$th canonical
matrix. Furthermore, $\{A_{i}\}_i\subset\mathbb{R}^{r\times r}$ above are independent
and zero-mean random matrices. To apply the Bernstein inequality (Lemma \ref{lem:Bernie for spec}), we first compute the parameter $\beta$ as
\begin{align}
\left\Vert A_{i}\right\Vert  & =\left\Vert \l(\epsilon_i-p \r)\cdot S_{k-1}^*E_{i,i}S_{k-1} \right\Vert
\qquad \mbox{(see \eqref{eq:def of As})}
\nonumber\\
 & \le \Vert \h{S}_{k-1}^{*}E_{i,i}\h{S}_{k-1}\Vert  \qquad
\l( \epsilon_{i}\in \{0,1\} \r)
 \nonumber\\
 & = \Vert \h{S}_{k-1}[i,:]\Vert _{2}^{2}
 \nonumber\\
& \le \frac{\eta(\h{\SU}_{k-1})r}{n}=:\beta.\qquad \qquad \mbox{(see \eqref{eq:def of coh})}
\label{eq:beta for Ais}
\end{align}
To compute the weak variance $\sigma$, we write that
\begin{align}
\l\| \E\l[  \sum_{i} A_i^2 \r]  \r\|
&
= \l\| \sum_i \mathbb{E}\l[\l(\epsilon_i-p\r)^2 \r]  ( \h{S}_{k-1}^* E_{i,i}\h{S}_{k-1})^2 \r\|
\qquad \mbox{(see \eqref{eq:def of As})}
\nonumber\\
& = \l\|  \sum_i p(1-p) ( \h{S}_{k-1}^* E_{i,i} \h{S}_{k-1})^2  \r\|
\qquad \l(\epsilon_i \sim \mbox{Bernoulli}(p)  \r) \nonumber\\
& \le p \l\|\sum_i (\h{S}_{k-1}^* E_{i,i}\h{S}_{k-1})^2  \r\| \nonumber\\
& = p \l\| \sum_i \h{S}_{k-1}^* E_{i,i} \h{S}_{k-1}\h{S}_{k-1}^* E_{i,i} \h{S}_{k-1} \r\| \nonumber\\
& = p \l\| \sum_i \|\h{S}_{k-1}[i,:] \|_2^2 \cdot \h{S}_{k-1}^* E_{i,i} \h{S}_{k-1} \r\|
\nonumber\\
& \le p \cdot \max_i \| \h{S}_{k-1}[i,:] \|_2^2 \cdot \l\| \sum_i \h{S}_{k-1}^* E_{i,i} \h{S}_{k-1}  \r\|  \nonumber\\
& = p \l\| \sum_i \h{S}_{k-1}^* E_{i,i}   \h{S}_{k-1} \r\|
\nonumber\\
& = p \cdot \frac{\eta( \h{\SU}_{k-1})r}{n} \cdot \l\| \sum_i \h{S}_{k-1}^* E_{i,i} \h{S}_{k-1}\r\|
\qquad \mbox{(see \eqref{eq:def of coh})}
\nonumber\\
& = p \cdot \frac{\eta( \h{\SU}_{k-1})r}{n}
\qquad \l( \sum_i E_{i,i}=I_n,\,\, \h{S}_{k-1}^* \h{S}_{k-1} = I_r \r)
\nonumber\\
& =: \sigma^2.
 \label{eq:sigma for Ais}
\end{align}
It also follows that
\begin{align}
\max\left(\log r\cdot \beta,\sqrt{\log r}\cdot\sigma\right)&
=\max\left(\frac{\log r\cdot\eta(\h{\SU}_{k-1})r}{n},\sqrt{\frac{{\log r\cdot p\cdot \eta(\h{\SU}_{k-1})r}}{n}}\right)
\qquad \mbox{(see \eqref{eq:beta for Ais} and  \eqref{eq:sigma for Ais})} \nonumber\\
&
\le\sqrt{\frac{{\log r\cdot p \cdot \eta(\h{\SU}_{k-1})r}}{n}}.\qquad\left(\mbox{if }p \ge\frac{\log r \cdot \eta(\h{\SU}_{k-1})r}{n}\right)
\label{eq:max b sigma Ais}
\end{align}
As a result, for $\alpha\ge1$ and except with a probability of at
most $e^{-C\alpha}$, it holds that
\begin{align}
\left\Vert \h{S}_{k-1}^{*}P_{\omega_{j}}\h{S}_{k-1}-\mathbb{E}\left[\h{S}_{k-1}^{*}P_{\omega_{j}}\h{S}_{k-1}\right]\right\Vert & \lesssim\alpha\max\left(\log r\cdot \beta,\sqrt{\log r}\cdot\sigma\right)
\qquad \mbox{(see Lemma \ref{lem:Bernie for spec})}
\nonumber\\
& \le \alpha \sqrt{\frac{{\log r \cdot p \cdot \eta(\h{\SU}_{k-1})r}}{n}}.\label{eq:large dev 1}
\end{align}
On the other hand, in order to apply the Bernstein inequality to the
second spectral norm in (\ref{eq:key}), we write that
\begin{equation}
\left\Vert \h{S}_{k-1}^{*}P_{\omega_{j}}\h{S}_{k-1}^{\perp}-\mathbb{E}\left[\h{S}_{k-1}^{*}P_{\omega_{j}}\h{S}_{k-1}^{\perp}\right]\right\Vert  =\left\Vert \sum_{i} \l(\epsilon_i-p \r) \h{S}_{k-1}^{*}E_{i,i}\h{S}_{k-1}^{\perp}\right\Vert =:\left\Vert \sum_{i}A_{i}\right\Vert ,
\label{eq:def of other As}
\end{equation}
where $\{\epsilon_i \}_i\overset{\operatorname{ind.}}{\sim} \mbox{Bernoulli}(p)$, $E_{i,i}\in \mathbb{R}^{n\times n}$ is the $i$th  canonical  matrix, and  $\{A_{i}\}_i\subset\mathbb{R}^{r\times(n-r)}$ are zero-mean and
independent random matrices. To compute the parameter $\beta$ here, we write that
\begin{align}
\left\Vert A_{i}\right\Vert  & =\left\Vert
\l(\epsilon_i-p  \r)
 \h{S}_{k-1}^{*}E_{i,i}\h{S}_{k-1}^{\perp}\right\Vert
\qquad \mbox{(see \eqref{eq:def of other As})}
  \nonumber \\
  & \le \|  \h{S}_{k-1}^* E_{i,i} \h{S}_{k-1}^\perp \| \qquad
  \l( \epsilon_i \in\{0,1\} \r) \nonumber\\
 & \le\Vert \h{S}_{k-1}^{*}E_{i,i}\Vert
\qquad \l( ( \h{S}_{k-1}^\perp)^* \h{S}_{k-1}^\perp = I_{n-r} \r)
\nonumber\\
 & =\Vert \h{S}_{k-1}[i,:]\Vert _{2}\nonumber\\
 & \le\sqrt{\frac{\eta(\h{\SU}_{k-1})r}{n}}=:\beta.
 \qquad \mbox{(see \eqref{eq:def of coh})}
\label{eq:b Ais 2}
\end{align}
To compute the weak variance $\sigma$, we notice that
\begin{align}
\l\| \E\l[ \sum_{i} A_i A_i^* \r] \r\| &
= \l\|   \sum_i \mathbb{E}\l[ \l( \epsilon_i-p\r)^2 \r] \h{S}_{k-1}^* E_{i,i} \h{S}_{k-1}^\perp ( \h{S}_{k-1}^\perp )^* E_{i,i} \h{S}_{k-1} \r\|
\qquad \mbox{(see \eqref{eq:def of other As})}
\nonumber\\
& = \l\| \sum_i p(1-p) \cdot \h{S}_{k-1}^* E_{i,i} \h{S}_{k-1}^\perp ( \h{S}_{k-1}^\perp )^* E_{i,i} \h{S}_{k-1} \r\|
\qquad \l( \epsilon_i \sim \mbox{Benoulli}(p) \r) \nonumber\\
& \le p \l\| \sum_i \h{S}_{k-1}^* E_{i,i} \h{S}_{k-1}^\perp ( \h{S}_{k-1}^\perp )^* E_{i,i} \h{S}_{k-1}  \r\| \nonumber\\
& \le p
\l\| \sum_i \h{S}_{k-1}^* E_{i,i}  E_{i,i} \h{S}_{k-1}  \r\|
\qquad \l(  \h{S}_{k-1}^\perp ( \h{S}_{k-1}^\perp )^* \preceq I_n  \r)
\nonumber\\
& = p \l\| \sum_i \h{S}_{k-1}^* E_{i,i}  \h{S}_{k-1}  \r\|  \nonumber\\
& = p. \qquad \l( \sum_i E_{i,i}=I_n,\,\, \h{S}_{k-1}^* \h{S}_{k-1} = I_r \r)
 \label{eq:pre sigma Ais leg 1}
\end{align}
In a similar fashion, we find that
\begin{align}
\left\Vert \mathbb{E}\l[\sum_{i}A_{i}^{*}A_{i}\r]\right\Vert
& \le p \l\|\sum_i (\h{S}_{k-1}^\perp )^* E_{i,i} \h{S}_{k-1} \h{S}_{k-1}^* E_{i,i} \h{S}_{k-1}^\perp   \r\|
\nonumber\\
&= p \l\| \sum_i \| \h{S}_{k-1}[i,:]\|_2^2 \cdot   (\h{S}_{k-1}^\perp )^* E_{i,i} \h{S}_{k-1}^\perp \r\| \nonumber\\
& \le p \cdot \max_i \|\h{S}_{k-1}[i,:]  \|_2^2 \cdot \l \| \sum_i (\h{S}_{k-1}^\perp )^* E_{i,i} \h{S}_{k-1}^\perp \r\| \nonumber\\
& = p \cdot \max_i \|\h{S}_{k-1}[i,:]  \|_2^2
\qquad \l( \sum_i E_{i,i}=I_n,\,\, (\h{S}_{k-1}^\perp )^* \h{S}_{k-1}^\perp = I_{n-r} \r) \nonumber\\
& = p \cdot \frac{\eta(\h{\SU}_{k-1})r}{n},
\qquad \mbox{(see \eqref{eq:def of coh})}
 \label{eq:pre sigma Ais leg 2}
\end{align}
and finally
\begin{align}
\sigma
& =
\max\l( \l\| \E \sum_{i} A_i A_i^*\r\|,
\l\| \E \sum_{i} A_i^*A_i  \r\|
   \r)
\nonumber\\
& =\max\left(\sqrt{p},\sqrt{p}\cdot\sqrt{\frac{\eta(\h{\SU}_{k-1})r}{n}}\right)
\qquad \mbox{(see \eqref{eq:pre sigma Ais leg 1} and \eqref{eq:pre sigma Ais leg 2})}
\nonumber\\
&
=\sqrt{p}.\qquad\left(\eta(\h{\SU}_{k-1})\le\frac{n}{r}\right)
\label{eq:sigma Ais 2}
\end{align}
We now compute
\begin{align}
\max\left(\log n\cdot \beta,\sqrt{\log n}\cdot\sigma\right)
& =\max\left(\log n\sqrt{\frac{\eta(\h{\SU}_{k-1}}{n}},\sqrt{{\log n\cdot p}}\right)
\qquad \mbox{(see \eqref{eq:b Ais 2} and  \eqref{eq:sigma Ais 2})}
\nonumber\\
& =\sqrt{{\log n\cdot p}}.\qquad\left(\mbox{if }p\ge \frac{\log n\cdot\eta(\h{\SU}_{k-1})r}{n}\right)
\label{eq:other m cnd}
\end{align}
Therefore, for $\alpha\ge1$ and except with a probability of at most
$e^{-C\alpha}$, it holds that
\begin{align}
\left\Vert \h{S}_{k-1}^{*}P_{\omega_{j}}\h{S}_{k-1}^{\perp}-\mathbb{E}\left[\h{S}_{k-1}^{*}P_{\omega_{j}}\h{S}_{k-1}^{\perp}\right]\right\Vert &
\lesssim \alpha\max\left(\log n\cdot \beta,\sqrt{\log n}\cdot\sigma\right)
\qquad \mbox{(see Lemma \ref{lem:Bernie for spec})}
\nonumber\\
& =\alpha\sqrt{{\log n\cdot p}}.\label{eq:large dev 2}
\end{align}
Overall, by substituting the large deviation bounds (\ref{eq:large dev 1})
and (\ref{eq:large dev 2}) into (\ref{eq:key}), we find that
\begin{align*}
p \Vert \widehat{Z}_{j} \Vert  & \le\left\Vert \h{S}_{k-1}^{*}P_{\omega_{j}}\h{S}_{k-1}-\mathbb{E}\left[\h{S}_{k-1}^{*}P_{\omega_{j}}\h{S}_{k-1}\right]\right\Vert \cdot \Vert \widehat{Z}_{j}\Vert +\left\Vert \h{S}_{k-1}^{*}P_{\omega_{j}}\h{S}_{k-1}^{\perp}-\mathbb{E}\left[\h{S}_{k-1}^{*}P_{\omega_{j}}\h{S}_{k-1}^{\perp}\right]\right\Vert
\qquad \mbox{(see \eqref{eq:key})}
\\
 & \lesssim \alpha \sqrt{\frac{{\log r\cdot p \cdot \eta(\h{\SU}_{k-1})r}}{n}}\cdot \Vert \widehat{Z}_{j} \Vert   + \alpha\sqrt{{\log n\cdot p}},
 \qquad \mbox{(see \eqref{eq:large dev 1} and \eqref{eq:large dev 2})}
\end{align*}
except with a probability of at most $e^{-C \alpha}$ and under \eqref{eq:max b sigma Ais} and \eqref{eq:other m cnd}. It immediately follows that
\begin{align}
 \Vert \widehat{Z}_{j} \Vert
& \lesssim\frac{\alpha\sqrt{\frac{\log n}{ p}}}{1-\sqrt{\frac{\alpha^{2}\log r\cdot\eta(\h{\SU}_{k-1})r}{pn}}}
\qquad \mbox{(see the next line)}
\nonumber\\
& \lesssim \alpha\sqrt{\frac{\log n}{p}},\qquad\left(\mbox{if }\frac{\alpha^{2}\log r\cdot\eta(\h{\SU}_{k-1})r}{pn}\lesssim 1\right)
\end{align}
except with a probability of at most $e^{-C\alpha}$.
 In light of \eqref{eq:max b sigma Ais} and \eqref{eq:other m cnd}, we assume that $p \gtrsim \alpha^2 \log n \cdot \eta(\h{\SU}_{k-1}) r/n$.
Then using
the union bound and with the choice of $\alpha=\alpha'\log b$, it
follows that
\[
\max_{j\in [1:b]} \Vert \widehat{Z}_{j} \Vert \lesssim\alpha'\log b\sqrt{\frac{\log n }{p}},
\]
provided that $p\gtrsim \alpha'^{2}\log^{2}b\cdot\log n\cdot\eta(\h{\SU}_{k-1})r/n$
and except with a probability of at most $be^{-C\alpha'\log b}=b^{-C\alpha'}$. Invoking (\ref{eq:brk .51}), we finally conclude that
\begin{align*}
\left\Vert O(P_{\h{\SU}_{k-1}^{\perp}}S_k )\right\Vert _{F} &
\le \max_j \| \widehat{Z}_j \| \cdot
\| P_{\h{\SU}_{k-1}^\perp} P_{\SU} \|_F \|Q_k\|
\qquad \mbox{(see \eqref{eq:brk .51})}
\nonumber\\
& \lesssim \alpha'\log b\sqrt{\frac{\log n}{p}}\cdot\Vert P_{\h{\SU}_{k-1}^{\perp}}P_{\SU}\Vert _{F}\|Q_k\|.
\end{align*}
A bound in expectation also easily follows: Let $\delta$ denote the factor of $\delta'$ in last line above. Then we have that 
\begin{align}
\E \l\| P_{\SU^\perp} \cdot \l[ O\circ P_{\Omega_k} \r] (P_{\h{\SU}_{k-1}^\perp} S_k )\r \|_F & = \delta
\int_{0}^{\infty} \Pr\l[ \l\| P_{\SU^\perp} \cdot \l[ O\circ P_\Omega \r] (P_{\h{\SU}_{k-1}^\perp} S_k )\r \|_F > \alpha'  \delta \r]\, d\alpha' \nonumber\\
& \le \delta \l( 1+\int_{1}^{\infty}  
\Pr\l[ \l\| P_{\SU^\perp} \cdot \l[ O\circ P_\Omega \r] (P_{\h{\SU}_{k-1}^\perp} S_k )\r \|_F > \alpha' \delta   \r] \, d\alpha' 
 \r)
\nonumber\\
 & \le \delta \l( 1+ \int_1^{\infty} b^{-\alpha' } \, d\alpha'\r)   \nonumber\\
 & \le 2\delta \nonumber\\
 & = 2 \log b\sqrt{\frac{\log n}{p}}\cdot\Vert P_{\h{\SU}_{k-1}^{\perp}}P_{\SU}\Vert _{F}\|Q_k\|. 
\end{align}
This completes the proof of Lemma \ref{lem:O norm}.

\section{Properties of a Standard Random Gaussian Matrix}
\label{sec:cnd coh of Gaussian}

As a supplement to 
Remark 
\ref{rem:largeFirstBlock}, we show here that a standard  random Gaussian matrix $G\in\mathbb{R}^{b\times r}$ is well-conditioned and incoherent when $b$ is sufficiently large. 
 From \cite[Corollary 5.35]{vershynin2010introduction} and for fixed $\alpha \ge 1$, recall that
\begin{align}
\label{eq:norms of Gaussian}
\sqrt{b}- \Cl{Gaussian} \alpha \sqrt{r} \le \sigma_{r}(G)\le \sigma_{1}(G)\le \sqrt{b} + \Cr{Gaussian} \alpha \sqrt{r},
\end{align}
except with a probability of at most $e^{-\alpha^2 r}$. 
It follows that
\begin{align}
\nu(G) = \frac{\sigma_{1}(G)}{\sigma_{r}(G)}
\le \frac{\sqrt{b}+\Cr{Gaussian}\alpha \sqrt{r}}{\sqrt{b}-\Cr{Gaussian}\alpha \sqrt{r}},
\end{align}
which can be made close to one by choosing $b\gtrsim 
\alpha^2 r$.

For the coherence, note that $G(G^*G)^{-\frac{1}{2}}\in\mathbb{R}^{b\times r}$ is an orthonormal basis for $\mbox{span}(G)$. Using the definition of coherence in \eqref{eq:def of coh}, we then write that 
\begin{align}
\eta\l( \mbox{span}(G) \r) & = \frac{b}{r} \max_{i\in[1:b] }
\l\| G[i,:] \l(G^*G \r)^{-\frac{1}{2}}\r\|_2^2
\qquad \mbox{(see \eqref{eq:def of coh})} \nonumber\\
& \le \frac{b}{r} \max_i \l\| G[i,:]\r\|_2^2 \cdot 
\|\l(G^*G\r)^{-\frac{1}{2}} \|^2 \nonumber\\
& = \frac{b}{r} \max_i \l\| G[i,:]\r\|_2^2 \cdot \l( \sigma_{r}(G)\r)^{-2} \nonumber\\
& \le \frac{b}{r}\max_i \l\| G[i,:]\r\|_2^2 \cdot \l(\sqrt{b}-\Cr{Gaussian}\alpha \sqrt{r}\r)^{-2}
\qquad \mbox{(see \eqref{eq:norms of Gaussian})} \nonumber\\
& \le \frac{b}{r} \max_i \l\| G[i,:]\r\|_2^2 \cdot \l(\frac{b}{2}-\Cr{Gaussian}^2\alpha^2 r \r)^{-1}
\qquad \l( (a-b)^2 \ge \frac{a^2}{2}-b^2\r) \nonumber\\
& \lesssim  \frac{b}{r} \max_i \l\| G[i,:]\r\|_2^2 \cdot \l(b-\Cl{denom} \alpha^2 r \r)^{-1},
\label{eq:eta Gaussian 1}
\end{align}
except with a probability of at most $e^{-\alpha^2 r}$. For fixed $i$, $\|G[i,:]\|_2^2$ is a chi-squared random variable  with $r$ degrees of freedom so that
\begin{align}
\Pr\l[\l\| G[i,:]\r\|_2^2 \gtrsim \beta \cdot r \r] \le e^{-\beta},
\end{align}
for  $\beta\ge 1$. An application of the union bound and the choice of  
 $\beta=C\alpha \log b$ 
then leads us to
\begin{align}
\Pr\l[\max_{i\in[1:b]}\l\| G[i,:]\r\|_2^2 \gtrsim  \alpha \log b\cdot r \r] \le b \cdot b^{-C\alpha }= b^{-C\alpha}. 
\label{eq:union bnd Gaussian}
\end{align}
Substituting the bound above back into \eqref{eq:eta Gaussian 1} yields that
\begin{align}
\eta\l( \mbox{span}(G) \r)  & \lesssim \frac{b}{r}\max_i \l\| G[i,:]\r\|_2^2 \cdot \l(b-\Cr{denom}r\r)^{-1} \qquad \mbox{(see \eqref{eq:eta Gaussian 1})} \nonumber\\
& \lesssim  \frac{b}{r} \cdot  r\alpha \log b\cdot \l(b-\Cr{denom}r\r)^{-1}
\qquad \mbox{(see \eqref{eq:union bnd Gaussian})} \nonumber\\
& \lesssim \frac{\alpha b\log b}{b-\Cr{denom}\alpha^2 r},
\end{align}
except with a probability of at most $e^{-\alpha r}+b^{-C\alpha}$. In particular, when $b \ge  2  \Cr{denom} \alpha^2 r$, we find that $\eta(\mbox{span}(G))\lesssim \alpha \log b$ except with a probability of at most $e^{-\alpha r}+ b^{-C\alpha}$.

\section{Alternative Initialization \label{sec:alt init}}

$\alg$ in Algorithm \ref{alg:Alg} is initialized by truncating the SVD of the first  incomplete block $Y_1\in \mathbb{R}^{n\times b}$, where we often take $b=O(r)$ to keep the computational and storage requirements of $\alg$ minimal, see Remarks \ref{rem:complexity} and \ref{rem:storage}. Put differently, with the notation of Section \ref{sec:interp}, even though our end goal is to compute rank-$r$ truncated SVD of the full (but hidden) data block $S_1 \in\mathbb{R}^{n\times b}$, $\alg$ is initialized with  truncated SVD of the incomplete (but available) block $Y_1=P_{\Omega_1}(S_1)\in\mathbb{R}^{n\times b}$, which fills in the missing entries with zeros. Indeed,  when the first incomplete block $Y_1$ arrives, there is no prior knowledge to leverage and zero-filling the missing entries in $Y_1$ is  a sensible strategy. 
In contrast, for the rest of blocks $k\ge 2$, $\alg$ uses its previous estimate $\h{\SU}_{k-1}$ to fill out the erased entries in $Y_k$ before updating its estimate to $\h{\SU}_k$, see Algorithm~\ref{alg:Alg}.

One might  instead initialize $\alg$ with a larger block. More specifically, suppose that we change the first block  size to $b_1\ge b$ while keeping the rest of the blocks at the same size $b$. Then we set $\widehat{\SU}_1$ to be the span of  leading $r$ left singular vectors of the first incomplete block $Y_1\in\mathbb{R}^{n\times b_1}$, while the rest of steps in Algorithm \ref{alg:Alg} do not change. As the size of the first block $b_1$ increases,  $\widehat{\SU}_1$ increasingly better approximates the true subspace $\SU$. Indeed, one might consider $Y_1=P_{\Omega_1}(S_1) = S_1 + (P_{\Omega_1}(S_1)-S_1)$ as a ``noisy'' copy of $S_1$, where the noise is due to the erasures. Roughly speaking then, as $b_1$ increases, the energy of the ``signal'' part, namely $\| S_1\|_F$, grows faster than and eventually dominates the energy of the random noise $\|P_{\Omega_1}(S_1)-S_1\|_F$. This intuition is made precise by the following result which loosely speaking 
states that 
$$
d_{\GR}(\SU,\h{\SU}_1) \lesssim  \sqrt{\frac{r}{pn}},
$$
when $b_1 = \Omega(n)$. This result is proved in Appendix \ref{sec:proof of lemma init} with the aid of standard large deviation bounds.
\begin{prop}
\label{lem:init}
Consider an $r$-dimensional subspace $\SU$ with orthonormal basis $S\in\mathbb{R}^{n\times r}$.  For an  integer $b_1\ge r$, let the coefficient vectors $\{q_t\}_{t=1}^{b_1}\subset \mathbb{R}^r$ be independent copies of a random vector $q\in\mathbb{R}^r$.
For every $t\in[1:b_1]$,  we observe each coordinate of  $s_t = S q_t \in \SU$ independently  with a probability of $p$ and collect the observed entries in $y_t\in\mathbb{R}^n$, supported on the random index set $\omega_t\subseteq[1:n]$. We set $Q_1 = [q_1\,q_2\,\cdots\,q_{b_1}] \in\mathbb{R}^{r\times b_1 }$ 
 and  $Y_1 = [y_1\, y_2\,\cdots\,y_{b_1}]\in\mathbb{R}^{n\times b_1}$ for short. Let also $\widehat{\SU}_1$ be the span of  leading $r$ left singular vectors of $Y_1$.  

Then, for fixed $\alpha,\nu\ge 1$ and $1 \le \eta \le b_1/r$, it holds that
\begin{align}
d_{\GR}(\SU,\h{\SU}_1)
 \lesssim \alpha \cdot  \nu \sqrt{
  \l( 1\vee \frac{n}{b_1} \r) \frac{\l( \eta \vee \eta\l(\SU \r)\r) r\log(n\vee b_1)}{pn}},
\label{eq:init eq}
\end{align}
except with a probability of at most $e^{-\alpha}+ \Pr[ \nu(Q_1) > \nu ]+\Pr[ \eta(\Q_1)> \eta ]$. Above, $\nu(Q_1)$ is the condition number of $Q_1$,  $\eta(\Q_1)$ is the coherence of $\Q_1 = \mbox{span}(Q_1^*)$ (see (\ref{eq:def of coh})), and $a\vee b := \max\{a,b\}$. 
\end{prop}
\begin{rem}
\emph{\label{rem:largeFirstBlock}\textbf{[Discussion of Proposition \ref{lem:init}]}
As \eqref{eq:init eq} suggests, for $\widehat{\SU}_1$ to be close to $\SU_1$, the first block should be a wide matrix, namely $b_1=O(n)$. 
This dependence on the block size was anticipated. Indeed, it is
well-understood that one needs  $O(n)$ samples in order for the sample covariance matrix to closely approximate the covariance matrix of a random vector in $\R^n$  \cite{vershynin2012close}. As an example, consider the case where the coefficient vectors $\{q_t\}_{t=1}^{b_1}$ are standard random Gaussian vectors and so $Q_1\in\R^{n\times b_1}$ is a standard random Gaussian matrix, namely populated with zero-mean independent Gaussian random variables with unit variance. Then both probabilities above are small when $b_1$ is sufficiently large. More specifically, we show in Appendix \ref{sec:cnd coh of Gaussian} that 
\begin{equation}
\Pr[\nu(Q_1)>\nu ] \le \exp\l(-C\frac{\nu-1}{\nu+1} \r),
\qquad \mbox{when } b\gtrsim r,
\end{equation}
\begin{equation}
\Pr[\eta(\Q_1) > \eta ] \le \exp(-C\eta/\log b),
\qquad \mbox{when } b \log^2 b\gtrsim \eta^2 r,
\end{equation}
for a Gaussian coefficient matrix $Q_1$. 
\hfill\qedsymbol}
\end{rem}

It is also worth noting that initializing $\alg$ with a large first block can be done \emph{without} increasing the storage requirement or computational complexity of $\alg$, namely replacing the block size $b=O(r)$ in first step of Algorithm \ref{alg:Alg} with $b_1=O(n)$ can be done without losing the streaming nature of $\alg$. 
More specifically, with the alternative initialization, naively computing the truncated SVD of the first block requires $O(b_1n)=O(n^2)$ bits of storage and $O(r b_1n)=O(rn^2)$  flops. These requirements can be significantly lowered by implementing a state-of-the-art streaming PCA algorithm such as the ``power method'' in \cite{mitliagkas2013memory}. 
This suggests a two-phase algorithm. 
In the first phase, the power method is applied to the incoming data, 
where the missing entries are filled with zeros. This phase produces the estimate $\h{\SU}_1$ in $\alg$, which serves as an  initialization for the second phase in which the main loop of $\alg$ is applied to the incoming blocks, producing the estimates $\{\h{\SU}_k\}_{k\ge 2}$. If $b_1$ is sufficiently large, the first phase  brings us within the basin of attraction of the true subspace $\SU$ and activates  the locally linear convergence of $\alg$ to $\SU$ in the second phase, see Theorems \ref{thm:local cvg expectation} and \ref{thm:main result}.

\end{document}